\documentclass{article}
\usepackage{amsmath}
\usepackage{wright}
\usepackage{tikz}
\usetikzlibrary{matrix}

\newtheorem*{theorem*}{Theorem}

\newcommand{\mixed}{\textup{\textbf{Mix}}}
\newcommand{\gps}{\calA_{\mathrm{GPS}}}

\newcommand{\hayashi}{\calA_{\mathrm{Hayashi}}}

\newcommand{\purifychan}{\Phi_{\mathrm{Purify}}}
\newcommand{\swap}{\mathrm{SWAP}}

\newcommand{\reg}[1]{\mathsf{#1}}

\newcommand{\Pisym}{\Pi_{\mathrm{sym}}}

\newcommand{\spec}{\mathrm{spec}}

\newcommand{\Dtr}{\mathrm{D}_{\mathrm{tr}}}
\newcommand{\DBur}{\mathrm{D}_{\mathrm{B}}}
\newcommand{\Dchi}{\mathrm{D}_{\chi^2}}
\newcommand{\Fid}{\mathrm{F}}

\title{The Keyl--Werner algorithm is not optimal for spectrum estimation}
\author{Angelos Pelecanos\thanks{UC Berkeley. \texttt{\{apelecan,jspilecki,ewin,jswright\}@berkeley.edu}}  \and Jack Spilecki\footnotemark[1] \and Ewin Tang\footnotemark[1]\and John Wright\footnotemark[1]}
\date{}

\begin{document}

\maketitle

\begin{abstract}
    We give an algorithm which, given $n = O(d^2 \cdot (\log\log(d)/\log(d))^2)$ copies of $\rho$, estimates the eigenvalues of $\rho$ to constant error in total variation distance.
    Thus, we can learn the eigenvalues of a quantum state with fewer copies than the $\Theta(d^2)$ needed to run full state tomography.
    This is the first improvement to spectrum estimation over the influential Keyl--Werner algorithm, which uses $n = \Theta(d^2)$ copies, thereby resolving a question raised by Keyl and Werner in 2001~\cite{KW01} and refuting a 2016 conjecture of Wright~\cite{Wri16}.

    Our main technical tool is a new tomography guarantee, where the error of tomography in a particular direction $\ket{w}$ scales with $\bra{w } \rho \ket{w}$ for all directions simultaneously.
    From this stronger ``relative-error'' bound, we recover better algorithms for principal component analysis in Bures distance and tomography in $\chi^2$-divergence as corollaries.
\end{abstract}

\newpage

\hypersetup{linktocpage}
\tableofcontents
\thispagestyle{empty}

\newpage

\section{Introduction} \label{sec:intro}

In 2001, Keyl and Werner discovered an elegant quantum algorithm for estimating the spectrum $\alpha = (\alpha_1, \ldots, \alpha_d)$ of a mixed state $\rho \in \C^{d \times d}$~\cite{KW01}.
Using representation theory, their algorithm samples a natural statistic from $n$ copies of $\rho$ known as a Young diagram, which is drawn by measuring the copies in the so-called Schur basis.\footnote{
    We refer the reader to \cite{Wri16,OW17b} for an explanation of the Keyl--Werner algorithm (also known as the empirical Young diagram algorithm) and the relevant mathematical background.
    We will not need them for our results.
}
They show that this diagram, when viewed as a histogram, is a good estimator of the spectrum, converging to $\alpha$ as $n$ goes to infinity.

The Keyl--Werner algorithm ushered in the modern era of representation-theoretic algorithms for problems in quantum state learning.
The main challenge for such problems is that optimal sample complexities often require highly entangled measurements, but designing these complex measurements can be challenging.
Over the years, representation theory has emerged as the standard method for reasoning about highly entangled measurements in a principled manner,
and the Keyl--Werner algorithm was an early and especially influential example of how to apply representation theory to general qudit systems.
Beyond its historical importance, the Keyl--Werner algorithm is also a central building block in many of our optimal algorithms for full state tomography~\cite{Key06,HHJ+16,OW16,OW17a}, the even more challenging task of estimating the entire matrix $\rho$ from copies.
It has been called ``surprising''~\cite{KW01} and ``remarkable''~\cite{CM06}.
But is it optimal?

This question was initially posed by Keyl and Werner themselves, who stated:

\begin{center}
\textit{Although the estimate we discuss is asymptotically exact,
it is not at all clear \\
whether and in what sense it might be
optimal, even for finite $n$.} ---\cite{KW01}
\end{center}
\noindent
Since then, this question has continued to be asked in the literature~\cite{Wri16,AA23,PTTW25},
and determining the optimal sample complexity of spectrum estimation remains a longstanding open problem in quantum learning theory.

The Keyl--Werner algorithm has been heavily analyzed.
A string of follow-up works~\cite{HM02,CM06,OW15,OW16,OW17a} derived a sequence of increasingly refined bounds in the case of finite $n$, with the ultimate conclusion being that the output of the Keyl--Werner algorithm is $\eps$-close to $\alpha$ in total variation distance when given $n = O(d^2/\epsilon^2)$ copies of $\rho$, and that there are examples of states where this algorithm requires $n = \Omega(d^2/\epsilon^2)$ copies.
Full state tomography can be solved with $n = \Theta(d^2/\eps^2)$ copies, so the question of whether the Keyl--Werner algorithm is optimal can be rephrased as: \textit{can spectrum estimation be solved with fewer copies than full state tomography?}

Though we hope that spectrum estimation is cheaper than tomography, the spectrum still captures much of the useful information about a mixed state $\rho$.
Often we do not care about the particular eigenbasis associated to a state, so many properties of interest can be derived from the spectrum, including rank, purity, and von Neumann and R\'{e}nyi entropies.
In general, any unitarily invariant property of a state can be derived from its spectrum, making spectrum estimation a task which unifies many problems throughout learning theory.\footnote{
    We note that there is an equivalent viewpoint: properties of the spectrum of a mixed state correspond to properties of the Schmidt coefficients of a bipartite pure state.
    For example, if $\ket{\psi}$ is a pure state on two registers of dimension $d$ with Schmidt decomposition $\sum_{i=1}^d \sqrt{\alpha_i} \ket{u_i}\ket{v_i}$, then the entanglement entropy across the two registers is equal to the von Neumann entropy of the reduced density matrix on one of the registers.
    For any such property, it is known that algorithms for learning such properties can, without loss, throw away one of the registers and then estimate the corresponding property of the spectrum of the state on the remaining register~\cite{SW22}.
    So, spectrum estimation is equivalent to Schmidt coefficient estimation.
}
Indeed, historically, improved algorithms for various unitarily invariant problems, such as mixedness testing and entropy estimation, have gone hand-in-hand with advances in spectrum estimation~\cite{OW15,OW16,OW17a,BOW19,AISW20}.
This makes spectrum estimation a particularly fundamental subproblem of full state tomography.

However, it is not at all clear how to solve this subproblem any more efficiently than solving full state tomography in its entirety.
Let $\ket{v_1},\dots,\ket{v_d}$ be the eigenbasis of $\rho$, so that $\rho = \sum_{i=1}^d \alpha_i \ketbra{v_i}$.
Then each copy of $\rho$ can be viewed as providing a random eigenvector $\ketbra{v_i}$, where $i$ is sampled with probability $\alpha_i$.
If we knew the eigenvectors in advance, then we could count the number of occurrences of each eigenvector in the sample and use this count to estimate $\alpha$.
But we do not know the eigenvectors in advance, and moreover, with $n = o(d^2)$ copies the eigenvectors are fundamentally unknowable.
This is because $\Theta(d^2)$ parameters are needed to specify all $d$ eigenvectors of $\rho$, too many to estimate with $o(d^2)$ samples.
Spectrum estimation in the $n = o(d^2)$ regime therefore is the somewhat counterintuitive task of estimating the frequencies of elements in a sample when the elements themselves cannot be known.

Nevertheless, weak evidence points to the Keyl--Werner algorithm being improvable.
In their original paper, Keyl and Werner themselves observed that improvements could be made to their algorithm for very small input sizes, but these improvements did not suggest a concrete algorithm for larger numbers of copies.
More recently, the work of~\cite{PTTW25} showed a separation between spectrum estimation and full state tomography in the setting where only unentangled measurements are allowed.
However, the main question in the setting of unrestricted measurements has remained unresolved.

\ignore{In 2001, Keyl and Werner discovered a ``remarkable''~\cite{CM06} quantum algorithm for estimating the spectrum $\alpha = (\alpha_1, \ldots, \alpha_d)$ of a mixed state $\rho \in \C^{d \times d}$~\cite{KW01}.
Their algorithm established a ``surprising''~\cite{KW01} and elegant connection between the eigenvalues of mixed quantum states and the representation theoretic properties of Young diagrams.
Keyl and Werner showed that the estimate $\widehat{\balpha}$ that their algorithm produces on input $n$ copies of $\rho$ converges to $\alpha$ in the limit as the number of copies $n$ goes to $\infty$.
A string of follow-up work~\cite{HM02,CM06,OW15,OW16,OW17a} derived a sequence of increasingly refined bounds in the case of finite $n$,
with the ultimate conclusion being that $\widehat{\balpha}$ will be $\epsilon$-close to $\alpha$ in total variation distance when $n = O(d^2/\epsilon^2)$ copies,
and that there are examples of states where $n = \Omega(d^2/\epsilon^2)$ copies are necessary.
The Keyl--Werner algorithm ushered in the modern era of representation-theoretic algorithms for quantum state learning,
and it forms a central building block for several algorithms for full state tomography~\cite{Key06,HHJ+16,OW16,OW17a}.

But is it optimal?}

\subsection{Spectrum estimation beyond Keyl--Werner}

Our main result is a new algorithm for spectrum estimation which uses $n = o(d^2)$ samples and therefore improves on all previous algorithms.
As a consequence, it demonstrates that the Keyl--Werner algorithm is not optimal for spectrum estimation, resolving the question of \cite{KW01}.

\begin{theorem}
    \label{thm:main}
    There is an algorithm which, given $n$ copies of a quantum state $\rho \in \C^{d \times d}$ with spectrum $\alpha = (\alpha_1, \ldots, \alpha_d)$, outputs a vector $\widehat{\balpha}$ with the guarantee that $\dtv{\alpha}{\widehat{\balpha}} \leq \epsilon$ with probability $\geq 0.99$, so long as
    \begin{equation*}
        n = O\Big(\frac{d^2 \cdot (\log \log d)^2}{\epsilon^4 \cdot (\log d)^2}\Big).
    \end{equation*}
\end{theorem}

For $\eps = \omega(\log\log(d) / \log(d))$, this algorithm uses asymptotically fewer samples than Keyl--Werner.
Note that we do not expect Keyl--Werner to be sub-optimal across the full range of $\eps$; see the section below on sorted distribution estimation for more details.
Prior to this work, all that was known about the sample complexity of spectrum estimation was an upper bound of $O(d^2/\eps^2)$ (from analyzing Keyl--Werner) and a lower bound of $\Omega(d/\eps^2)$~\cite{OW16}.
Finding the optimal complexity of this task was a question posed by Wright~\cite[Section 10.2]{Wri16} and Anshu and Arunachalam~\cite[Question 1]{AA23}; this work marks the first progress made on this question since 2016.
For $\eps$ constant, we improve the sample complexity of spectrum estimation by \emph{two} factors of $\log(d)$.
Wright conjectured that only one such factor could be saved, and that a lower bound of $\Omega(d^2 / \log(d))$ could be shown~\cite{Wri16}.
\Cref{thm:main} refutes that conjecture.

So, how close is our algorithm to optimal?
To discuss this question, we introduce the classical analogue of spectrum estimation.

\paragraph{Sorted distribution estimation, the analogous classical problem.}
As is common in quantum learning theory, it is illuminating to consider the analogous classical problem.
When the input state has a diagonal density matrix, i.e.\ $\rho$ is classical, spectrum estimation is equivalent to the problem of estimating the sorted distribution: given samples from a discrete distribution $\alpha$ on $d$ elements, output an estimate of the $d$ associated probabilities, sorted from largest to smallest.

The classical analogue of the question we consider is whether the sorted distribution can be estimated with $o(d)$ samples,
which is provably too few samples to learn the distribution $\alpha$ itself.
Indeed, in this regime, the sample is so small that one cannot even observe most of the elements within $\alpha$'s support, and yet one is still asked to estimate the frequencies of all these elements.
As a result, this task is sometimes referred to as ``estimating the unseen''~\cite{VV11a,VV13,VV17}.
(By analogy, our quantum spectrum estimation task could be referred to as ``estimating the unknowable'', as we are estimating the frequencies of eigenvectors in a regime where the eigenvectors can't be known.)

That estimating the unseen is possible is a surprising and influential finding in the classical literature.
This was first discovered by Paninski~\cite{Pan04} in the context of estimating the entropy of a distribution.
Then, Valiant and Valiant~\cite{VV11a} generalized it to estimating support size, motivated by the ``unseen species'' problem in statistics.\footnote{
    The unseen species problem asks to estimate the support size of a distribution from samples.
    Instantiations of this question include ``How many new species will I discover on my next expedition, given that I discovered this many on my last?''~\cite{FCW43} and ``How many words did Shakespeare know, given that we only know the words which appear in his writing?''~\cite{ET76}
}
They, along with a subsequent line of research~\cite{VV13,VV17,VV11b,WY15,WY16,JVHW15}, further generalized these ideas to show that the full sorted distribution $\alpha^{\geq}$ can be estimated with $n = o(d)$ samples, culminating in the local moment matching framework of Han, Jiao, and Weissman~\cite{HJW18}.
For any parameter $\gamma >0$, their work gives an algorithm for estimating $\alpha^{\geq}$ to $\epsilon$ error using only $n = O(d/(\log(d) \cdot \epsilon^2))$ samples, so long as $\epsilon \geq 1/d^{1-\gamma}$; when $\epsilon \leq 1/d$, one can simply use the sorted empirical distribution as before to estimate $\alpha^{\geq}$ with $n = O(d/\epsilon^2)$ samples.
In addition, they proved that these upper bounds were optimal by showing matching lower bounds in these two regimes.

\paragraph{Comparison to related results.}
Let us return now to spectrum estimation.
In analogy with sorted distribution estimation, it appears that a complexity of $d^2 / \log(d)^2$ could plausibly be optimal, thereby making \cref{thm:main} optimal in $d$ up to $\log\log$ factors.
On the other hand, it seems likely that our dependence on $\eps$ is suboptimal, and that it can likely be improved, say to $\eps^2$.
That said, it very well may be that, like sorted distribution estimation, spectrum estimation is as costly as full state tomography for sufficiently small $\eps$.

The recent work of~\cite{PTTW25}
investigates spectrum estimation in the regime where the copies of $\rho$ must be measured with unentangled measurements.
In this setting, they showed that spectrum estimation can be done with $n = O(d^3 \cdot (\log\log(d)/\log(d))^4/\eps^6)$ copies, meaning that the improvement over full state tomography is at least $\log(d)^4$.
It is curious that the separation is so large, both for the unentangled and entangled settings; we do not yet have a clear explanation for why this is the case.
As in our entangled algorithm, their $\eps$ dependence is large and probably sub-optimal;
they show that there are significant obstacles to improving this dependence using current techniques,
and these obstacles apply to this work too.

Finally, they provide
computational evidence suggesting that the sample complexity of spectrum estimation, in the setting of entangled measurements, is $\Omega(d^{2-\gamma})$ for every constant $\gamma > 0$, which would mean that the $d$ dependence of our algorithm is tight up to sub-polynomial factors.
Additionally, we believe that our results are tight in $d$ up to $\log\log$ factors.

\subsection{A new, fine-grained tomography guarantee}

Our main technical contribution is a new kind of guarantee for state tomography.
As we will explain in the technical overview, the challenge of spectrum estimation is principally one of ``bucketing'', i.e.\ deciding in which directions $\rho$ is large versus small.
A natural approach for bucketing is to run tomography to get an estimate of $\rho$, and then use this estimate to create the buckets.
However, prior to this work, the guarantees on tomography algorithms were too weak to implement bucketing.
This was the primary obstacle to better spectrum estimation.

\subsubsection{The classical case}
In the classical setting, bucketing is possible because, for natural estimators, the variance in a direction scales with how large the distribution is in that direction.
In particular, suppose we have drawn $n$ samples from a discrete distribution $\alpha = (\alpha_1, \ldots, \alpha_d)$ and computed the empirical distribution $\widehat{\balpha}$.
Then for each $1 \leq i \leq d$, the random variable $\widehat{\balpha}_i$ is distributed as $(1/n)$ times a $\mathrm{Binomial}(n, \alpha_i)$ random variable, and therefore satisfies
\begin{equation*}
\Var[\widehat{\balpha}_i] = \frac{1}{n} \cdot \alpha_i (1-\alpha_i) \leq \frac{1}{n} \cdot \alpha_i.
\end{equation*}
Hence, $\stddev[\widehat{\balpha}_i] \leq \sqrt{\alpha_i/n}$, and so $\widehat{\balpha}_i = \alpha_i \pm O(\sqrt{\alpha_i/n})$ with probability 99\% due to Chebyshev's inequality.
This means that $\widehat{\balpha}_i$ will give a good estimate of $\alpha_i$ for reasonable values of $n$, and it will give tighter and tighter estimates the smaller that $\alpha_i$ is.

\paragraph{Sub-gamma concentration.}
For classical bucketing, we need to produce good estimates not just for a single $\alpha_i$, but for all $\alpha_i$'s at once.
To do so, we will need bounds on how well $\widehat{\balpha}_i$ concentrates around its mean which are stronger than what Chebyshev's inequality is able to provide us.
A good start is the multiplicative Chernoff bound, which states that
\begin{equation*}
\widehat{\balpha}_i = \alpha_i \pm t \cdot \sqrt{\alpha_i /n}
\end{equation*}
except with probability $2e^{-t^2/3}$,
for all $0 \leq t \leq \sqrt{\alpha_i n}$.
This means that the distribution of $\widehat{\balpha}_i$ decays as a Gaussian (i.e.\ it is a \emph{sub-Gaussian} random variable) with variance $\alpha_i/n$, up to roughly $\sqrt{\alpha_i n}$ standard deviations.
Beyond that, however, it decays more slowly, at a rate which is no longer sub-Gaussian.
In this regime, there are various bounds which describe its rate of decay with varying degrees of accuracy;
a reasonably tight characterization is that $\widehat{\balpha}_i$ is a $(v, c)$-sub-\emph{gamma} random variable, with parameters $v = \alpha_i/n$ and $c = 1/n$ (here, $v$ is known as the \emph{variance parameter}, matching the variance we computed above).
A sub-gamma random variable is one whose medium tail acts like a Gaussian but whose long tail acts like an exponential random variable. Sub-gamma random variables obey the concentration bound
\begin{align}\label{eq:t-squared}
        \Pr\big[ \Abs{\widehat{\balpha}_i - \alpha_i}  \geq \sqrt{2v}\cdot t + c\cdot t^2 \big] \leq 2e^{-t^2}.\footnotemark
    \end{align}\footnotetext{It is more common to state this sub-gamma concentration bound with the ``$t$'' parameter replaced by ``$\sqrt{t}$'', as we will do in the rest of the text. We have stated it in this way here in order to more easily compare it with the Chernoff bound.}%
To understand this bound, note that for small $t$, the $\sqrt{2 v}t$ term dominates, in which case it matches sub-Gaussian decay we saw from the Chernoff bound, up to constant factors.
    And for larger $t$, the $ct^2$ term dominates, which yields the exponential tail.
    The cross-over point is at $t = \sqrt{2 \alpha_i n}$ standard deviations, as expected (up to constant factors) from the Chernoff bound.\footnote{
To see why \emph{some} sort of non-Gaussian bound is necessary in the long-tail, consider the event that $\widehat{\balpha}_i = 1$, in which case $\widehat{\balpha}_i$ is $t = \sqrt{n/\alpha_i}$ standard deviations above its mean.
A sub-Gaussian tail would predict that this event occurs with probability at most $2 e^{- n /\alpha_i}$, when in fact it occurs with the much higher probability of $\alpha_i^n = e^{-n \ln(1/\alpha_i)}$; this matches the prediction of our exponential long tail, up to the $\ln(1/\alpha_i)$ in the exponent.}

\paragraph{A classical relative error bound.}
Now, to produce good estimates for all $\alpha_i$'s at once, let us simply apply the concentration bound in \Cref{eq:t-squared} for $t = O(\sqrt{\log(d)})$.
If we do so and rearrange, it states that
\begin{equation}\label{eq:classical-uniform-bound}
    \Abs{\widehat{\balpha}_i - \alpha_i} \leq C \cdot \sqrt{\frac{\log(d)}{n} \cdot \Big(\alpha_i + \frac{\log(d)}{n}\Big)},
\end{equation}
except with probability $0.01/d$. Hence, by a union bound, this bound holds for all $i \in [d]$ simultaneously, with probability at least 0.99.
As a result, for any $\alpha_i$ which is sufficiently large, namely $\alpha_i \geq \log(d)/n$, we are able to resolve its true value up to $\sqrt{\log(d)}$ standard deviations.
We refer to a bound of this form as a ``relative error bound'', as it gives us a bound on each coordinate $i \in [d]$ in terms of the corresponding $\alpha_i$.

With this bound in hand, it is usually straightforward to design a classical bucketing algorithm.
These algorithms will typically create a ``small'' bucket, consisting of those $i \in [d]$ for which $\widehat{\balpha}_i \leq \log(d)/n$,
and other buckets centered at certain probability values $p \in [0, 1]$ which contain those $\widehat{\balpha}_i$'s in the range $p \pm \sqrt{\log(d)/n} \cdot \sqrt{p}$.

\subsubsection{Our quantum guarantee}\label{sec:quantum-uniform}

In the quantum case,
we would like a tomography algorithm in which the error of the estimator $\widehat{\brho}$ scales with how large the state $\rho$ is in that direction;
in other words, we want the error $\bra{w} (\widehat{\brho} - \rho) \ket{w}$ to scale with $\bra{w} \rho \ket{w}$, for all states $\ket{w}$ simultaneously.
This can be viewed as a quantum analogue of a relative error bound.
From the literature, the strongest relative error bound for a tomography algorithm comes from the ``standard'' unentangled tomography algorithm of~\cite[Section 5.1]{Wri16} and~\cite{GKKT20}. It follows from \cite[Theorem 5]{GKKT20} that the output of this algorithm $\widehat{\brho}$ satisfies
\begin{equation}\label{eq:gkkt-infinity}
    \Vert \widehat{\brho} - \rho \Vert_{\infty} \leq O(\sqrt{d/n})
\end{equation}
with high probability.
For this algorithm, this result shows that we can bound the error $\bra{w} (\widehat{\brho} -\rho) \ket{w}$ for all states $\ket{w}$ simultaneously, but it has the downside that the bound of $O(\sqrt{d/n})$ does not improve for smaller values of $\bra{w} \rho \ket{w}$.
We note that this $\ell_{\infty}$-norm guarantee is optimal up to constant factors, as any asymptotic improvement would lead to pure state tomography algorithms with too-good-to-be-true sample complexities.

Recently, the work of~\cite{PSTW25} introduced a new mixed state tomography algorithm called $\mixed(\gps)$, which is the combination of the random purification channel of~\cite{TWZ25} and the pure state tomography algorithm of Grier, Pashayan, and Schaeffer~\cite{GPS24}.
We will describe these two ingredients and their combination in greater detail in \Cref{sec:our-algorithm} below.
For now, what is important to us is that they were able to establish an error bound for this estimator that \emph{does} scale with $\bra{w} \rho \ket{w}$, at the downside of no longer being a uniform bound.
In particular,
\cite[Theorem 1.7]{PSTW25} shows that if $\widehat{\brho}$ is the output of $\mixed(\gps)$, then
\begin{equation}\label{eq:gps-variance}
    \Var[\bra{w} \widehat{\brho} \ket{w}] \leq \frac{2}{n} \cdot \bra{w} \rho \ket{w} + \frac{d}{n^2}.
\end{equation}
Equivalently, its standard deviation is
\begin{equation}\label{eq:stddev-bound}
    \stddev[\bra{w} \widehat{\brho} \ket{w}] \leq \sqrt{\frac{2}{n} \cdot \bra{w} \rho \ket{w} + \frac{d}{n^2}},
\end{equation}
which means that with 99\% probability, $\bra{w} \widehat{\brho} \ket{w}$ will deviate from $\bra{w} \rho \ket{w}$ by at most a constant times the right-hand side of this expression.
This essentially matches the standard deviation that we saw in the classical case (with $\alpha_i$ replaced by $\bra{w} \rho \ket{w}$), except for an additional $d/n^2$ term.
This means that the error bound does scale with the value $\bra{w} \rho \ket{w}$, but only when this value is sufficiently large.

\paragraph{Sub-gamma concentration.}
To translate this guarantee into a relative error bound, we will first need to show a much stronger bound on how well $\bra{w} \widehat{\brho} \ket{w}$ concentrates around its mean than what we can derive from Chebyshev's inequality.
We show that, just like in the classical case, this random variable exhibits sub-gamma concentration.

\begin{proposition}[A sub-gamma bound for $\mixed(\gps)$] \label{prop:mixed_states_observables_subgamma-intro}
    Let $\rho \in \C^{d \times d}$ be a mixed state, and let $\widehat{\brho}$ be the output of $\mixed(\gps)$ given $n$ copies of $\rho$. For any state $\ket{w} \in \C^d$, the random variable $\bra{w} \widehat{\brho} \ket{w}$
    is $(v,c)$-sub-gamma, with
    \begin{equation*}
    v = \frac{4}{n} \cdot \bra{w} \rho \ket{w} + \frac{2d}{n^2},
    \qquad c = \frac{1}{n}.
    \end{equation*}
\end{proposition}

The variance parameter $v$ matches the variance computed in \Cref{eq:gps-variance}, up to constant factors, as we would expect.
To gain intuition for this result, note that if we apply the concentration bound from \Cref{eq:t-squared}, the cross-over point from a Gaussian tail to an exponential tail occurs at roughly the value of $t = \sqrt{n \cdot \bra{w} \rho \ket{w} + d}$.
This means that we are still in the Gaussian part of the tail when $t = \sqrt{d}$, no matter what $\bra{w} \rho \ket{w}$ is, which means that we always have the concentration
\begin{equation}\label{eq:u-gonna-union-bound?}
        \Pr\big[ \ABS{\bra{w}\widehat{\brho} \ket{w} - \bra{w}\rho \ket{w}}  \geq 2\sqrt{2vd}\big] \leq 2e^{-d}.
    \end{equation}

\paragraph{A quantum relative error bound.}
Now we would like to use our stronger concentration bound to derive a relative error bound on the error of the output of $\mixed(\gps)$.\footnote{
One could also try to derive a relative error bound for the stronger variant $\mixed^+(\gps)$, also introduced in \cite{PSTW25}. However, it is not clear to us that this would yield a \emph{better} relative error bound. The reason is that while we might expect one of the terms in $v$ to improve (in particular, we should have $2d/n^2 \to 2\E[\ell(\blambda)]/n^2$), this improvement does not propagate to an improved bound version of \Cref{eq:goal?}. In particular, since the sub-gamma concentration bound is of the form $O( \sqrt{v d} + c d )$ for failure probability $e^{-d}$, and since we expect $c = 1/n$ independent of the improvement to $v$, we still expect to end up with a ``floor'' of $d/n$ in \Cref{eq:goal?}. Since proving sub-gamma concentration bounds for $\mixed^+(\gps)$ would also require more technical proofs, we leave to future work the question of whether this paper's ideas, combined with $\mixed^+(\gps)$, yields improved results.
}
Repeating the union bound proof from the classical case, a natural strategy is to apply \Cref{eq:u-gonna-union-bound?} to every unit vector $\ket{w}$ in an $\epsilon$-net.
Since there are roughly exponential (in $d$) vectors in such a net, union bounding this guarantee over all vectors in the net would imply that
\begin{equation}\label{eq:goal?}
    \ABS{\bra{w}\widehat{\brho} \ket{w} - \bra{w}\rho \ket{w}}  \leq O(\sqrt{vd}) = O\Big(\sqrt{\frac{d}{n}\cdot\Big(\bra{w}\rho\ket{w} + \frac{d}{n}\Big)}\Big),
\end{equation}
for all vectors $\ket{w}$ in the $\epsilon$-net, at least with some probability.
Having established this guarantee for all vectors in the net, one might hope to then derive a similar statement for every unit vector $\ket{w}$ in $\C^d$.
This would indeed give us a relative error bound, as the concentration in \Cref{eq:goal?} improves the smaller that $\bra{w} \rho \ket{w}$ is.

However, this strategy does not seem to work.
The reason is that the relative error bound in \Cref{eq:goal?} asks for a different amount of concentration for each unit vector $\ket{w}$, whereas $\epsilon$-net arguments are more well-suited to establishing the same level of concentration for all vectors.
Indeed, the $\ell_{\infty}$ bound of \Cref{eq:gkkt-infinity} is established via an $\epsilon$-net argument.
In spite of this issue, we are still able to show, through different means, a relative error bound which achieves the precise scaling that this incorrect argument would have suggested.

\begin{theorem}[A relative error bound]
    \label{thm:main-tool}
    There is an algorithm which, given $n$ copies of a state $\rho \in \C^{d \times d}$, outputs a matrix $\widehat{\brho}$ with the guarantee that, with probability $\geq 0.99$, for all pure states $\ket{w}$,
    \begin{equation*}
        \ABs{\bra{w} ( \widehat{\brho} - \rho )  \ket{w} }\leq C \cdot \sqrt{ \frac{d}{n} \cdot \Big( \bra{w} \rho \ket{w} + \frac{d}{n} \Big) }.
    \end{equation*}
    Here, $C$ is a universal constant.
\end{theorem}

This precisely generalizes \Cref{eq:classical-uniform-bound}, the relative error bound from the classical setting,
to the quantum setting, at the expense of replacing the  $\log(d)$ terms in the classical expression with $d$'s in the quantum expression;
intuitively, this change of parameters is the result of the fact that there are $d$ events to union bound over in the classical case, versus roughly $e^d$ events in the quantum case.
Because $\bra{w} \rho \ket{w} \leq 1$,
the right-hand side is always at most $O(\sqrt{d/n})$ (at least, in the interesting regime of $n \geq d$),
which recovers the $\ell_{\infty}$-norm bound in \Cref{eq:gkkt-infinity} in the worst case and improves on it in other cases.
To our knowledge, we are the first to develop a tomography algorithm with this kind of guarantee.
As we will explain in more detail in \Cref{sec:new-bounds-overview} below, this guarantee is strong enough to imply the properties we want of a bucketing algorithm.
It is also strong enough to imply several other applications to quantum state tomography, discussed in the next subsection.

\subsubsection{Applications of our relative error bound}

\paragraph{Learning in Bures $\chi^2$-divergence.}

We first consider the problem of learning in the \emph{Bures $\chi^2$-divergence} distance measure.
The Bures $\chi^2$ divergence is an especially challenging distance measure which was previously studied in depth in the work of Flammia and O'Donnell~\cite{FO24}.
They gave a tomography algorithm which, when given a rank-$r$ state $\rho$ as input, outputs an estimator $\widehat{\brho}$ with $\chi^2$ error $\epsilon$, using $n = \widetilde{O}(\sqrt{r d^3}/\epsilon)$ copies, where the $\widetilde{O}(\cdot)$ notation means that there are extra logarithmic factors which are being hidden.
By applying our new relative error bound, we are able to sharpen their result by removing these logarithmic factors.

\begin{theorem}[Improved bounds for $\chi^2$ learning]
    There is an algorithm which, given $n$ copies of a mixed state $\rho \in \C^{d \times d}$ of (known) rank $r$, outputs a matrix $\widehat{\brho}$ with the guarantee that $\Dchi(\rho \| \widehat{\brho}) \leq \epsilon$ with probability $\geq 0.99$, so long as $n = O(\sqrt{rd^3}/\eps)$.
\end{theorem}

Although we do not have matching lower bounds for this problem, Flammia and O'Donnell argue in \cite[Remark 3.17]{FO24} that $O(\sqrt{rd^3}/\epsilon)$ is likely the right answer for this problem.
As a result, we believe that our algorithm likely achieves optimal sample complexity for this problem. We refer the reader to~\cite{FO24} for a more in-depth discussion of the significance of learning in this distance measure.

\paragraph{Principal component analysis.}
The final application we consider is to a generalization of tomography known as \emph{principal component analysis}, in which our goal is to output the low-rank matrix that best approximates the state $\rho$.
This is motivated by the case when $\rho$ is itself a low-rank state with a small amount of noise added to it, and we would like to recover the original low-rank part of $\rho$ with as few copies as possible.
As an example, if $\rho_{\leq k}$ is the restriction of $\rho$ to its top-$k$ eigenvectors, then $\rho_{\leq k}$ is the closest rank-$k$ matrix to $\rho$, and it satisfies
\begin{equation*}
    \Dtr(\rho, \rho_{\leq k}) = \frac{1}{2} ( \alpha_{k+1} + \dots + \alpha_d).
\end{equation*}
In this case, we might hope to find a rank-$k$ matrix which approximates $\rho$ almost as well as $\rho_{\leq k}$ does.
O'Donnell and Wright~\cite{OW16} showed that this is possible with the same number of copies as the easier task of rank-$k$ trace distance tomography.
They proved the following.

\begin{theorem}[Trace distance PCA]
    There is an algorithm which, given $n$ copies of a mixed state $\rho \in \C^{d \times d}$, outputs a rank-$k$ matrix $\widehat{\brho}_{\leq k}$ with the guarantee that, with probability at least 0.99,
    \begin{equation*}
        \Dtr(\rho, \widehat{\brho}_{\leq k}) \leq  \frac{1}{2} ( \alpha_{k+1} + \dots + \alpha_d) + \epsilon,
    \end{equation*}
    so long as $n = O(kd/\eps^2)$.
\end{theorem}

As an application of our relative error tomography bound, we give a new proof of this statement.
In addition, we are able to give a similar statement for PCA in the related \emph{Bures distance}.

\begin{theorem}[Bures distance PCA]
    There is an algorithm which, given $n$ copies of a mixed state $\rho \in \C^{d \times d}$, outputs a rank-$k$ matrix $\widehat{\brho}_{\leq k}$ with the guarantee that, with probability at least 0.99,
    \begin{equation*}
        \DBur(\rho, \widehat{\brho}_{\leq k}) \leq \sqrt{\alpha_{k+1} + \cdots + \alpha_d} + \epsilon,
    \end{equation*}
    so long as $n = O(kd/\eps^2)$.
\end{theorem}
\noindent
It is known that $n = \Theta(kd/\epsilon^2)$ copies are necessary and sufficient to perform tomography of rank-$k$ states in Bures distance due to~\cite{Yue23,PSW25,SSW25}, and so this recovers that bound and extends it to the harder task of principal component analysis.
The only previous result of this type that we are aware of is from the work of O'Donnell and Wright, where they proved a PCA-type result for the related \emph{quantum Hellinger distance}~\cite[Theorem 1.19]{OW17a}.
Their result has some additional logarithmic factors that make it unlikely to be optimal,
and so modulo the fact that quantum Hellinger distance and Bures distance are not precisely equivalent, our bound essentially improves on theirs.
We note that an even stronger Bures distance-type PCA bound was conjectured in~\cite{PTTW25}, but this stronger bound is unachievable.
\Cref{subsection_PCA_fidelity} discusses these subtleties in greater detail.

\subsection{Open problems}

\paragraph{Optimal spectrum estimation.}
Our work shows that the Keyl--Werner algorithm is not optimal,
and that spectrum estimation is possible in the regime of $n = o(d^2)$ samples.
But it does not settle the optimal complexity of spectrum estimation.
We believe that the bound we achieve is optimal in terms of the dimension $d$, up to the $\log \log d$ factors, but we also believe its $\epsilon^{-4}$ dependence on the error parameter is suboptimal and should be $\epsilon^{-2}$ instead. Thus, we conjecture that the correct bound for spectrum estimation is
\begin{equation*}
    n = \Theta\Big(\frac{d^2}{\log^2(d) \cdot \epsilon^2}\Big),
\end{equation*}
up to possible additional $\log\log(d)$ factors that we are omitting for simplicity,
and in the regime of ``sufficiently large'' $\epsilon$.
(We conjecture this for ``sufficiently large'' $\epsilon$ because, as we have seen above, classical local moment matching only improves upon the naive bound of $d/\epsilon^2$ when $\epsilon \gg 1/d$.)
Our techniques do not seem capable of achieving an $\epsilon^{-2}$ dependence,
and we believe that any algorithm capable of achieving this bound will require new and interesting ideas.
In addition, we still do not have lower bounds which come anywhere close to our upper bounds: the best known lower bound for spectrum estimation remains the $n = \Omega(d/\epsilon^2)$ lower bound which follows from the mixedness testing lower bound of O'Donnell and Wright~\cite{OW15}. (The proof of this lower bound was recently simplified~\cite{OW25}.)
Any improvement on this lower bound would be a long-sought-after breakthrough.

\ignore{
\paragraph{Improved relative error bounds.}
Next, we believe that our improved tomography bound from \Cref{thm:main-tool} is likely to have further applications.
We also believe that it is likely improvable, at least in certain regimes.
To understand why, recall that the tomography algorithm that we use in this work applies the random purification channel to $\rho^{\otimes n}$ and performs Hayashi's algorithm on the output.
But~\cite{PSTW25} showed that in certain situations, there is a way of improving upon this algorithm by replacing the random purification channel by a more fine-grained technique that they call \emph{quasi-purification}.
To illustrate the types of improvements that we might see, recall from \Cref{eq:stddev-bound} above that if $\widehat{\brho}$ is the output of the tomography algorithm that we study, and if $\ket{w}$ is a unit vector in $\C^d$, then
\begin{equation}\label{eq:stddev-restated}
    \stddev[\bra{w} \widehat{\brho} \ket{w}] \leq \sqrt{\frac{1}{n} \cdot \bra{w} \rho \ket{w} + \frac{d}{n^2}}.
\end{equation}
If instead $\widehat{\brho}$ is the output of the algorithm from~\cite{PSTW25} which uses quasi-purification, then it achieves an improved standard deviation
\begin{equation*}
    \stddev[\bra{w} \widehat{\brho} \ket{w}] \leq \sqrt{\frac{1}{n} \cdot \bra{w} \rho \ket{w} + \frac{\E[\ell(\blambda)]}{n^2}}.
\end{equation*}
The $\E[\ell(\blambda)]$ is a particular representation theoretic quantity which has the property that it is always bounded by $\min\{r, 2 \sqrt{n}\}$, where $r$ is the rank of $\rho$.
If $\rho$ is rank-$d$, then this can give you savings over \Cref{eq:stddev-restated} in the regime when $n = o(d^2)$, which happens to be exactly the regime that this paper operates in; these savings were necessary for several of the downstream tomography applications in~\cite{PSTW25}.
We believe it might be possible to derive sub-gamma concentration bounds for this quasi-purification-based tomography algorithm, as we have done for the purification-based algorithm we study in this work, and potentially these bounds could translate to a stronger version of our \Cref{thm:main-tool}.
It would be interesting to see if there are applications of such an improved bound. (For example, it is unclear to us whether this would yield immediate improvements to our spectrum estimation bound.)

\paragraph{Low-rank principal component analysis.}
Given a mixed state $\rho$, the best rank-$r$ approximation to $\rho$ is $\rho_{\leq r}$, in which we zero out all but the top $r$ of $\rho$'s eigenvalues.
This has a trace distance error to $\rho$ given by
\begin{equation*}
    \mathrm{D}_{\mathrm{tr}}(\rho, \rho_{\leq r}) = \frac{1}{2} \cdot (\alpha_{r+1} + \cdots + \alpha_d).
\end{equation*}
The work of \cite{OW16} gave a tomography algorithm which satisfies a \emph{rank-$r$ principal component analysis (PCA)} guarantee; namely, with $n = O(dr/\epsilon^2)$ samples, it outputs a state $\widehat{\brho}_{\leq r}$ which is almost as good of a rank-$r$ approximation as $\rho_{\leq r}$, in the sense that
\begin{equation*}
    \mathrm{D}_{\mathrm{tr}}(\rho, \widehat{\brho}_{\leq r}) \leq \frac{1}{2} \cdot (\alpha_{r+1} + \cdots + \alpha_d) + \epsilon.
\end{equation*}
This recovers and generalizes the $n = O(rd/\epsilon^2)$ bound needed to perform rank-$r$ tomography.

Open since then has been the question of whether rank-$r$ PCA could be achievable in other distance metrics, such as Bures distance.
A variant of Bures distance PCA appears in~\cite[Theorem 1.19]{OW17a}, but with apparently suboptimal parameters. To conjecture the right parameters, let us note that the Bures distance error of $\rho_{\leq r}$ to $\rho$ is
\begin{equation*}
    \sqrt{1 + \tr(\rho_{\leq r}) - 2\cdot F(\rho, \rho_{\leq r})} = \sqrt{\alpha_{r+1} + \cdots + \alpha_d}.
\end{equation*}
The left-hand side here is the correct generalization of the Bures distance to the case that one of the two states is subnormalized. Then we conjecture that there is a tomography algorithm which uses $n = O(rd/\epsilon^2)$ copies and outputs an estimate $\widehat{\brho}_{\leq r}$ satisfying
\begin{equation*}
    \sqrt{1 + \tr(\widehat{\brho}_{\leq r}) - 2\cdot F(\rho, \widehat{\brho}_{\leq r})} \leq \sqrt{\alpha_{r+1} + \cdots + \alpha_d} + \epsilon.
\end{equation*}
We leave this as a problem for future work, and we believe that this is an example of a problem that an improved relative error bound using quasi-purification might help on.
}

\section{Technical overview} \label{sec:tech}

\subsection{Local moment matching}

Our algorithm is motivated by a technique from the classical learning theory literature known as \emph{local moment matching}.
We first summarize this technique;
for a much more thorough treatment of local moment matching,
we refer the reader to \cite[Section 2]{PTTW25}.

Suppose we are given $n$ samples $\bx_1, \ldots, \bx_n \in [d]$ drawn independently from a distribution $\alpha = (\alpha_1, \ldots, \alpha_d)$.
The natural classical analogue of full state tomography is estimating the distribution $\alpha$;
the natural classical analogue of spectrum estimation is estimating the sorted distribution $\alpha^{\geq} \coloneq \mathrm{sort}(\alpha)$, where $\mathrm{sort}(\cdot)$ sorts its input from largest to smallest.
The most basic algorithm for estimating $\alpha$ is to compute the histogram $\bh = (\bh_1, \ldots, \bh_d)$ of the sample,
where $\bh_i$ is the number of $i$'s in $\bx$,
and output the \emph{empirical distribution} $\widehat{\balpha} = \bh/n$.
It is known that $\widehat{\balpha}$ is $\epsilon$-close to $\alpha$ with high probability when $n = \Theta(d/\epsilon^2)$, and when this happens, $\widehat{\balpha}^{\geq}$ is also $\epsilon$-close to $\alpha^{\geq}$.
For this problem, though, better estimators can be constructed with the technique of local moment matching.
These estimators use only $n = O(d/(\log(d) \cdot \epsilon^2))$ samples provided $\eps$ is sufficiently large~\cite{HJW18},
and this is optimal.

Before explaining \emph{local} moment matching, we explain the simpler (and less effective) strategy of moment matching.
In moment matching, we estimate the moments of $\alpha$, given by
\begin{equation*}
    p_k(\alpha) = \sum_{i=1}^d \alpha_i^k,
\end{equation*}
where $p_k(\cdot)$ is the $k$-th \emph{power sum symmetric polynomial}.
Then we ``match'' the moments by computing a sorted distribution $\widehat{\balpha}^{\geq}$ whose moments match our estimates.
It is known that the first $d$ moments $p_1(\alpha), \ldots, p_d(\alpha)$ are enough to uniquely specify $\alpha^{\geq}$,
and the hope is that if we have good estimates of the moments, then $\widehat{\balpha}^{\geq}$ will be close to $\alpha^{\geq}$.

Given a sample $\bx = (\bx_1, \ldots, \bx_n) \sim \alpha^{\otimes n}$, there is a canonical estimator $c_k(\bx)$ for estimating the $k$-th moment $p_k(\alpha)$. This estimator has the property that it is \emph{unbiased}, meaning that $\E_{\bx}[c_k(\bx)] = p_k(\alpha)$,
and among all unbiased estimators for $p_k(\alpha)$ it is the one with the minimum variance.
When all of the probability values of $\alpha$ are small (say, $O(1/d)$), then $c_k(\bx)$ will have relatively small variance, and the estimate of $p_k(\alpha)$ that it gives us will be so good that our ``vanilla'' moment matching strategy will work and allow us to recover $\alpha^{\geq}$ approximately.
But when $\alpha$ has some large probability values, then the variance of $c_k(\bx)$ will be so large that it will wash out any contribution from the small probability values.
In this case, moment matching will still allow us to learn the large values of $\alpha^{\geq}$, but it will not allow us to resolve the small values.
This is an issue, because $\alpha^{\geq}$ could have a constant fraction of its mass on elements of weight $O(1/d)$, and if we can't learn these then we can't learn $\alpha^{\geq}$.

Local moment matching aims to solve this issue by first performing a coarse ``bucketing'' of the probability values $\alpha_i$ to learn which are large and which are small,
and then it performs moment matching separately within each bucket.
To do so, it splits the sample into two equal halves of size $n = o(d)$ each.
Writing $\bx = (\bx_1, \ldots, \bx_n)$ for the first half of the sample, it will compute the empirical distribution $\widehat{\balpha}$ corresponding to $\bx$;
any coordinate $i \in [d]$ for which $\widehat{\balpha}_i$ is greater than some bound $B > 0$ will be thrown into the $\mathsf{Large}$ bucket, and the remaining coordinates will be thrown into the $\mathsf{Small}$ bucket.
Then, writing $\by = (\by_1, \ldots, \by_n)$ for the second half of the sample, it will use $\by$ to perform local moment matching separately with $\mathsf{Large}$ and $\mathsf{Small}$.
For some properties of $\alpha^{\geq}$, such as the Shannon entropy, this two-bucket approach suffices to achieve optimal sample complexity.
Learning $\alpha^{\geq}$ itself requires a more fine-grained approach involving more than two buckets if we wish to achieve optimal sample complexity; however, two buckets are sufficient if we at least wish to achieve optimal scaling in the parameter $d$ (and are therefore sufficient for achieving $n = o(d)$ scaling).

\subsection{Local moment matching in the quantum setting}
\label{sec:lmm-quantum}

Recently, Pelecanos, Tan, Tang, and Wright~\cite{PTTW25} suggested a general framework for applying local moment matching to the quantum setting.
We summarize their framework here; see \cite[Section 3]{PTTW25} for a thorough treatment.

Let $\rho \in \C^{d \times d}$ be the mixed state whose spectrum $\alpha = (\alpha_1, \ldots, \alpha_d)$ we want to estimate.
As in the classical setting, it is natural to want to learn $\alpha$ by estimating the moments $p_k(\alpha)$ and then performing moment matching.
In the quantum setting, there are also canonical, unbiased, minimum-variance estimators for the moments; we describe these estimators below.
However, just as in the classical setting, the existence of large eigenvalues will blow up the variance of these estimators, making them useless for estimating the small eigenvalues of $\rho$.

To overcome this, \cite{PTTW25} suggested the following quantum analogue of the ``bucketing'' step from local moment matching:
Divide the copies into two sets of size $n = o(d^2)$,
and run a full state tomography algorithm on the first set of $n$ copies to produce a coarse estimate $\widehat{\brho}$ of $\rho$.
Let $\bPi$ be the projector onto the eigenvectors of $\widehat{\brho}$ with eigenvalues at least $B$ for some bound $B > 0$, and let $\overline{\bPi}$ be the projector onto the complement of $\bPi$.
Ideally, if $\widehat{\brho}$ is a good enough estimator of $\rho$, then $\bPi$ should be close to a projector onto $\rho$'s large eigenvalues (namely, eigenvalues which are $\geq B/2$), and $\overline{\bPi}$ should be close to a projector onto $\rho$'s small eigenvalues (namely, eigenvalues which are $\leq 2B$).
If this is the case, then we can take each of our remaining $n$ copies of $\rho$ and measure them with the projective measurement $\{\bPi, \overline{\bPi}\}$;
this will convert each copy of $\rho$ into a copy of $\bPi \rho \bPi + \overline{\bPi} \rho \overline{\bPi}$,
which only has large eigenvalues in the $\bPi$ subspace and small eigenvalues in the $\overline{\bPi}$ subspace.
We can then perform moment matching separately within the $\bPi$ and $\overline{\bPi}$ subspaces.

For this to work, let us see what sort of guarantees we will need from our bucketing step.
Ideally, $\bPi \rho \bPi$ will only have large eigenvalues and $\overline{\bPi} \rho \overline{\bPi}$ will only have small eigenvalues; if this does not occur, then we call it a \emph{misclassification error}.
In addition, we want the spectrum of $\bPi \rho \bPi + \overline{\bPi} \rho \overline{\bPi}$ to be close to $\rho$, as otherwise the spectrum we are learning in our moment matching step is different than $\alpha$; the extent to which this does not occur we call \emph{alignment error}.
Misclassification errors can also occur in the classical setting, but alignment error is unique to the quantum setting because measuring a quantum state disturbs it.

We focus on achieving low misclassification error, and in particular showing that $\overline{\bPi} \rho \overline{\bPi}$ only has small eigenvalues. (Ensuring that $\bPi \rho \bPi$ only has large eigenvalues is less important because $\bPi$ will be a small-dimensional subspace, so we can learn the eigenvalues of $\bPi \rho \bPi$ without moment matching, regardless of their magnitude.)
This guarantee is an $\ell_{\infty}$ norm guarantee on $\overline{\bPi} \rho \overline{\bPi}$: it is equivalent to the statement that
\begin{equation*}
    \Vert \overline{\bPi} \rho \overline{\bPi} \Vert_{\infty} \leq 2B.
\end{equation*}
What \cite{PTTW25} observed is that we can show low misclassification error statements of this form
if the tomography algorithm that we use itself has low $\ell_{\infty}$ norm error.
We illustrate this with their chosen tomography algorithm, the ``standard'' unentangled tomography algorithm due to~\cite[Section 5.1]{Wri16} and~\cite{GKKT20}. As mentioned in \Cref{sec:quantum-uniform}, \cite[Theorem 5]{GKKT20} implies that the output of this algorithm $\widehat{\brho}$ satisfies
\begin{equation}\label{eq:gkkt-infinity-2}
    \Vert \widehat{\brho} - \rho \Vert_{\infty} \leq O(\sqrt{d/n})
\end{equation}
with high probability.
If this holds, then we have
\begin{equation*}
    \Vert \overline{\bPi} \rho \overline{\bPi} \Vert_{\infty}
    \leq \Vert \overline{\bPi} \widehat{\brho} \overline{\bPi} \Vert_{\infty}
    + \Vert \overline{\bPi} (\rho - \widehat{\brho}) \overline{\bPi} \Vert_{\infty}
    \leq B + \Vert \rho - \widehat{\brho} \Vert_{\infty}
    \leq B + O(\sqrt{d/n}),
\end{equation*}
where the first step is the triangle inequality and the second step is because $\overline{\bPi}$ only projects onto those eigenvalues of $\widehat{\brho}$ which are less than $B$. Since we want the right-hand side to be at most $2B$ for low misclassification error, it suffices to take $n = d/B^2$.
They eventually set $B = (1/d) \cdot (\log(d) / \log \log(d))^2$ (at least for constant $\epsilon$),
which gives their sample complexity of $n = d^3 \cdot (\log \log(d)/\log(d))^4$ (at least for constant~
$\epsilon$).

The work \cite{PTTW25} focused on the setting of unentangled measurements because we have strong $\ell_{\infty}$ norm error bounds on the standard unentangled tomography algorithm.
They observed, however, that this bucketing step presents a significant barrier for achieving spectrum estimation with entangled measurements, which we need to use if we want an algorithm which uses $n = o(d^2)$ copies.
This is because although we have several entangled tomography algorithms~\cite{Key06,HHJ+16,OW16,PSW25} with strong guarantees,
prior to this work we did not have any strong $\ell_{\infty}$-type guarantees for these algorithms.
In general, proving $\ell_{\infty}$ norm bounds can often be more difficult than proving bounds in other norms, and doing so for any of these algorithms appears to be well beyond our current set of techniques.

In addition, it was perhaps difficult to even guess what sort of ``$\ell_{\infty}$-type'' guarantee we would want our entangled tomography algorithm to satisfy;
as it turns out, the $\ell_{\infty}$ bound in \Cref{eq:gkkt-infinity-2} is optimal among \emph{all} algorithms (both unentangled and entangled),
so we cannot hope for an entangled tomography algorithm with a better $\ell_{\infty}$ norm error.
The reason for this is that the $\ell_{\infty}$ bound in \Cref{eq:gkkt-infinity-2} is strong enough to imply a learning algorithm for pure quantum states in $\C^d$ which uses only $n = O(d/\epsilon^2)$ copies. This is optimal, due to the lower bounds of~\cite{SSW25}.
Thus, any $\ell_{\infty}$ bound which improves on \Cref{eq:gkkt-infinity-2} would imply a pure state tomography algorithm which contradicts known lower bounds.

\subsection{New full state tomography bounds}\label{sec:new-bounds-overview}

This situation changed with the recent work of Pelecanos, Spilecki, Tang, and Wright~\cite{PSTW25}.
Their main result is that mixed state tomography reduces to pure state tomography, which is the special case of mixed state tomography when the input state has rank one.
Using this reduction, they designed new and conceptually clean mixed state tomography algorithms,
and they used these new algorithms to give new tomography bounds
and substantially simpler proofs of previously known tomography bounds.
The primary technical contribution of this paper is to extend these techniques to derive a new $\ell_{\infty}$-type guarantee for mixed state tomography which we are able to use to perform bucketing.

This guarantee is stated in our \Cref{thm:main-tool} above.
We show a tomography algorithm which, when run on $n$ copies of $\rho$, outputs an estimate $\widehat{\brho}$ satisfying the following: with high probability, for all pure states $\ket{w} \in \C^d$,
\begin{equation}\label{eq:our-guarantee}
        \ABs{\bra{w} ( \widehat{\brho} - \rho )  \ket{w}} \leq C \cdot \sqrt{ \frac{d}{n} \cdot \Big( \bra{w} \rho \ket{w} + \frac{d}{n} \Big) }.
\end{equation}
Because $\bra{w} \rho \ket{w} \leq 1$,
the right-hand side is always at most $O(\sqrt{d/n})$ (at least, in the interesting regime of $n \geq d$),
which recovers the bound in \Cref{eq:gkkt-infinity-2} and matches it in the worst case.
But our bound gives even stronger concentration for smaller values of $\bra{w} \rho \ket{w}$, bottoming out at an error of $O(d/n)$ whenever $\bra{w} \rho \ket{w} \leq d/n$.
To see how this guarantee can be useful,
suppose we used this $\widehat{\brho}$ to compute a bucketing $\{\bPi, \overline{\bPi}\}$.
Then for simplicity we would like to show that $\Vert \overline{\bPi} \rho \overline{\bPi} \Vert_{\infty} \leq 2B$,
though we will show a slightly weaker bound of $O(B)$ instead.
Let $\ket{w}$ be any vector contained in $\overline{\bPi}$. Then
\begin{align*}
    \bra{w} \rho  \ket{w}
    &= \bra{w}  \widehat{\brho} \ket{w}
    +  \bra{w} (\rho - \widehat{\brho}) \ket{w} \\
    & \leq B
    + C \cdot \sqrt{ \frac{d}{n} \cdot \Big( \bra{w} \rho \ket{w} + \frac{d}{n} \Big) } \tag{by $\ket{w} \in \overline{\bPi}$ and \Cref{eq:our-guarantee}} \\
    & \leq B
    + \frac{1}{2} \cdot \Big(C^2 \cdot \frac{d}{n} + \bra{w} \rho \ket{w} + \frac{d}{n}\Big),
\end{align*}
where the last step used $a \cdot b \leq \frac{1}{2}(a^2 + b^2)$. Rearranging gives us
\begin{equation*}
    \bra{w} \rho \ket{w} \leq 2 B + (C^2 +1) \cdot \frac{d}{n}.
\end{equation*}
This is $O(B)$ so long as $n = \Omega(d/B)$,
which is an improvement over the unentangled case by a factor of $B$.
As in~\cite{PTTW25}, we will eventually set $B = (1/d) \cdot (\log(d) / \log \log(d))^2$ (at least for constant $\epsilon$),
which gives our sample complexity of $n = d^2 \cdot (\log \log(d)/\log(d))^2$ (at least for constant $\epsilon$).

Having motivated our tomography guarantee, we next discuss our choice of tomography algorithm and the proof of our guarantee.

\subsubsection{Our algorithm}\label{sec:our-algorithm}
We begin with the reduction from~\cite{PSTW25}.
It relies on a recent construction due to Tang, Wright, and Zhandry~\cite{TWZ25} known as the \emph{random purification channel}.
This is a quantum channel $\purifychan$
    which converts copies of a mixed state $\rho \in \C^{d \times d}$ into copies of a \emph{purification} $\ket{\rho} \in \C^d \otimes \C^d$ of $\rho$, which is a pure state satisfying $\tr_2(\ketbra{\rho}) = \rho$.
More precisely, the random purification channel acts as follows:
\begin{equation*}
    \purifychan(\rho^{\otimes n}) = \E_{\ket{\brho}} \ketbra{\brho}^{\otimes n},
\end{equation*}
where $\ket{\brho}$ is distributed as a uniformly random purification of $\rho$.
Given this, the reduction works as follows: take $\rho^{\otimes n}$ and apply the random purification channel to it, resulting in $\ket{\brho}^{\otimes n}$.
Now run a pure state tomography algorithm on this state, and let $\ket{\bu} \in \C^d \otimes \C^d$ be the result.
Output $\tr_2(\ketbra{\bu})$ as your estimate for $\rho$.

How well this algorithm performs depends on your choice of pure state tomography algorithm.
The most obvious choice is the pure state tomography algorithm of Hayashi~\cite{Hay98}, which is, in a certain precise sense, the literal optimal pure state tomography algorithm that one can perform.
However, Hayashi's algorithm is not an unbiased estimator for the state.
Indeed, if we run Hayashi's algorithm on the state $\ket{\brho}^{\otimes n}$ and it outputs the vector $\ket{\bu}$, then the expectation $\E_{\ket{\bu}}[\ketbra{\bu}]$ is not equal to $\ketbra{\brho}$: the former is a mixed state of rank $> 1$, whereas the latter is pure.
Grier, Pashayan, and Schaeffer~\cite{GPS24} observed that a modification of Hayashi's algorithm makes it unbiased: set
\begin{equation*}
    \widehat{\bsigma} \coloneqq \frac{d^2+n}{n} \cdot \ketbra{\bu} - \frac{1}{n} \cdot I_{d^2},
\end{equation*}
then $\E_{\widehat{\bsigma}}[\widehat{\bsigma}] = \ketbra{\rho}$, and so $\widehat{\bsigma}$ is an unbiased estimator. We use this as our pure state tomography algorithm,
so that the final mixed state tomography algorithm that we analyze is the one that outputs $\widehat{\brho} = \tr_2(\widehat{\bsigma})$.
Because $\widehat{\bsigma}$ is unbiased,
\begin{equation*}
    \E[\widehat{\brho}] = \E \tr_2(\widehat{\bsigma}) = \E \tr_2(\ketbra{\brho}) = \rho,
\end{equation*}
using the fact that $\ketbra{\brho}$ is a purification of $\rho$. Hence, this $\widehat{\brho}$ is also an unbiased estimator of $\rho$.
This algorithm is convenient to analyze because it translates desired mixed state tomography guarantees into more tractable guarantees about the pure state tomography algorithm of Grier, Pashayan, and Schaeffer.

\subsubsection{Proving our guarantee}

The starting point of our proof of \Cref{eq:our-guarantee} is \Cref{prop:mixed_states_observables_subgamma-intro}, which states that for each state $\ket{w} \in \C^d$,
the random variable $\bra{w} \widehat{\brho} \ket{w}$ is a $(v, c)$-sub-\emph{gamma} random variable, with parameters
\begin{equation}\label{eq:first-sub-gamma}
    v = \frac{2d}{n^2} + \frac{4}{n} \cdot \bra{w} \rho \ket{w},
    \quad
    c = \frac{1}{n}.
\end{equation}
A sub-gamma random variable is defined by having a bounded moment generating function, and so we prove this statement by a delicate analysis of the moment generating function of $\bra{w} \widehat{\brho} \ket{w}$.
As argued in \Cref{sec:quantum-uniform}, there is now a natural approach for using this sub-gamma concentration to prove \Cref{eq:our-guarantee} that ultimately does not work.
Repeating this argument, a natural strategy is to apply the concentration bound for sub-gamma random variables to every unit vector $\ket{w}$ in an $\epsilon$-net, and then use this to derive a similar statement for every unit vector $\ket{w}$ in $\C^d$.
However, this strategy does not seem to work, because we are looking for a different amount of concentration for each unit vector $\ket{w}$, whereas $\epsilon$-net arguments are more well-suited to establishing the same level concentration for all vectors, i.e.\ $\ell_{\infty}$ norm bounds.

Handling this obstacle requires us to change our proof strategy somewhat substantially.
To begin, we show that our desired guarantee from \Cref{eq:our-guarantee} does actually follow from an $\ell_{\infty}$ norm bound, but on a different matrix.
To define this matrix, let us assume that the state $\rho$ is diagonal in the standard basis, so that $\rho = \sum_{i=1}^d \alpha_i \cdot \ketbra{i}$.
Then we show that it suffices to prove an operator norm bound of $O(1)$ on the matrix
\begin{equation*}
    \calL_R^{-1}(\widehat{\brho} - \rho) \coloneqq \sum_{i,j=1}^d \frac{2}{R_i + R_j} \cdot (\widehat{\brho}_{i j} - \rho_{i j}) \cdot \ketbra{i}{j},
    \quad
    \text{where }R_i = \sqrt{\frac{d}{n} \cdot \Big(\alpha_i + \frac{d}{n}\Big)}.\footnote{In the main text, we actually consider a rescaled version of this operator, and set $R_i = \sqrt{\frac{n}{d} \cdot (\alpha_i + \frac{d}{n})}$ instead. This is for technical reasons: for our analysis, we want the observable $O_w$, soon to be described, to satisfy $\norm{O_w}_\infty \leq 1$.}
\end{equation*}
Why exactly this matrix is called ``$\calL_R^{-1}(\widehat{\brho} - \rho)$'' we defer to later.
What is important is that this matrix is equal to $\widehat{\brho} -\rho$, except with the $(i, j)$-th entry scaled up for smaller values of $\alpha_i$ and $\alpha_j$.
Intuitively, this rescaling will force $\widehat{\brho}_{ij}$ and $\rho_{i j}$ to be closer for smaller values of $\alpha_i$ and $\alpha_j$, and this is exactly the sort of statement we want to prove in \Cref{eq:our-guarantee}.

To show that $\Vert \calL_R^{-1}(\widehat{\brho} - \rho) \Vert_{\infty} \leq O(1)$, it \emph{does} suffice to show $\Abs{\bra{w} \calL_R^{-1}(\widehat{\brho} - \rho) \ket{w}} \leq O(1)$ for all unit vectors $\ket{w}$ in an $\epsilon$-net.
For a fixed $w$, we can write this quantity as
\begin{equation*}
    \bra{w} \calL_R^{-1}(\widehat{\brho} - \rho) \ket{w}
    = \sum_{i, j=1}^d \frac{2}{R_i + R_j} \cdot (\widehat{\brho}_{i j} - \rho_{ij}) \cdot w_i^\dagger w_j
    = \tr(O_w \cdot (\widehat{\brho} - \rho)),
\end{equation*}
where $O_w$ is the Hermitian observable given by
\begin{equation*}
    O_w = \sum_{i, j=1}^d \frac{2}{R_i + R_j} \cdot w_i^{\dagger} w_j \cdot \ketbra{j}{i}.
\end{equation*}
Our goal is to show that $\tr(O_w \cdot (\widehat{\brho} - \rho))$ concentrates extremely well about its mean. This follows from a much more general concentration statement about observables. In particular, we show that for any observable $O \in \C^{d \times d}$ with $\norm{O}_\infty \leq 1$, the random variable $\tr(O \cdot \widehat{\brho})$ is $(v, c)$-sub-gamma, with parameters
\begin{equation*}
    v = \frac{4d \cdot \tr(O^2) + 8n\cdot \tr(O^2 \cdot \rho)}{n^2},
    \qquad c = \frac{2}{n}.
\end{equation*}
This generalizes the statement from \Cref{eq:first-sub-gamma} above (up to constant factors) by taking $O = \ketbra{w}$, and proving this statement is the main technical work in this paper.

\subsection{Entangled moment estimation}
\label{sec:tech-overview-moment-estimation}

Now we discuss our moment estimation algorithm.
The family of moment estimators we use was introduced by O'Donnell and Wright~\cite{OW15} and analyzed by Acharya, Issa, Shende, and Wagner~\cite{AISW20}. The particular perspective that we take on these estimators is due to B\u{a}descu, O'Donnell, and Wright~\cite{BOW19}.

These estimators are most conveniently stated in terms of observables.
Recall that an observable $O \in \C^{d \times d}$ is a Hermitian matrix, which implies that it has an eigendecomposition $O = \sum_{i=1}^d o_i \cdot \ketbra{u_i}$ with real-valued eigenvalues $o_1, \ldots, o_d$.
Given a quantum state $\sigma \in \C^{d \times d}$, ``measuring the observable'' entails measuring $\sigma$ in the $\ket{u_1}, \ldots, \ket{u_d}$ basis, receiving outcome $\bi \in [d]$, and outputting the eigenvalue $o_{\bi}$.
It is easy to see that the output of this experiment has expectation $\tr(O \cdot \sigma)$ and variance $\tr(O^2 \cdot \sigma) - \tr(O \cdot \sigma)^2$.

\paragraph{A purity estimator.}
Now, given copies of a mixed state $\rho \in \C^{d \times d}$ with spectrum $\alpha = (\alpha_1, \ldots, \alpha_d)$, our goal is to estimate the moments $p_k(\alpha)$ for various values of $k$.
Equivalently, we can write these quantities without reference to the spectrum:
\begin{equation*}
    p_k(\alpha) = \sum_{i=1}^d \alpha_i^k = \tr(\rho^k).
\end{equation*}
We begin with the simplest nontrivial moment, the second moment $\tr(\rho^2)$, also known as the ``purity'' of~$\rho$.
Our estimator will make use of the $\swap$ matrix, which we recall is the Hermitian matrix acting on $\C^d \otimes \C^d$ as follows:
\begin{equation*}
    \swap \cdot \ket{i} \otimes \ket{j} = \ket{j} \otimes \ket{i}, \quad \text{for all }i, j \in [d].
\end{equation*}
The $\swap$ matrix satisfies the well-known identity $\tr(\swap \cdot A \otimes B) = \tr(A \cdot B)$; plugging in $A = B = \rho$ gives us $\tr(\swap \cdot \rho^{\otimes 2}) = \tr(\rho^2)$.
Hence, if we interpret the $\swap$ matrix as an observable, then measuring this observable on $\rho^{\otimes 2}$ gives us an unbiased estimator for the purity $\tr(\rho^2)$, and it has variance
\begin{equation*}
    \tr(\swap^2 \cdot \rho^{\otimes 2}) - \tr(\swap \cdot \rho^{\otimes 2})^2
    = 1 - \tr(\rho^2)^2.
\end{equation*}
For most $\rho$, this variance will be close to 1, which will give us a poor estimator on just two copies.
To understand why, note that the $\swap$ is a unitary matrix which satisfies $\swap^2 = I$, and so it only has $\pm 1$ eigenvalues.
This means that if we measure $\swap$ on $\rho^{\otimes 2}$ when $\rho$ has purity $\tr(\rho^2)$ close to 0, then the eigenvalue we measure will always be far from the purity, even if it equals the purity in expectation.

In general, though, we have $n$ copies of $\rho$, not just two, and we would like to use these to improve the variance of our estimator.
One natural strategy is to group the $n$ copies into $n/2$ pairs, measure the $\swap$ observable on each pair, and then average together the $n/2$ results.
This will result in a better variance as $n$ increases;
however, B\u{a}descu, O'Donnell, and Wright observed that there is an even better strategy one can perform.
Define the matrix
\begin{equation*}
    O_2 = \frac{1}{\binom{n}{2}} \cdot \sum_{i < j} \swap_{i j},
\end{equation*}
where $\swap_{i j}$ is the matrix which applies $\swap$ to the $i$-th and $j$-th registers and acts as the identity on the remaining registers.
This matrix is Hermitian, which means we can interpret it as an observable, and if we measure this observable on $\rho^{\otimes n}$ then the output will have expectation
\begin{equation*}
    \tr(O_2 \cdot \rho^{\otimes n})
    = \frac{1}{\binom{n}{2}} \sum_{i < j} \tr(\swap_{ij} \cdot \rho^{\otimes  n})
    = \frac{1}{\binom{n}{2}} \sum_{i < j} \tr(\rho^2)
    = \tr(\rho^2).
\end{equation*}
As a result, this observable also gives an unbiased estimator for the purity $\tr(\rho^2)$.
But its variance will be substantially smaller than that of the $\swap$ matrix,
due to the averaging over all $i < j$.
Indeed, B\u{a}descu, O'Donnell, and Wright showed that it is the minimum variance estimator among all unbiased estimators for $\tr(\rho^2)$.

\paragraph{An estimator for higher moments.}
To generalize this construction to higher moments, we must first define a generalization of the swap matrix.
Given a permutation $\pi \in S_n$, we write $P(\pi)$ for the matrix which acts on $(\C^d)^{\otimes n}$ as follows:
\begin{equation*}
    P(\pi) \cdot \ket{i_1, \ldots, i_n} = \ket*{i_{\pi^{-1}(1)}, \ldots, i_{\pi^{-1}(n)}}, \quad \text{for all }i_1, \ldots, i_n \in [d].
\end{equation*}
If $\pi$ is a $k$-cycle,
then $\tr(P(\pi) \cdot \rho^{\otimes n}) = \tr(\rho^{k})$, which generalizes the $\swap$ identity.
As a result, the natural estimator for the $k$-th moment is
\begin{equation*}
    O_k = \mathop{\mathrm{avg}}_{\text{$k$-cycles $\pi$}}\{P(\pi)\}.
\end{equation*}
As in the purity case,
this is an unbiased estimator for the $k$-th moment $\tr(\rho^k)$,
and it has the minimum variance among all such estimators.
This is the family of observables that we will use to estimate our moments.
These observables have the property that they all simultaneously commute with each other, which means that they can be simultaneously measured with each other.
Formally, the reason for this is that each $O_k$ lies in the center of the group algebra $\C[S_n]$, which implies that each $O_k$ commutes with every matrix of the form $P(\pi)$.

This observable perspective is due to~\cite{BOW19}.
These estimators were first introduced by~\cite{OW16},
where they were presented in an explicit diagonalization, and it was in this form that they were analyzed by~\cite{AISW20}.
We believe this observable perspective is conceptually clearer, and it makes doing calculations much simpler, so we have chosen to present it in this form.

\section{Preliminaries}

\subsection{Notation}
In this section, we establish some general notation.

\begin{itemize}
    \item We use \textbf{boldface} for random variables, and \textsf{sans serif} for registers, i.e.\ subsystems of a multipartite system.
For instance, $\rho = \rho_{\reg{A}_1 \dots \reg{A}_n}$ denotes a state on $n$ registers. When the underlying registers are clear from context, we will sometimes drop the subscript.
We then write reduced states via partial trace. For example,
$\rho_{\reg{A}_1} = \tr_{\reg{A}_2 \dots \reg{A}_n}(\rho).$
Unless otherwise specified, $\reg{1}$ refers to the first subsystem, $\reg{2}$ to the second, and so on.

\item We use $\norm{X}_k$ to denote the Schatten $k$-norm. In particular, $\norm{X}_1$, $\norm{X}_2$, and $\norm{X}_\infty$ are the trace, Frobenius, and operator norms, respectively.

\item Given a Hermitian operator $A \in \C^{d \times d}$, we use $\lambda_i(A)$ for $i \in [d]$ to denote the $i$-th largest eigenvalue of $A$. The notation
\begin{equation*}
    \spec(A) = (\lambda_1(A), \dots, \lambda_d(A)),
\end{equation*}
refers to the eigenvalue list of $A$ sorted in non-increasing order, and $\spec(A)_{\leq r}$ is the eigenvalue list truncated to the top $r$ entries. We will also use $A_{\leq k}$ for the restriction of $A$ to its top $k$ eigenvectors, i.e.\ if $A = \sum_{i=1}^d \lambda_i(A) \cdot \ketbra{v_i}$, then $$A_{\leq k} = \sum_{i=1}^k \lambda_i(A) \cdot \ketbra{v_i}.$$ We will allow ties to be settled arbitrarily. 

\item For two tuples $a = (a_1, \dots, a_m)$ and $b = (b_1, \dots, b_n)$, we write $a||b$ for their concatenation, i.e.\ $$a||b \coloneq (a_1, \dots, a_m, b_1, \dots, b_n).$$

\item We use $n^{\uparrow k}$ and $n^{\downarrow k}$ for the \emph{$k$-th rising power} and \emph{$k$-th falling power}, respectively. That is,
\begin{equation*}
    n^{\uparrow k} \coloneq n (n+1) \dots (n+k-1), \qquad n^{\downarrow k} \coloneq n (n-1) \dots (n-k+1).
\end{equation*}

\item We write $1\{ X \}$ for the indicator function, which is $1$ if $X$ is true, and $0$ otherwise. We also use the Kronecker delta, $\delta_{ij}$, which is equal to $1\{ i = j\}$.

\item Finally, less a notation than a convention: we use the phrase ``with high constant probability'' to mean with probability at least a sufficiently large constant, which we take to be $99\%$. The choice of constant is arbitrary by standard amplification arguments (see, e.g., \cite[Proposition 2.4]{HKOT23}).

\end{itemize}

\subsection{The symmetric subspace}

In this section, we recall some basic facts about the symmetric subspace. For much more, see, e.g., \cite{Har13} or~\cite{Mel24}.

We write the symmetric group on $n$ elements as $S_n$. The symmetric group has a natural representation on $(\C^d)^{\otimes n}$, $P_d$, where $P_d(\pi)$ acts by permuting the $n$ tensor factors according to $\pi$. That is, for all basis kets,
\begin{equation*}
    P_d(\pi) \cdot \ket{x_1, \dots, x_n} = \ket*{x_{\pi^{-1}(1)}, \dots, x_{\pi^{-1}(n)}}.
\end{equation*}
When $d$ is clear from context, we will drop the subscript and write $P(\pi)$.

The \emph{symmetric subspace} on $(\C^d)^{\otimes n}$ is the subspace fixed by all permutations:
\begin{equation*}
    \lor^n \C^d \coloneq \big\{ \ket{\psi} \in (\C^d)^{\otimes n} \, \mid \, P(\pi) \cdot \ket{\psi} = \ket{\psi}, \, \forall \pi \in S_n \big\}.
\end{equation*}
Equivalently,
\begin{equation*}
    \lor^n \C^d = \mathrm{span} \big\{ \ket{\phi}^{\otimes n} \, \mid \, \ket{\phi} \in \C^d \big\}.
\end{equation*}

Consider a computational basis vector $\ket{x_1, \dots, x_n}$. The \emph{type} of $x = (x_1, \dots, x_n)$ is the vector $\tau = (\tau_1, \dots, \tau_d)$, where $\tau_i$ is the number of indices $j$ for which $x_j = i$. We define the \emph{weight} of a type $\tau$ as $\mathrm{wt}(\tau) \coloneq \sum_{i} \tau_i = n$. We also want to identify $\tau$ with the \emph{set of all $x$ of type $\tau$}, and write $x \in \tau$ to denote that $x$ has type $\tau$. The \emph{type vector} corresponding to a type $\tau$ is the state
\begin{equation*}
    \ket{\tau} \coloneq \frac{1}{\sqrt{\Abs{\tau}}} \cdot \sum_{x \in \tau} \ket{x}.
\end{equation*}
The collection of all type vectors forms an orthonormal basis of the symmetric subspace. By counting type vectors, we see
\begin{equation*}
    \dim(\lor^n \C^d) = \binom{n+d-1}{n}.
\end{equation*}
The projector onto the symmetric subspace is written $\Pisym^{(n,d)}$. We have
\begin{equation*}
    \Pisym^{(n,d)} = \E_{\bpi \sim S_n} \big[P_d(\bpi) \big] = \dim(\lor^n \C^d) \cdot \E_{\ket{\bu} \sim \mathrm{Haar}} \big[ \ketbra{\bu}^{\otimes n} \big].
\end{equation*}
The first equality shows that $\Pisym^{(n,d)}$ commutes with $M^{\otimes n}$ for any operator $M \in \C^{d \times d}$, since each $P_d(\pi)$ commutes with $M^{\otimes n}$. We may drop either or both of the superscripts of $\Pisym^{(n,d)}$ when they are clear from context.

\subsection{Tomography}

Quantum state tomography is one of the most basic tasks in quantum learning theory. Given $n$ copies of an unknown rank-$r$ state $\rho \in \C^{d \times d}$, the goal is to produce a classical estimate $\widehat{\brho}$ of $\rho$ by performing measurements and processing the outcomes.

\subsubsection{Pure state tomography}

Pure state tomography refers to the special case of quantum state tomography in which the rank $r = 1$. The principal algorithm for pure state tomography using entangled measurements is \emph{Hayashi's algorithm} \cite{Hay98}, which is sample-optimal in all standard parameters \cite{PSTW25,SSW25}. Moreover, it is provably the optimal algorithm for learning pure states in fidelity \cite{Hay98}.

{
\floatstyle{boxed}
\restylefloat{figure}
\begin{figure}[H]
Given $n$ copies of $\ket{u} \in \C^{d}$:
\begin{enumerate}
    \item Measure the copies with the POVM $\{\dim(\lor^n \C^d) \cdot \ketbra{v}^{\otimes n} \cdot dv\}$. Let $\ket{\bv}$ be the outcome.
    \item Output $\ketbra{\bv}$.
\end{enumerate}
\caption{Hayashi's pure state tomography algorithm $\hayashi$.}
\label{fig:Hayashi_algorithm}
\end{figure}
}

It is often more convenient to consider the following debiased version of Hayashi's algorithm, due to Grier, Pashayan, and Schaeffer \cite{GPS24}. Their algorithm also uses Hayashi's measurement; what changes is the output.
{
\floatstyle{boxed}
\restylefloat{figure}
\begin{figure}[H]
Given $n$ copies of $\ket{u} \in \C^{d}$:
\begin{enumerate}
    \item Measure the copies with the POVM $\{\dim(\lor^n \C^d) \cdot \ketbra{v}^{\otimes n} \cdot dv\}$. Let $\ket{\bv}$ be the outcome.
    \item Output the estimator
    \begin{equation*}\widehat{\bsigma} \coloneqq \frac{d+n}{n} \cdot \ketbra{\bv} - \frac{1}{n} \cdot I_d.\end{equation*}
\end{enumerate}
\caption{The Grier--Pashayan--Schaeffer pure state tomography algorithm $\gps$.}
\label{fig:GPS_algorithm}
\end{figure}
}
Like Hayashi's algorithm, the Grier--Pashayan--Schaeffer algorithm is sample-optimal in all standard parameters. It also has the additional property of being unbiased: $\E [ \widehat{\bsigma} ] = \ketbra{u}$.

The following result will be useful in our analysis: it gives an exact expression for the $k$-th moment of the outcome of Hayashi's measurement.

\begin{lemma}[Chiribella's theorem] \label{thm:Chiribella}
Let $\ket{u} \in \C^d$ be a pure state, and let $\ket{\bv}$ be the output of Hayashi's algorithm, given $n$ copies of $\ket{u}$. Then
\begin{equation*}
    \E \big[\ketbra{\bv}^{\otimes k} \big] = \frac{k!}{(d+n)^{\uparrow k}} \cdot \Pisym^{(k,d)}\cdot  \bigg( \sum_{\ell=0}^k \binom{n}{\ell} \cdot \binom{k}{\ell} \cdot \ketbra{u}^{\otimes \ell} \otimes I^{\otimes (k-\ell)}\bigg) \cdot \Pisym^{(k,d)}.
\end{equation*}
\end{lemma}

This identity is a specialization of Chiribella's original theorem, proved in \cite{Chi11}. The version we give here follows from the reformulation in \cite[Theorem 21]{GPS24b}.\footnote{Note, however, a typo: their overall expression should be multiplied by $k!$. In the third equation of their proof, we should use $\binom{d+k+n-1}{k} = \frac{(d+n)^{\uparrow k}}{k!}$ rather than $(d+n)^{\uparrow k}$ as written.}

\subsubsection{Mixed state tomography}

Mixed state tomography refers to the general case of tomography. For many years, the sample-optimal algorithms for mixed state tomography \cite{HHJ+16,OW16,PSW25,Yue23,SSW25} were much more complicated than sample-optimal pure state tomography algorithms \cite{Hay98,GKKT20}. Recently, however, mixed state tomography was reduced to pure state tomography using the random purification channel \cite{TWZ25}, yielding simple sample-optimal algorithms for mixed state tomography \cite{PSTW25}. We first recall the random purification channel.

\begin{lemma}[The random purification channel]\label{lem:random_purification_channel}
    Let $n \geq 1$, $d \geq 1$, $r \leq d$ be integers. There exists a quantum channel $\purifychan^{(d,r)}(\cdot)$ which acts as follows. Given as input $n$ copies of a rank-$r$ mixed state $\rho \in \C^{d \times d}$,
    \begin{equation*}
        \purifychan^{(d,r)} ( \rho^{\otimes n} ) = \E_{\ket{\brho}} \big[ \ketbra{\brho}^{\otimes n} \big].
    \end{equation*}
    Here, the expectation is over a uniformly random purification $\ket{\brho} \in \C^d \otimes \C^r$ of $\rho$. By a uniformly random purification, we mean the distribution of purifications obtained by applying a Haar-random unitary $\bU$ to the second register of any fixed purification $\ket{\rho} \in \C^d \otimes \C^r$.
\end{lemma}

The random purification channel's action can be viewed as sampling a random purification $\ket{\brho}$, and then giving us $n$ copies of that purification. Since learning any purification suffices to learn a mixed state, we obtain the following meta-algorithm for mixed state tomography, analyzed in \cite{PSTW25}.

{
\floatstyle{boxed}
\restylefloat{figure}
\begin{figure}[H]
Given $n$ copies of a rank-$r$ mixed state $\rho$:
\begin{enumerate}
    \item Apply $\purifychan^{(d, r)}$ to prepare $n$ copies of a random purification $\ket{\brho} \in \C^d \otimes \C^r$.
    \item Run a pure state tomography algorithm $\calA$ on $\ket{\brho}^{\otimes n}$. Let $\widehat{\bsigma}$ be the output.
    \item Output $\widehat{\brho} = \tr_{\reg{2}}(\widehat{\bsigma})$.
\end{enumerate}
\caption{A framework for lifting a pure state tomography algorithm $\calA$ to a mixed state tomography algorithm $\mixed(\calA)$.}
\label{fig:reduction}
\end{figure}
}

Particularly important for both \cite{PSTW25} and the current paper is $\mixed(\gps)$. For completeness, we now describe this algorithm in full, although the ingredients are already contained in \Cref{fig:GPS_algorithm} and \Cref{fig:reduction}.

{
\floatstyle{boxed}
\restylefloat{figure}
\begin{figure}[h!]
Given $n$ copies of a rank-$r$ mixed state $\rho$:
\begin{enumerate}
    \item Apply $\purifychan^{(d, r)}$ to prepare $n$ copies of a random purification $\ket{\brho} \in \C^d \otimes \C^r$.
    \item Run $\gps$ on $\ket{\brho}^{\otimes n}$. That is,
    \begin{enumerate}
        \item Measure the copies of $\ket{\brho}$ with the POVM $\{ \dim(\lor^n \C^{dr}) \cdot \ketbra{v}^{\otimes n} \cdot dv\}$. Let $\ket{\bv} \in \C^d \otimes \C^r$ be the outcome.
        \item Define
        \begin{equation*}
            \widehat{\bsigma} \coloneq \frac{dr + n}{n} \cdot \ketbra{\bv} - \frac{1}{n} \cdot I_{dr}.
        \end{equation*}
    \end{enumerate}
    \item Output $\widehat{\brho} = \tr_{\reg{2}}(\widehat{\bsigma})$.
\end{enumerate}
\caption{The mixed state tomography algorithm $\mixed(\gps)$.}
\label{fig:mix_gps}
\end{figure}
}

The algorithm $\mixed(\gps)$ is sample-optimal in all standard parameters, and unbiased, i.e.\ $\E [ \widehat{\brho} ] = \rho$.

\subsection{Sub-gamma random variables}
Here we review some basic facts about sub-gamma random variables. All of the material below is standard. For much more, see, for example, \cite[Chapter 2.4]{BLM13}.

\begin{definition}[Sub-gamma random variable] \label{def:subgamma}
    Let $\bX$ be a real-valued random variable, with mean $\E[\bX] = \mu$. Then $\bX$ is \emph{$(v, c)$-sub-gamma} if, for all $\Abs{t} < 1/c$,
    \begin{equation*}
        \E\big[ e^{t(\bX - \mu)}\big] \leq \exp(\frac{vt^2}{2(1 - c\Abs{t})}).
    \end{equation*}
    The numbers $v$ and $c$ are called the \emph{variance parameter} and \emph{scale parameter}, respectively.
\end{definition}

If $v_1 \leq v_2$ and $c_1 \leq c_2$, then any $(v_1,c_1)$-sub-gamma random variable is also $(v_2,c_2)$-sub-gamma. We will also use the following basic closure properties.

\begin{lemma}[Shifts, rescalings, and sums of sub-gamma random variables]
\label{lem:scalings_sums_subgammas}
Suppose $\bX_1$ is $(v_1,c_1)$-sub-gamma and $\bX_2$ is $(v_2,c_2)$-sub-gamma. Then:
\begin{enumerate}
    \item \label{item:subgamma_shift} For any $\Delta \in \R$, the random variable $\bX_1 + \Delta$ is $(v_1,c_1)$-sub-gamma.
    \item \label{item:subgamma_scale} For any $\lambda \in \R$, the random variable $\lambda \bX_1$ is $(\lambda^2 \cdot v_1, \Abs{\lambda} \cdot c_1)$-sub-gamma.
    \item \label{item:subgamma_sum} $\bX_1 + \bX_2$ is $\big(2 (v_1+v_2),2 \max(c_1,c_2)\big)$-sub-gamma.
\end{enumerate}
\end{lemma}

\begin{proof}
    \Cref{item:subgamma_shift} and \Cref{item:subgamma_scale} follow straightforwardly from \Cref{def:subgamma}. For \Cref{item:subgamma_sum}, we have 
    \begin{equation*}
        \E \big[e^{t(\bX_1 + \bX_2 - \mu_1 - \mu_2)} \big] \leq \sqrt{\E \big[e^{2t(\bX_1 -\mu_1)} \big]} \cdot \sqrt{\E \big[ e^{2t(\bX_2- \mu_2)} \big]} \leq \exp(\frac{2v_1t^2}{2(1 - 2c_1\Abs{t})}) \cdot \exp(\frac{2v_2t^2}{2(1 - 2c_2\Abs{t})}).
    \end{equation*}
    The first inequality is Cauchy--Schwarz. Using $c_i \leq \max(c_1, c_2)$ and merging the two exponentials completes the proof.
\end{proof}

\begin{lemma}[Mixtures of sub-gamma random variables]
\label{lem:mixture_of_subgammas}
Suppose $\bX$ is obtained by first sampling an index $\bi \in I$ according to a probability measure $\pi$, and then sampling from $\bX_{\bi}$, where each $\bX_i$ is $(v,c)$-sub-gamma and satisfies $\E[\bX_i]=\mu$. Then $\bX$ is also $(v,c)$-sub-gamma.
\end{lemma}

\begin{proof}
    For all $t$ with $\Abs{t} < 1/c$,
    \begin{equation*}
        \E\big[e^{t(\bX-\mu)}\big]
        = \int_I \E\big[e^{t(\bX_i-\mu)}\big] \, d\pi(i)
        \leq \int_I \exp(\frac{vt^2}{2(1-c\Abs{t})}) \, d\pi(i)
        = \exp(\frac{vt^2}{2(1-c\Abs{t})}).
        \qedhere
    \end{equation*}
\end{proof}

Finally, we have the following tail bound that can be found in~\cite[Theorem 2.3]{BLM13}.

\begin{lemma}[Tail bound for sub-gamma random variables] \label{lem:subgamma_probability_bound}
    If $\bX$ is $(v,c)$-sub-gamma, then, for all $\delta > 0$,
    \begin{equation*}
        \Pr\big[ \Abs{\bX - \mu}  \geq \sqrt{2v\delta} + c\delta \big] \leq 2e^{-\delta}.
    \end{equation*}
\end{lemma}

\subsection{Moments of the binomial distribution}
We collect a few bounds for the moments of the binomial distribution that we will need to analyze our moment matching estimator for sub-normalized states in~\Cref{sec:moment-estimation}.

\begin{lemma}
    \label{lem:binom-facts-downarrow}
    Let $\bx \sim \mathrm{Binom}(n, p)$. Then $\E[\bx^{\downarrow k}] = n^{\downarrow k} \cdot p^k$.
\end{lemma}

\begin{proof}
    Consider the function
    \begin{equation*}
        f(y) \coloneq \E\big[y^{\bx}\big] = \sum_{x=0}^n \binom{n}{x} \cdot p^x \cdot (1-p)^{n-x} \cdot y^x = \big( (1-p) + p\cdot y \big)^n.
    \end{equation*}
    On the one hand, the $k$-th derivative of this function is
    \begin{equation*}
        \frac{\partial^k}{\partial y^k} f(y)
         = \frac{\partial^k}{\partial y^k} \E \big[ y^{\bx} \big] = \E\big[\bx^{\downarrow k} y^{\bx - k}\big].
    \end{equation*}
    However, on the other hand, we have
    \begin{equation*}
        \frac{\partial^k}{\partial y^k} f(y) = \frac{\partial^k}{\partial y^k} \big( (1-p) + p \cdot y \big)^n = p^k \cdot n^{\downarrow k} \cdot \big( (1-p) + p \cdot y \big)^{n-k}.
    \end{equation*}
    Setting these two expressions equal, and evaluating at $y = 1$ gives the desired result:
    \begin{equation*}
        \E\big[\bx^{\downarrow k} \big] = \frac{\partial^k}{\partial y^k} f(y) \Big|_{y=1} = n^{\downarrow k} \cdot p^k. \qedhere
    \end{equation*}
\end{proof}

\begin{lemma}
    \label{lem:binom-facts-power}
    Let $\bx \sim \mathrm{Binom}(n, p)$. Then $\E[\bx^k] \leq 2^k \cdot (k^k + n^k \cdot p^k)$.
\end{lemma}

\begin{proof}
    If $p = 0$, then $\bx^k = 0$, which satisfies the bound trivially. So, assume $p > 0$. It is known from~\cite{Ahl22} that the $k$-th raw moment of the binomial distribution satisfies
    \begin{equation*}
        \E\big[\bx^k\big] \leq \bigg(\frac{k}{\log\big(1 + \frac{k}{np}\big)}\bigg)^k.
    \end{equation*}
    Since $\log(1 + u) \geq u/(1+u)$ for $u \geq 0$, it holds that
    \begin{equation*}
        \E\big[\bx^k\big] \leq \bigg(k \cdot \Big( \frac{k + np}{k} \Big)\bigg)^k
        \leq 2^k \cdot \big(k^k + n^k \cdot p^k\big). \qedhere
    \end{equation*}
\end{proof}

\begin{lemma}
    \label{lem:binom-facts-var-downarrow}
    Let $\bx \sim \mathrm{Binom}(n, p)$. Then
    \begin{equation*}
        \Var\big[\bx^{\downarrow k}\big] \leq 2^k \cdot k! \cdot \big(n^k \cdot p^k + n^{2k-1}\cdot p^{2k-1} \big).
    \end{equation*}
\end{lemma}

\begin{proof}
    The square of a falling factorial satisfies the following identity,
    \begin{equation*}
        \big(x^{\downarrow k}\big)^2
        = \sum_{\ell=0}^k \binom{k}{\ell}^2 \cdot \ell! \cdot x^{\downarrow (2k - \ell)},
    \end{equation*}
    a special case of~\cite[Corollary 8]{Ros02}.\footnote{This corollary tells us \begin{equation*}
        x^{\downarrow k} \cdot x^{\downarrow k} = \sum_{0 \leq i \leq 2k} c_i \cdot x^{\downarrow i}.
    \end{equation*}
    Here $c_i$ counts the number of ways to partition $[2k]$ into \emph{exactly} $i$ nonempty sets such that each set contains neither two elements from $A = \{1, \dots, k\}$ nor two elements from $B = \{k+1, \dots, 2k\}$. For this quantity to be nonzero, we must have $i \in \{k, \dots, 2k\}$. Taking $i = 2k-\ell$ for $\ell \in \{0, 1, \dots, k\}$, we must have $\ell$ two-element sets containing one element of $A$ and one element of $B$. The number of ways to choose such pairs is $\binom{k}{\ell}^2$, and the number of matchings is $\ell!$. }
    Thus, by \Cref{lem:binom-facts-downarrow},
    \begin{equation*}
        \E\big[\big(\bx^{\downarrow k}\big)^2\big]
        = \sum_{\ell=0}^k \binom{k}{\ell}^2 \cdot \ell! \cdot \E\big[\bx^{\downarrow (2k - \ell)}\big]
        = \sum_{\ell=0}^k \binom{k}{\ell}^2 \cdot \ell! \cdot n^{\downarrow (2k-\ell)} \cdot p^{2k-\ell}.
    \end{equation*}
    We can rewrite the variance as
    \begin{align*}
        \Var\big[\bx^{\downarrow k}\big]
        &= \E\big[\big(\bx^{\downarrow k}\big)^2\big] - \E\big[\bx^{\downarrow k}\big]^2 \\
        &= \bigg(\sum_{\ell=0}^k \binom{k}{\ell}^2 \cdot \ell! \cdot n^{\downarrow (2k-\ell)} \cdot p^{2k-\ell}\bigg) - \big(n^{\downarrow k}\big)^2 \cdot p^{2k} \\
        &= \bigg(\sum_{\ell=1}^k \binom{k}{\ell}^2 \cdot \ell! \cdot n^{\downarrow (2k-\ell)} \cdot p^{2k-\ell}\bigg) + \Big(n^{\downarrow (2k)} - \big(n^{\downarrow k}\big)^2 \Big)\cdot p^{2k}.
    \end{align*}
    It holds that $n^{\downarrow (2k)} \leq (n^{\downarrow k})^2$, and that $n^{\downarrow (2k-\ell)} \leq n^{2k-\ell}$. Hence
    \begin{equation*}
        \Var[\bx^{\downarrow k}] \leq \sum_{\ell=1}^k \binom{k}{\ell}^2 \cdot \ell! \cdot n^{2k-\ell} \cdot p^{2k-\ell}.
    \end{equation*}
    Moreover, it holds for all $u \geq 0$ that $u^t \leq u^{a} + u^{b}$ for any $a \leq t \leq b$. In particular, for all $\ell \in [k]$ we have $k \leq 2k-\ell \leq 2k-1$, and hence $n^{2k-\ell} \cdot p^{2k-\ell} \leq n^k \cdot p^k + n^{2k-1} \cdot p^{2k-1}$. We then simplify the bound on the variance as follows:
    \begin{align*}
        \Var\big[\bx^{\downarrow k}\big]
        &\leq \sum_{\ell=1}^k \binom{k}{\ell}^2 \cdot \ell! \cdot \big(n^k \cdot p^k + n^{2k-1} \cdot p^{2k-1} \big) \\
        & = \big(n^k \cdot p^k + n^{2k-1} \cdot p^{2k-1} \big) \cdot \sum_{\ell=1}^k \binom{k}{\ell} \cdot \frac{k!}{(k - \ell)!} \\
        &\leq k! \cdot \big(n^k \cdot p^k + n^{2k-1} \cdot p^{2k-1} \big) \cdot \sum_{\ell=1}^k \binom{k}{\ell} \\
        &\leq 2^k \cdot k! \cdot \big(n^k \cdot p^k + n^{2k-1} \cdot p^{2k-1}\big).
    \end{align*}
    This completes the proof.
\end{proof}

\subsection{Entangled moment estimators}

In this subsection, we state the moment estimation result that we will need for our algorithm. Recall from~\Cref{sec:tech-overview-moment-estimation} that our goal is, given $n$ copies of a state $\rho \in \C^{d \times d}$ with spectrum $\alpha = (\alpha_1, \dots, \alpha_d)$, to estimate the moments $p_k(\alpha) = \tr(\rho^k)$ of the spectrum, for all $k \leq n$, by measuring the observables
\begin{equation*}
    O_k = \mathop{\mathrm{avg}}_{\text{$k$-cycles $\pi$}}\{P(\pi)\}.
\end{equation*}
Recall also that ``measuring the observable'' refers to the process of measuring $\rho^{\otimes n}$ in the eigenbasis of $O_k$ and then outputting the eigenvalue corresponding to the measurement outcome.
To understand this process, consider the eigendecomposition of $O_k$. By Proposition 7.4 and Remark 7.8 of~\cite{BOW19},
\begin{equation}
    \label{eq:eigendecomposition-of-Ok}
    O_k = \sum_{\lambda} \frac{p^{\#}_{(k)}(\lambda)}{n^{\downarrow k}} \cdot \Pi_{\lambda}.
\end{equation}
Here, the $p^{\#}_{(k)}$'s are a well-studied family of polynomials from the representation theory of the symmetric group.
In addition, the $\Pi_{\lambda}$'s are a particular family of orthogonal projectors indexed by Young diagrams~$\lambda$ of $n$ boxes and at most $d$ rows,
and the procedure of performing the projective measurement $\{\Pi_{\lambda}\}_{\lambda}$ is known as \emph{weak Schur sampling}. The exact definition and properties of $p^{\#}_{(k)}$ and $\Pi_{\lambda}$ are outside the scope of this paper;
we bring them up only to quote the result we need from~\cite{AISW20}, which we will then recast in terms of the $O_k$ observables.
We can therefore restate our estimators as follows.
\begin{definition}[Moment estimation algorithm]
    \label{def:moment-est-alg}
    Given $n$ copies of $\rho \in \C^{d \times d}$, measuring the $O_k$ observable for $k \leq n$ is equivalent to the following procedure.
    \begin{enumerate}
        \item Perform the projective measurement $\{\Pi_{\lambda}\}_{\lambda}$ on $\rho^{\otimes n}$. Let $\blambda$ be the resulting Young diagram.
        \item Return $\widehat{\bp}_k(\rho) = p^{\#}_{(k)}(\blambda)/n^{\downarrow k}$.
    \end{enumerate}
\end{definition}
\begin{remark}
    \label{rem:estimate-moments-simultaneously}
    As noted in \Cref{sec:tech-overview-moment-estimation}, the $O_k$ observables commute with each other.
    As a result, they can be simultaneously diagonalized and therefore simultaneously measured.
    This shared eigenbasis is shown explicitly in \Cref{eq:eigendecomposition-of-Ok} and corresponds to the $\Pi_{\lambda}$ projectors.
    Thus, one can produce estimates $\widehat{\bp}_{k_1}(\rho), \widehat{\bp}_{k_2}(\rho), \dots, \widehat{\bp}_{k_m}(\rho)$ for multiple moments using the same $n$ samples simultaneously. This is done by performing weak Schur sampling, and upon obtaining the Young diagram $\blambda$, we compute the estimates
    \begin{equation*}
        \widehat{\bp}_{k_1}(\rho) = \frac{p^{\#}_{(k_1)}(\blambda)}{n^{\downarrow k_1}},\quad
        \widehat{\bp}_{k_2}(\rho) = \frac{p^{\#}_{(k_2)}(\blambda)}{n^{\downarrow k_2}}, \quad\dots,\quad
        \widehat{\bp}_{k_m}(\rho) = \frac{p^{\#}_{(k_m)}(\blambda)}{n^{\downarrow k_m}}
    \end{equation*}
    using the same $\blambda$.
\end{remark}
In~\cite{AISW20}, Acharya, Issa, Shende, and Wagner obtained optimal bounds for the complexity of estimating the moments of states by studying the $p^{\#}_{(k)}$ polynomials to understand the variance of $\widehat{\bp}_k(\rho)$. We will make use of the following result.

\begin{lemma}[{\cite[Lemma 9]{AISW20}}]
    \label{lem:moment-estimation-norm}
    Let $\rho \in \C^{d \times d}$ be a mixed state,
    and let $k, n \in \N$.
    Suppose we perform weak Schur sampling on $\rho^{\otimes n}$ and receive as outcome the Young diagram $\blambda$.
    Then
    \begin{equation*}
        \E\big[p^{\#}_{(k)}(\blambda)\big] = n^{\downarrow k} \cdot \tr(\rho^k).
    \end{equation*}
    Moreover,
    \begin{equation*}
       \Var\big[p^{\#}_{(k)}(\blambda)\big] \leq k^{6k} \cdot n^k + k^{6k} \cdot n^{2k-1} \cdot \tr(\rho^{2k-1}).
    \end{equation*}
\end{lemma}
The discussion above shows that $\widehat{\bp}_k(\rho)$ is an unbiased estimator for the $k$-th moment of $\rho$ whenever we are given $n \geq k$ copies of the state. For the spectrum learning algorithm, we will estimate the moments of a sub-normalized state, which may require us to operate in the regime where the number of copies $n$ is strictly smaller than the moment order $k$. In that case, the quantity $\widehat{\bp}_k(\rho)$ is not well-defined, since we cannot divide by $n^{\downarrow k} = 0$. Fortunately, the quantity $p^{\#}_{(k)}(\blambda)$ is still defined and equal to zero, allowing us to extend the moment estimation algorithm to sub-normalized states.

\section{Observable concentration and relative error bounds} \label{sec:moments}

In this section, we prove the refined tomography guarantee stated in \Cref{thm:main-tool}. This is our main tool for analyzing our entangled bucketing algorithm, which itself is one of the main ingredients in our new spectrum learning algorithm.

To prove this bound, we first consider the problem of estimating $X \coloneq \tr(O \cdot \rho)$, given an observable $O$ and copies of an unknown state $\rho$. A natural meta-algorithm uses the copies of $\rho$ to prepare an estimate $\widehat{\brho}$ and then computes the plug-in estimator $\bX \coloneq \tr(O \cdot \widehat{\brho})$. If $\widehat{\brho}$ is unbiased, then so is $\bX$, i.e.\ $\E[\bX] = X$.

Our starting point is an analysis of how well $\bX$ concentrates about its expectation. For brevity, we refer to this as \emph{observable concentration}, although more precisely it is concentration of the plug-in estimator for the observable's expectation value. In \Cref{sec:obs-conc-pure-states} we show that if $\rho$ is pure, and $\widehat{\brho}$ is generated using $\gps$, then $\bX$ achieves sub-gamma concentration. In \Cref{sec:obs-conc-mixed-states}, we extend this result to the case where $\rho$ is mixed, and we use $\mixed(\gps)$ to form $\widehat{\brho}$.
Next, in \Cref{sec:refined-uniform-via-observable-conc}, we prove our main tool, \Cref{thm:main-tool}, a bound on $\Abs{\bra{w} (\widehat{\brho} - \rho ) \ket{w}}$ for all $\ket{w}$ simultaneously, by applying our observable concentration results to a bespoke observable. We also state and prove a pair of straightforward corollaries which find application in \Cref{sec:spectrum_learning,sec:additional_tomography_results}. Finally, in \Cref{sec:relative_error_bounds_higher_rank}, we extend \Cref{thm:main-tool}, viewed as a statement about $\Abs{\tr( O \cdot (\widehat{\brho} - \rho))}$ for rank-one observables $O \leftarrow \ketbra{w}$, to higher-rank observables.


\subsection{Observable concentration for pure states}
\label{sec:obs-conc-pure-states}

In this section, we prove the following result.

\begin{proposition} [The $\gps$-based plug-in estimator is sub-gamma] \label{prop:pure_states_observable_subgamma}
     Let $\ket{u} \in \C^D$ be a pure state, and let $\widehat{\bsigma}$ be the output of $\gps$ given $n$ copies of $\ket{u}$. Let $O \in \C^{D \times D}$ be a positive semidefinite operator such that $\norm{O}_\infty \leq 1$. Then the random variable
    $\bX \coloneq \tr( O \cdot \widehat{\bsigma})$
    is $(v,c)$-sub-gamma, with
    \begin{equation*}
    v = \frac{2\cdot \tr(O^2) + 4n\cdot \tr(O^2 \cdot \ketbra{u})}{n^2},
    \qquad c =\frac{1}{n}. \end{equation*}
    More generally, if $O \in \C^{D \times D}$ is any observable such that $\norm{O}_\infty \leq 1$, then $\bX$ is $(2v,2c)$-sub-gamma.
\end{proposition}

\begin{proof}
    First, assume $O$ is PSD. Recall that $\gps$ first measures $\ketbra{u}^{\otimes n}$ with Hayashi's measurement, obtaining a pure state $\ketbra{\bu}$, and then outputs
    \begin{equation*}
        \widehat{\bsigma} = \frac{n+D}{n} \cdot \ketbra{\bu} - \frac{1}{n} \cdot I_D.
    \end{equation*}
    By design, we have $\E[\widehat{\bsigma}] = \ketbra{u}$. Thus, with $\bX = \tr(O \cdot \widehat{\bsigma})$, we have $\E[\bX] = \tr(O \cdot \ketbra{u})$.

    We want to show that $\bX$ is sub-gamma, which amounts to showing an upper bound on the moment generating function (MGF) of $\bX$. To this end, we first observe that since $x \mapsto \exp(-tx)$ is convex for any fixed $t$, by Jensen's inequality,
    \begin{equation}\label{eq:symmetrized_moment_upper_bounds_regular}
        \E\big[e^{t(\bX - \E[\bX])} \big] = \E\big[e^{t\bX}\big] \cdot  e^{ -t\E[\bX]} \leq \E\big[e^{t\bX}\big] \cdot \E \big[e^{-t\bX'} \big] = \E\big[e^{t(\bX-\bX')}\big].
    \end{equation}
    Here, $\bX'$ is an independent copy of $\bX$; more generally, throughout this proof, we will denote independent copies of random variables with a prime. Thus, it suffices to bound the right-hand side, which is the symmetrized MGF of $\bX$ (i.e.\ the MGF of the symmetrized random variable $\bX - \bX'$). One benefit of considering the symmetrized MGF is that $\bX - \bX' = \bY - \bY'$, where $\bY$ is the shifted version of $\bX$ obtained by removing the additive constant, which will help to simplify our analysis. That is, 
    \begin{equation*}
        \bY \coloneq \frac{n+D}{n} \cdot \tr\big( O \cdot \ketbra{\bu} \big).
    \end{equation*}

    Toward upper-bounding the symmetrized MGF of $\bY$, consider the $k$-th moment of $\bY - \bY'$. Since~$\bY'$ is identically distributed as $\bY$, we have
    \begin{equation*}
        \E \big[ (\bY - \bY')^k \big] = \sum_{\ell=0}^k (-1)^{k-\ell} \cdot \binom{k}{\ell} \cdot \E\big[ \bY^\ell\big] \cdot \E\big[(\bY')^{k-\ell}\big] = \sum_{\ell=0}^k (-1)^{k-\ell} \cdot \binom{k}{\ell} \cdot \E\big[ \bY^\ell\big] \cdot \E\big[\bY^{k-\ell}\big].
    \end{equation*}
    The moments of $\bY$ can be computed as
    \begin{equation*} \label{eq:moments_of_bY}
        \E\big[\bY^\ell\big] = \Big( \frac{n+D}{n} \Big)^{\ell} \cdot \E\big[\tr\big( O \cdot \ketbra{\bu}\big)^\ell\big] = \Big( \frac{n+D}{n} \Big)^{\ell} \cdot \tr( O^{\otimes \ell} \cdot \E\big[\ketbra{\bu}^{\otimes \ell}\big]).
    \end{equation*}
    We now use Chiribella's theorem (\Cref{thm:Chiribella}), which states that
    \begin{equation*}
    \E \big[\ketbra{\bu}^{\otimes \ell} \big] = \frac{\ell!}{(n+D)^{\uparrow \ell}} \cdot \Pisym^{(\ell,D)}\cdot  \Big( \sum_{m=0}^\ell \binom{n}{m} \cdot \binom{\ell}{m} \cdot \ketbra{u}^{\otimes m} \otimes I^{\otimes (\ell-m)}\Big) \cdot \Pisym^{(\ell,D)},
    \end{equation*}
    to get
    \begin{equation} \label{eq:moments_of_bY_2}
        \E \big[ \bY^\ell \big]  = \frac{\ell!}{n^\ell} \cdot \frac{(n+D)^\ell}{(n+D)^{\uparrow \ell}} \cdot \sum_{m=0}^{\ell} \binom{n}{m} \cdot \binom{\ell}{m} \cdot \tr\Big( O^{\otimes \ell} \cdot \Pisym^{(\ell,D)} \cdot \big(\ketbra{u}^{\otimes m} \otimes I^{\otimes (\ell - m)}\big) \cdot \Pisym^{(\ell,D)} \Big).
    \end{equation}
    To simplify our expression, we will use the abbreviation
    \begin{equation} \label{eq:Tlm_def}
        T_{\ell m} \coloneq \tr\Big( O^{\otimes \ell} \cdot \Pisym^{(\ell,D)} \cdot \big(\ketbra{u}^{\otimes m} \otimes I^{\otimes (\ell - m)}\big) \cdot \Pisym^{(\ell,D)} \Big).
    \end{equation}
    The quantity $T_{\ell m}$ is of the form $\tr(A \cdot B)$ for PSD operators $A = \Pisym^{(\ell,D)} \cdot O^{\otimes \ell} \cdot \Pisym^{(\ell,D)}$ and $B = \ketbra{u}^{\otimes m} \otimes I^{\otimes (\ell - m)}$, so $T_{\ell m} \geq 0$. With this notation, \Cref{eq:moments_of_bY_2} becomes
    \begin{equation} \label{eq:moments_of_bY3}
        \E \big[ \bY^\ell \big]  = \frac{\ell!}{n^\ell} \cdot \frac{(n+D)^\ell}{(n+D)^{\uparrow \ell}} \cdot \sum_{m=0}^{\ell} \binom{n}{m} \cdot \binom{\ell}{m} \cdot T_{\ell m}.
    \end{equation}

    From the previous expression, we see that these moments would simplify if one could replace the more-complicated $(n+D)^{\uparrow \ell}$ with the simpler expression $(n+D)^{\ell}$. Although this is not literally possible, in \Cref{lem:moments_increase_when_dropping_uparrow} below we show that there exists a random variable $\bZ$ with these simpler moments:
    \begin{equation} \label{eq:moments_of_Z}
        \E \big[ \bZ^\ell\big] = \frac{(n+D)^{\uparrow \ell}}{(n+D)^\ell} \cdot \E\big[ \bY^\ell \big] = \frac{\ell!}{n^\ell} \cdot  \sum_{m=0}^{\ell} \binom{n}{m} \cdot \binom{\ell}{m} \cdot T_{\ell m}.
    \end{equation}
    The same lemma also shows that, for all nonnegative integers $k$,
    \begin{equation*}
        \E \big[ (\bY - \bY')^k \big] \leq \E \big[ (\bZ - \bZ')^k \big],
    \end{equation*}
    from which we conclude
    \begin{equation*}
        \E \big[e^{t(\bY - \bY')} \big] \leq \E \big[e^{t(\bZ - \bZ')} \big].
    \end{equation*}
    Therefore, it now suffices to bound the symmetrized moments of $\bZ$.

    An intuitive way to think about this step is to imagine first embedding $\ket{u} \in \C^D$ into a much larger space $\C^{D'}$ with $D' \gg D$, and then performing tomography in the larger space. The corresponding moments are given by \Cref{eq:moments_of_bY3} with $D$ replaced by $D'$. As $D' \to \infty$, the factor $(n+D')^\ell/(n+D')^{\uparrow \ell}$ tends to $1$, yielding the moments in \Cref{eq:moments_of_Z}. Moreover, $T_{\ell m}$, which in principle depends on $D'$ through, e.g.\ $\Pisym^{(\ell,D)}$, is unchanged by this embedding once $D' \geq D$, since $\ketbra{u}$ is supported on a rank-$D$ subspace. Heuristically, passing to a larger ambient dimension should only make tomography harder, so we expect these limiting moments to dominate those in the original $D$-dimensional problem. \Cref{lem:moments_increase_when_dropping_uparrow} makes this intuition rigorous by constructing a random variable $\bZ$ with exactly these larger moments.

    The moments of $\bZ$, given in \Cref{eq:moments_of_Z}, are sufficiently tractable to permit an explicit computation of its MGF. In \Cref{lem:mgf_Z} below, we compute the MGF of $\bZ$, finding, for $\Abs{t} < n$,
    \begin{equation*}
        \E \big[ e^{t \bZ} \big] = \det \Big( I - \frac{t}{n} \cdot O \Big)^{-1} \cdot \Big( \bra{u} \cdot \Big( I - \frac{t}{n} \cdot O \Big)^{-1} \cdot \ket{u} \Big)^n.
    \end{equation*}
    With the MGF of $\bZ$ in hand, we can analyze the symmetric MGF using
    \begin{equation*}
        \E \big[e^{t(\bZ - \bZ')} \big] = \E \big[e^{t\bZ} \big] \cdot \E \big[e^{-t\bZ'} \big] = \E\big[e^{t\bZ} \big] \cdot \E \big[e^{-t\bZ} \big].
    \end{equation*}
    In particular, in \Cref{lem:symmetric_mgf_Z_bound} below, we obtain the following upper bound on the symmetric MGF, valid in the same range $\Abs{t} < n$:
    \begin{equation*}
        \E \big[e^{t(\bZ - \bZ')} \big] \leq \exp \Big( \frac{t^2/n^2}{1-\Abs{t}/n} \cdot \Big( \tr(O^2) + 2n \cdot \tr(O^2\ketbra{u})\Big) \Big).
    \end{equation*}
    Therefore, for such $t$, we have the following chain:
    \begin{equation*}
        \E\big[e^{t(\bX - \E[\bX])} \big] \leq \E\big[e^{t(\bX - \bX')} \big] = \E\big[e^{t(\bY - \bY')} \big] \leq \E\big[e^{t(\bZ - \bZ')} \big] \leq \exp \Big( \frac{t^2/n^2}{1-\Abs{t}/n} \cdot \Big( \tr(O^2) + 2n \cdot \tr(O^2\ketbra{u})\Big) \Big).
    \end{equation*}
    We conclude that $\bX$ is $(v,c)$-sub-gamma, for
    \begin{equation*}
    v = \frac{2\cdot \tr(O^2) + 4n\cdot \tr(O^2 \cdot \ketbra{u})}{n^2},
    \qquad c = \frac{1}{n}.
    \end{equation*}
    This completes the proof for PSD observables.

    To extend to the non-PSD case, decompose $O = O_+ - O_-$, where $O_\pm$ are individually PSD and supported on orthogonal subspaces, and define random variables $\bX_\pm = \tr( O_\pm \cdot \widehat{\bsigma})$. By our previous work, $\bX_\pm$ is $(v_\pm, c)$-sub-gamma, with
    \begin{equation*}
        v_\pm = \frac{2 \cdot \tr\big(O_\pm^2\big) + 4n \cdot \tr\big(O_\pm^2 \cdot \ketbra{u} \big)}{n^2}, \qquad c = \frac{1}{n}.
    \end{equation*}
    Then, by \Cref{item:subgamma_scale} and \Cref{item:subgamma_sum} of \Cref{lem:scalings_sums_subgammas}, $\bX = \bX_+ - \bX_-$ is $\big(2(v_+ +v_-),2c\big)$-sub-gamma, and
    \begin{align*}
        v_+ + v_- & = \frac{2 \cdot \Big(\tr\big(O_+^2\big)+\tr\big(O_-^2\big)\Big) + 4n \cdot \Big(\tr\big(O_+^2 \cdot \ketbra{u} \big) +\tr\big(O_-^2 \cdot \ketbra{u} \big)\Big)}{n^2} \\
        & = \frac{2\cdot \tr(O^2) + 4n\cdot \tr(O^2 \cdot \ketbra{u})}{n^2}.
    \end{align*}
    This completes the proof for non-PSD observables.
\end{proof}

\subsubsection{Proofs of deferred lemmas}

The remainder of this section is devoted to filling in deferred proofs of results used to show \Cref{prop:pure_states_observable_subgamma}. We begin with the following lemma.

\begin{lemma}[$\bZ$ exists, and its symmetric moments dominate $\bY$'s]\label{lem:moments_increase_when_dropping_uparrow}
Let $\bY$ be any real-valued random variable. There exists a real-valued random variable $\bZ$ such that the following two conditions hold. First, for every integer $k \geq 0$ for which $\E[\bY^k]$ exists,
\begin{equation*}
\E\big[\bZ^k\big]
=
\frac{(n+D)^{\uparrow k}}{(n+D)^k}\cdot \E\big[\bY^k\big].
\end{equation*}
Second, for every integer $k \geq 0$ for which $\E[(\bY-\bY')^k]$ exists,
\begin{equation*}
\E\big[(\bY-\bY')^k\big]
\leq
\E\big[(\bZ-\bZ')^k\big].
\end{equation*}
\end{lemma}

\begin{proof}
    Let $\bG \sim \mathrm{Gamma}(\alpha, \theta)$, where $\mathrm{Gamma}(\alpha, \theta)$ is the Gamma distribution with \emph{shape parameter} $\alpha = n+D$, \emph{scale parameter} $\theta = 1/(n+D)$, and PDF
    \begin{equation*}
        f(x; \alpha, \theta) = \frac{1}{\theta^\alpha \Gamma(\alpha)}  x^{\alpha - 1} e^{-x/\theta},
    \end{equation*}
    for $x \geq 0$.
    The $k$-th moment of the Gamma distribution is $\theta^k \cdot \alpha^{\uparrow k}$:
    \begin{equation*}
        \int_{0}^{\infty} f(x;\alpha,\theta) \cdot x^k \cdot dx = \frac{\theta^{\alpha+k} \Gamma(\alpha+k)}{\theta^\alpha \Gamma(\alpha)} \int_0^\infty f(x;\alpha+k,\theta) \cdot dx = \frac{\theta^{\alpha+k} \Gamma(\alpha+k)}{\theta^\alpha \Gamma(\alpha)} = \theta^k  \cdot \alpha^{\uparrow k}.
    \end{equation*}
    Thus,
    \begin{equation} \label{eq:moments_of_G}
        \E \big[ \bG^k \big] = \frac{(n+D)^{\uparrow k}}{(n+D)^k}.
    \end{equation}
    We now define $\bZ \coloneq \bG \cdot \bY$. We have
    \begin{equation*}
        \E \big[ \bZ^k \big] = \E \big[ (\bG \cdot \bY)^k \big] = \E \big[ \bG^k  \big] \cdot \E \big[ \bY^k \big] = \frac{(n+D)^{\uparrow k}}{(n+D)^k}\cdot \E\big[\bY^k\big].
    \end{equation*}
    Now we turn to the symmetrized moments. If $k$ is odd, then
    \begin{equation*}
        \E\big[(\bY-\bY')^k\big] = \E\big[(\bZ-\bZ')^k\big] = 0.
    \end{equation*}
    So, suppose $k$ is even. For fixed $(y,y')$, consider $g_{y,y'}(x,x')\coloneq (xy-x'y')^k$. Since $t\mapsto t^k$ is convex for even $k$, and since $(x,x')\mapsto xy-x'y'$ is linear, the function $g_{y,y'}$ is jointly convex
in $(x,x')$. Hence, by Jensen's inequality,
    \begin{equation*}
        \E \big[ (\bG y - \bG' y')^k \big] = \E\big[ g_{y,y'}(\bG,\bG')\big] \geq g_{y,y'}\big( \E[\bG], \E[\bG']\big) = g_{y,y'}(1,1) = (y-y')^k.
    \end{equation*}
    Here, we have also used the fact that $\E[\bG] = 1$, from \Cref{eq:moments_of_G}. Averaging the above equation over $\bY$ and $\bY'$, we get:
    \begin{equation*}
        \E_{\bZ, \bZ'} \big[ (\bZ - \bZ')^k \big]  = \E_{\bY,\bY'} \Big[ \E_{\bG,\bG'}\big[ (\bG \bY - \bG' \bY')^k \big] \Big] \geq \E_{\bY,\bY'}\big[ (\bY - \bY')^k \big]. \qedhere
    \end{equation*}
\end{proof}

The relative simplicity of the moments of $\bZ$, compared to those of $\bY$, allows us to compute the explicit MGF of $\bZ$.

\begin{lemma}[The MGF of $\bZ$] \label{lem:mgf_Z}
    Suppose $\Abs{t} < n$, and $O$ is positive semidefinite with $\norm{O}_\infty \leq 1$. Then
    \begin{equation*}
        \E \big[ e^{t \bZ} \big] = \det \Big( I - \frac{t}{n} \cdot O \Big)^{-1} \cdot \Big( \bra{u} \cdot \Big( I - \frac{t}{n} \cdot O \Big)^{-1} \cdot \ket{u} \Big)^n.
    \end{equation*}
\end{lemma}

\begin{proof}
    First, we have
    \begin{align}
        \E \big[ e^{t\bZ}\big] & = \sum_{\ell=0}^{\infty} \E \big[ \bZ^\ell \big] \cdot \frac{t^\ell}{\ell!} \nonumber \\
        & = \sum_{\ell=0}^{\infty} \Big(  \frac{\ell!}{n^\ell} \cdot \sum_{m=0}^\ell \binom{n}{m} \cdot \binom{\ell}{m} \cdot T_{\ell m} \Big) \cdot \frac{t^\ell}{\ell!} \nonumber \tag{\Cref{eq:moments_of_Z}} \\
        & = \sum_{m=0}^{\infty} \binom{n}{m} \cdot \sum_{\ell=m}^{\infty} \binom{\ell}{m} \cdot T_{\ell m} \cdot \Big(\frac{t}{n}\Big)^\ell. \label{eq:mgf_Z_1}
    \end{align}
    Here we have used the fact that the $T_{\ell m}$ (defined in \Cref{eq:Tlm_def}) are nonnegative, if $O$ is PSD, to freely rearrange the terms in the infinite sum. This step further requires $t \geq 0$. However, by evaluating the series for nonnegative $t$, we show the series converges absolutely, and then this rearrangement step holds for negative $t$ as well.

    Our next step will be to evaluate the inner sum in \Cref{eq:mgf_Z_1}. To do so, it will be convenient to consider the generating function
    \begin{equation*}
        F(t;\epsilon) \coloneq \sum_{m=0}^\infty \epsilon^m \cdot  \sum_{\ell=m}^{\infty} \binom{\ell}{m} \cdot T_{\ell m} \cdot \Big(\frac{t}{n}\Big)^\ell .
    \end{equation*}
    The desired inner sum is the coefficient of $\epsilon^m$. Rearranging, we obtain
    \begin{equation}
        F(t;\epsilon) = \sum_{\ell =0}^\infty  \Big( \frac{t}{n} \Big)^\ell \cdot \sum_{m = 0}^\ell \epsilon^m \cdot \binom{\ell}{m} \cdot  T_{\ell m}. \label{eq:rearranged_F}
    \end{equation}
    The inner sum of \Cref{eq:rearranged_F} can be evaluated as:
    \begin{equation} \label{eq:inner_inner_sum}
        \sum_{m = 0}^\ell \binom{\ell}{m} \cdot \epsilon^m \cdot T_{\ell m} = \tr\Big(O^{\otimes \ell} \cdot \Pisym^{(\ell, D)}  \cdot \big(I +\epsilon \cdot \ketbra{u}\big)^{\otimes \ell} \cdot \Pisym^{(\ell, D)} \Big).
    \end{equation}
    To show this, we apply the following identity, which we prove below as \Cref{lem:aux_trace_lemma}: if $A$ is a permutationally invariant operator acting on $(\C^D)^{\otimes \ell}$, i.e.\ $A$ commutes with $P(\pi)$ for all permutations $\pi \in S_{\ell}$, and $B$ is any operator on $\C^D$, then for any $\epsilon$,
    \begin{equation*}
        \tr\Big(A \cdot \big( I + \epsilon \cdot B \big)^{\otimes \ell} \Big) = \sum_{m=0}^{\ell} \epsilon^m \cdot \binom{\ell}{m} \cdot \tr\Big( A \cdot \big( B^{\otimes m} \otimes I^{\otimes (\ell - m)} \big)\Big).
    \end{equation*}
    Specifically, we choose $A = \Pisym^{(\ell, D)} \cdot O^{\otimes \ell} \cdot \Pisym^{(\ell,D)}$, and $B = \ketbra{u}$. Here we are using the fact that $\Pisym^{(\ell, D)}$ and $O^{\otimes \ell}$ individually commute with all operators $P(\pi)$ for $\pi \in S_\ell$. This claimed identity then tells us:
    \begin{equation*}
        \tr\Big(\Pisym^{(\ell, D)} \cdot O^{\otimes \ell} \cdot \Pisym^{(\ell, D)}  \cdot \big(I +\epsilon \cdot \ketbra{u}\big)^{\otimes \ell}  \Big) = \sum_{m=0}^\ell \epsilon^m \cdot \binom{\ell}{m} \cdot \tr\Big( \Pisym^{(\ell, D)} \cdot O^{\otimes \ell} \cdot \Pisym^{(\ell, D)} \cdot \big( \ketbra{u}^{\otimes m} \otimes I^{\otimes (\ell - m)} \big) \Big).
    \end{equation*}
    The trace is equal to $T_{\ell m}$, by \Cref{eq:Tlm_def}, and we thus have \Cref{eq:inner_inner_sum}.

    Plugging \Cref{eq:inner_inner_sum} into \Cref{eq:rearranged_F} gives
    \begin{equation*}
        F(t;\epsilon) = \sum_{\ell=0}^\infty \Big( \frac{t}{n} \Big)^\ell \cdot \tr\Big(O^{\otimes \ell} \cdot \Pisym^{(\ell, D)}  \cdot \big(I +\epsilon \cdot \ketbra{u}\big)^{\otimes \ell} \cdot \Pisym^{(\ell, D)} \Big).
    \end{equation*}
    This seems only to have made our sum more complicated. However, after using the fact that $\sqrt{O}^{\otimes \ell}$ commutes with $\Pisym^{(\ell,D)}$ to shuffle the terms of the trace as
    \begin{equation*}
        \tr\Big(O^{\otimes \ell} \cdot \Pisym^{(\ell, D)}  \cdot \big(I +\epsilon \cdot \ketbra{u}\big)^{\otimes \ell} \cdot \Pisym^{(\ell, D)} \Big) = \tr\Big(\Pisym^{(\ell,D)} \cdot  \sqrt{O}^{\otimes \ell} \cdot \big(I +\epsilon \cdot \ketbra{u}\big)^{\otimes \ell} \cdot \sqrt{O}^{\otimes \ell} \Big),
    \end{equation*}
    We apply the following result, whose proof we defer to \Cref{lem:aux_sym_determinant_lemma}: if $A$ is a Hermitian operator on $\C^D$, then for any $z$ such that $\Abs{z} < 1/\norm{A}_\infty$,
        \begin{equation*}
            \sum_{\ell = 0}^\infty z^\ell \cdot \tr \big( \Pisym^{(\ell, D)} \cdot A^{\otimes \ell} \big) = \det( I - z A )^{-1}.
        \end{equation*}
    We apply this result with $A = \sqrt{O} \cdot \big(I +\epsilon \cdot \ketbra{u}\big) \cdot \sqrt{O}$, and $z = t/n$. Since $\norm{A}_\infty \leq 1 + \epsilon$, the lemma's hypotheses are met for all $\epsilon$ sufficiently small such that $\Abs{t}/n < 1/(1+\epsilon)$. For sufficiently small $\epsilon$, we thus have
    \begin{equation*}
        F(t;\epsilon) = \det\Big( I - \frac{t}{n} \cdot \sqrt{O} \cdot \big( I + \epsilon \ketbra{u} \big) \cdot \sqrt{O} \Big)^{-1}.
    \end{equation*}
    Recall that we want to extract from $F(t;\epsilon)$ the coefficient of $\epsilon^m$. To this end, we manipulate $F$ as follows:
    \begin{align*}
        F(t;\epsilon) & = \det\Big( \big( I - \frac{t}{n} \cdot O \big) - \frac{\epsilon t}{n} \cdot \sqrt{O} \ketbra{u} \sqrt{O}  \Big)^{-1} \\
        & = \det \Big( I - \frac{t}{n} \cdot O \Big)^{-1} \cdot \det \Big( I - \frac{\epsilon t}{n} \cdot \big( I - \frac{t}{n} \cdot O\big)^{-1/2} \cdot \sqrt{O}  \cdot \ketbra{u} \cdot \sqrt{O} \cdot \big( I - \frac{t}{n} \cdot O\big)^{-1/2} \Big)^{-1}.
    \end{align*}
    Here we have used the fact that $\Abs{t/n} < 1$ so that $I - \frac{t}{n} \cdot O$ is invertible. If we define
    \begin{equation} \label{eq:def_ketw}
        \ket{w} \coloneq \Big( I - \frac{t}{n} \cdot O\Big)^{-1/2} \cdot \sqrt{O}  \cdot \ket{u},
    \end{equation}
    then we have
    \begin{equation*}
        F(t;\epsilon) = \det \Big( I - \frac{t}{n} \cdot O \Big)^{-1} \cdot \det \Big( I - \frac{\epsilon t}{n} \cdot \ketbra{w} \Big)^{-1} = \det \Big( I - \frac{t}{n} \cdot O \Big)^{-1} \cdot \Big( 1 - \frac{\epsilon t}{n} \cdot \norm{w}^2\Big)^{-1}.
    \end{equation*}
    The second step holds since $\ketbra{w}$ is a rank-$1$ matrix, and hence $I - \frac{\epsilon t}{n} \cdot \ketbra{w}$ has only one non-unit eigenvalue, with corresponding eigenvector $\ket{w}$. We have now isolated the $\epsilon$ dependence into a term which can be expanded as a geometric series. In particular, the coefficient of $\epsilon^m$ in $F(t;\epsilon)$ is
    \begin{equation*}
        F(t;\epsilon) [\epsilon^m] = \det \Big( I - \frac{t}{n} \cdot O \Big)^{-1} \cdot \Big( \frac{t}{n} \Big)^m \cdot \norm{w}^{2m}.
    \end{equation*}
    Substituting this back into \Cref{eq:mgf_Z_1}, we get
    \begin{align}
        \E \big[ e^{t\bZ}\big] & = \sum_{m=0}^{\infty} \binom{n}{m} \cdot F(t;\epsilon) [\epsilon^m] \nonumber \\
        & = \det \Big( I - \frac{t}{n} \cdot O \Big)^{-1} \cdot \sum_{m=0}^{\infty} \binom{n}{m} \cdot \Big(\frac{t}{n}\Big)^{m} \cdot \norm{w}^{2m} \nonumber \\
        & = \det \Big( I - \frac{t}{n} \cdot O \Big)^{-1} \cdot \Big( 1  + \frac{t}{n} \cdot \norm{w}^2 \Big)^n. \label{eq:mgf_Z_2}
    \end{align}
    Finally, we use \Cref{eq:def_ketw} to evaluate the second term in parentheses:
    \begin{align*}
        1  + \frac{t}{n} \cdot \norm{w}^2 & = 1 + \frac{t}{n} \cdot \bra{u} \cdot \sqrt{O} \cdot \Big( 1 - \frac{t}{n} \cdot O \Big)^{-1} \cdot \sqrt{O} \cdot \ket{u} \\
        & = \bra{u} \cdot \Big( 1 - \frac{t}{n} \cdot O \Big)  \cdot \Big( 1 - \frac{t}{n} \cdot O \Big)^{-1}\cdot \ket{u} + \bra{u} \cdot \Big( \frac{t}{n} \cdot O \Big) \cdot \Big( 1 - \frac{t}{n} \cdot O \Big)^{-1} \cdot \ket{u} \\
        & = \bra{u} \cdot \Big( 1 - \frac{t}{n} \cdot O \Big)^{-1} \cdot \ket{u}.
    \end{align*}
    Plugging this into \Cref{eq:mgf_Z_2} completes the proof.
    \end{proof}

We deferred the proofs of two results in the previous argument. We fill these in now.

\begin{lemma}[Computation of a trace] \label{lem:aux_trace_lemma}
    Suppose $A$ is a permutationally invariant operator acting on $(\C^D)^{\otimes \ell}$, i.e.\ $A$ commutes with $P(\pi)$ for all permutations $\pi \in S_{\ell}$. Let $B$ be any operator on $\C^D$. Then for any $\epsilon$,
    \begin{equation*}
        \tr\Big(A \cdot \big( I + \epsilon \cdot B \big)^{\otimes \ell} \Big) = \sum_{m=0}^{\ell} \epsilon^m \cdot \binom{\ell}{m} \cdot \tr\Big( A \cdot \big( B^{\otimes m} \otimes I^{\otimes (\ell - m)} \big)\Big).
    \end{equation*}
\end{lemma}

\begin{proof}
    First, note that
    \begin{equation*}
        \big( I + \epsilon \cdot B \big)^{\otimes \ell} = \sum_{x \in \{0,1\}^\ell} ( \epsilon \cdot B )^{x_1} \otimes \dots \otimes (\epsilon \cdot B )^{x_\ell},
    \end{equation*}
    where we interpret $(\epsilon \cdot B)^{0}$ as $I$. Now, for each summand, there exists a permutation $\pi_x$ that moves all nontrivial factors to the front, i.e.\ such that
    \begin{equation*}
        P(\pi_x) \cdot \Big( ( \epsilon \cdot B )^{x_1} \otimes \dots \otimes (\epsilon \cdot B )^{x_\ell}\Big) \cdot P(\pi_x)^\dagger = \epsilon^{\Abs{x}} \cdot B^{\otimes \Abs{x}} \otimes I^{\otimes(\ell - \Abs{x})},
    \end{equation*}
    where $\Abs{x}$ is the Hamming weight of $x$. Next, since $A$ commutes with all permutations, we have
    \begin{align*}
        \tr \Big(  A \cdot \big(( \epsilon \cdot B )^{x_1} \otimes \dots \otimes (\epsilon \cdot B )^{x_\ell} \big)\Big)
        & = \tr \Big( P(\pi_x)^\dagger \cdot A \cdot P(\pi_x)\cdot \big(( \epsilon \cdot B )^{x_1} \otimes \dots \otimes (\epsilon \cdot B )^{x_\ell}\big) \Big) \\
        & = \tr \Big(  A \cdot P(\pi_x) \cdot \big(( \epsilon \cdot B )^{x_1} \otimes \dots \otimes (\epsilon \cdot B )^{x_\ell}\big)  \cdot P(\pi_x)^\dagger \Big) \\
        & = \epsilon^{\Abs{x}} \cdot \tr \Big(  A \cdot \big(B^{\otimes \Abs{x}} \otimes I^{\otimes(\ell - \Abs{x})}\big) \Big).
    \end{align*}
    Finally,
    \begin{align*}
        \tr\Big(A \cdot \big(I +\epsilon \cdot B\big)^{\otimes \ell}\Big) & = \sum_{x \in \{0,1\}^\ell} \tr \Big(  A \cdot \big(( \epsilon \cdot B )^{x_1} \otimes \dots \otimes (\epsilon \cdot B )^{x_\ell} \big)\Big) \\
        & = \sum_{x \in \{0,1\}^\ell}\epsilon^{\Abs{x}} \cdot \tr \Big(  A \cdot \big(B^{\otimes \Abs{x}} \otimes I^{\otimes(\ell - \Abs{x})}\big) \Big) \\
        & = \sum_{m=0}^\ell \epsilon^{m} \cdot \binom{\ell}{m} \cdot \tr\Big(  A \cdot \big(B^{\otimes m} \otimes I^{\otimes(\ell - m)}\big) \Big).
    \end{align*}
    In the last step, we have grouped together all terms of the same Hamming weight.
\end{proof}

    \begin{lemma}[Computation of a power series]\label{lem:aux_sym_determinant_lemma}
        Suppose $A$ is a Hermitian operator on $\C^D$. Then for any $z$ such that $\Abs{z} < 1/\norm{A}_\infty$,
        \begin{equation*}
            \sum_{\ell = 0}^\infty z^\ell \cdot \tr \big( \Pisym^{(\ell, D)} \cdot A^{\otimes \ell} \big) = \det( I - z A )^{-1}.
        \end{equation*}
    \end{lemma}
    \begin{proof}
        Let $A$ have eigenvalues $\lambda_1, \dots, \lambda_D$. We can assume $A$ is diagonal in the computational basis, because if $A = U \cdot \Lambda \cdot U^\dagger$, with $\Lambda = \mathrm{diag}(\lambda_1, \dots, \lambda_D)$, then $\det(I - zA) = \det(U \cdot ( I - z \Lambda) U^\dagger ) = \det(I - z \Lambda)$ and \begin{equation*}
            \tr \big( \Pisym^{(\ell, D)} \cdot A^{\otimes \ell} \big) = \tr \big( \Pisym^{(\ell, D)} \cdot U^{\otimes \ell} \cdot \Lambda^{\otimes \ell} \cdot (U^\dagger)^{\otimes \ell} \big) = \tr \big( \Pisym^{(\ell, D)} \cdot \Lambda^{\otimes \ell} \big),
        \end{equation*}
         since $\Pisym^{(\ell, D)}$ commutes with $U^{\otimes \ell}$. Moreover, since the type vectors form an orthonormal basis of the symmetric subspace, we can write $\Pisym^{(\ell,D)} = \sum_{\tau} \ketbra{\tau}$, where $\tau$ ranges over all types $(\tau_1, \dots, \tau_D)$, such that $\tau_i \geq 0$ and $\mathrm{wt}(\tau) = \sum_i \tau_i = \ell$, and where $\ket{\tau}$ is the corresponding normalized, uniform superposition over all computational basis vectors $\ket{x}$ of type $\tau$. Therefore,
         \begin{equation*}
             \tr \big( \Pisym^{(\ell, D)} \cdot A^{\otimes \ell}  \big)  = \sum_{\mathrm{wt}(\tau) = \ell} \bra{\tau} \cdot A^{\otimes \ell} \cdot \ket{\tau} = \sum_{\mathrm{wt}(\tau) = \ell} \lambda_1^{\tau_1} \dots \lambda_D^{\tau_D}.
         \end{equation*}
         Thus,
         \begin{align}
             \sum_{\ell=0}^\infty z^\ell \cdot \tr \big( \Pisym^{(\ell, D)} \cdot A^{\otimes \ell} \big) & = \sum_{\ell = 0}^\infty z^{\ell} \cdot \sum_{\mathrm{wt}(\tau) = \ell} \lambda_1^{\tau_1} \dots \lambda_D^{\tau_D} \nonumber \\
             & = \sum_{\mathrm{wt}(\tau)\geq 0} (z\lambda_1)^{\tau_1} \dots (z\lambda_D)^{\tau_D} \nonumber \\
             & = \Big(\sum_{\tau_1=0}^\infty (z\lambda_1)^{\tau_1}\Big) \dots \Big(\sum_{\tau_D=0}^\infty (z\lambda_D)^{\tau_D}\Big) \nonumber \\
             & = \prod_{i=1}^{D} \frac{1}{1 - z \lambda_i}. \label{eq:sym_subspace_trace_identity}
         \end{align}
         Finally, since $I-zA$ has eigenvalues $1 - z \lambda_i$, for $i \in [D]$, and since $\Abs{z \lambda_i} < 1$ by our hypotheses, $(I-zA)^{-1}$ has eigenvalues $(1 - z\lambda_i)^{-1}$, for $i \in [D]$. Comparing with \Cref{eq:sym_subspace_trace_identity} yields
         \begin{equation*}
             \sum_{\ell = 0}^\infty z^\ell \cdot \tr \big( \Pisym^{(\ell, D)} \cdot A^{\otimes \ell} \big) = \det( I - z A )^{-1}. \qedhere
         \end{equation*}
    \end{proof}

Finally, with the explicit MGF of $\bZ$ in hand, we are able to bound the symmetrized MGF of $\bZ$. Since the symmetrized MGF of $\bZ$ upper bounds the centered MGF of $\bX$, this allows us to conclude $\bX$ is sub-gamma.

\begin{lemma}[A bound on the symmetrized MGF of $\bZ$] \label{lem:symmetric_mgf_Z_bound}
    Suppose $\Abs{t} < n$, and $O$ is positive semidefinite with $\norm{O}_\infty \leq 1$. Then
    \begin{equation*}
        \E \big[e^{t(\bZ - \bZ')} \big] \leq \exp \Big( \frac{t^2/n^2}{1-\Abs{t}/n} \cdot \Big(  \tr(O^2) + 2n \cdot \tr(O^2\ketbra{u})\Big) \Big).
    \end{equation*}
\end{lemma}

\begin{proof}
    First, we have
    \begin{equation*}
        \E \big[e^{t(\bZ - \bZ')} \big] = \E \big[e^{t\bZ} \big] \cdot \E \big[e^{-t\bZ'} \big] = \E\big[e^{t\bZ} \big] \cdot \E \big[e^{-t\bZ} \big].
    \end{equation*}
    Thus,
    \begin{equation}
        \log \E \big[e^{t(\bZ - \bZ')} \big] = \log \E\big[e^{t\bZ} \big] + \log \E \big[e^{-t\bZ} \big]. \label{eq:log_of_MGF}
    \end{equation}
    We use our expression for the MGF of $\bZ$ to bound the right-hand side. By \Cref{lem:mgf_Z},
    \begin{equation*}
        \log \E\big[e^{\pm t\bZ} \big] = - \log \Big(\det \Big( I \mp \frac{t}{n} \cdot O \Big)\Big) + n \cdot \log \Big( \bra{u} \cdot \Big( 1 \mp \frac{t}{n} \cdot O \Big)^{-1} \cdot \ket{u} \Big) \eqcolon L_{\pm} + R_{\pm},
    \end{equation*}
    where we define $L_\pm$ as the first term in the center, and $R_{\pm}$ as the second.
    Suppose $O$ has eigenvectors $\ket{w_1}, \dots, \ket{w_D}$ with corresponding eigenvalues $\lambda_1, \dots, \lambda_D$, which satisfy $\Abs{\lambda_i} \leq 1$ for each $i$. From
    \begin{equation*}
        L_{\pm} = - \log \Big(\det \Big( I \mp \frac{t}{n} \cdot O \Big)\Big) = - \sum_{i=1}^{D} \log \Big( 1 \mp \frac{t}{n} \cdot \lambda_i \Big),
    \end{equation*}
    we have
    \begin{equation} \label{eq:bound_on_Ls}
        L_+ + L_- = - \sum_{i=1}^{D} \log \Big( 1 - \frac{t^2}{n^2} \cdot \lambda_i^2 \Big) \leq \sum_{i=1}^{D} \frac{\frac{t^2}{n^2} \cdot \lambda_i^2}{1 - \frac{t^2}{n^2} \cdot \lambda_i^2} \leq \sum_{i=1}^{D} \frac{\frac{t^2}{n^2} \cdot \lambda_i^2}{1 - \frac{t^2}{n^2}} \leq \frac{ \frac{t^2}{n^2}}{1 - \frac{\Abs{t}}{n}} \cdot \sum_{i=1}^D \lambda_i^2 = \frac{\frac{t^2}{n^2} }{1 - \frac{\Abs{t}}{n}} \cdot \tr\big( O^2 \big).
    \end{equation}
    Above, the first inequality is $-\log (1 - x) \leq \frac{x}{1-x}$ for $\Abs{x} < 1$, choosing $x = \big(\frac{t}{n} \cdot \lambda_i\big)^2$. The inequality holds because
    \begin{equation*}
        -\log(1-x) = \sum_{k=1}^\infty \frac{1}{k} \cdot x^k \leq \sum_{k=1}^\infty x^k = \frac{x}{1-x}.
    \end{equation*}
    For the other terms, $R_\pm$, we first note
    \begin{align}
        \bra{u} \cdot \Big( 1 \mp \frac{t}{n} \cdot O \Big)^{-1} \cdot \ket{u} = \sum_{i=1}^{D} \frac{1}{1 \mp \frac{t}{n} \cdot \lambda_i} \cdot \ABs{\braket{u}{w_i}}^2 & \leq \sum_{i=1}^{D} \Big( 1 \pm \frac{t}{n} \cdot \lambda_i + \frac{\frac{t^2}{n^2} \cdot \lambda_i^2}{1 - \frac{\Abs{t}}{n}} \Big) \cdot \ABs{\braket{u}{w_i}}^2 \nonumber \\
        & = 1 \pm \frac{t}{n} \cdot \tr\big(O \cdot \ketbra{u} \big) + \frac{\frac{t^2}{n^2} }{1 - \frac{\Abs{t}}{n}} \cdot \tr\big(O^2 \cdot \ketbra{u}\big).\label{eq:inner_product_eq_1}
    \end{align}
    Here, we have used the inequality
    \begin{equation*}
        \frac{1}{1-xy} = \sum_{k=0}^\infty x^k y^k \leq 1 + xy + x^2y^2 \cdot \sum_{k=0}^{\infty} \Abs{y}^k = 1 + xy + \frac{x^2y^2}{1-\Abs{y}},
    \end{equation*}
    which holds for all $\Abs{x} \leq 1$ and $\Abs{y} < 1$, taking $x = \pm \lambda_i$ and $y = \frac{t}{n}$. Then we have
    \begin{align}
        R_+ + R_- & = n \cdot \log \Big( \Big(\bra{u} \cdot \Big( 1 - \frac{t}{n} \cdot O \Big)^{-1} \cdot \ket{u} \Big) \times \Big( \bra{u} \cdot \Big( 1 + \frac{t}{n} \cdot O \Big)^{-1} \cdot \ket{u} \Big)\Big) \nonumber \\
        & \leq n \cdot \log \Big( \Big( 1 + \frac{\frac{t^2}{n^2} }{1 - \frac{\Abs{t}}{n}} \cdot \tr\big(O^2 \cdot \ketbra{u}\big)\Big)^2 - \Big( \frac{t}{n} \cdot \tr\big( O \cdot \ketbra{u} \big) \Big)^2  \Big) \tag{\Cref{eq:inner_product_eq_1}} \nonumber \\
        & \leq 2n \cdot \log \Big(  1 + \frac{\frac{t^2}{n^2} }{1 - \frac{\Abs{t}}{n}} \cdot \tr\big(O^2 \cdot \ketbra{u}\big)\Big) \nonumber \\
        & \leq 2n \cdot \frac{\frac{t^2}{n^2} }{1 - \frac{\Abs{t}}{n}} \cdot \tr\big(O^2 \cdot \ketbra{u}\big). \label{eq:bound_on_Rs}
    \end{align}
    Using the bounds from \Cref{eq:bound_on_Ls,eq:bound_on_Rs} in \Cref{eq:log_of_MGF}, we obtain
    \begin{equation*}
        \log \E \big[ e^{t(\bZ - \bZ')}\big] \leq \frac{ \frac{t^2}{n^2}}{1 - \frac{\Abs{t}}{n}} \cdot \Big( \tr \big( O^2 \big) + 2n \cdot \tr\big( O^2 \cdot \ketbra{u}\big) \Big).
        \end{equation*}
        This completes the proof.
\end{proof}

\subsection{Observable concentration for mixed states}
\label{sec:obs-conc-mixed-states}

Sub-gamma observable concentration for mixed states follows readily from our pure state result.

\begin{proposition} [The $\mixed(\gps)$-based plug-in estimator is sub-gamma] \label{prop:mixed_states_observables_subgamma}
    Let $\rho \in \C^{d \times d}$ be a rank-$r$ mixed state, and let $\widehat{\brho}$ be the output of $\mixed(\gps)$ given $n$ copies of $\rho$. Let $O \in \C^{d \times d}$ be a positive semidefinite operator such that $\norm{O}_\infty \leq 1$. Then the random variable
    $\bX \coloneq \tr( O \cdot \widehat{\brho})$
    is $(v,c)$-sub-gamma, with
    \begin{equation*}
    v = \frac{2r \cdot \tr(O^2) + 4n\cdot \tr(O^2 \cdot \rho)}{n^2},
    \qquad c = \frac{1}{n}.
    \end{equation*}
    More generally, if $O \in \C^{d \times d}$ is any observable such that $\norm{O}_\infty \leq 1$, then $\bX$ is $(2v,2c)$-sub-gamma.
\end{proposition}

\begin{proof}
    Recall that $\mixed(\gps)$ first applies the random purification channel to $\rho^{\otimes n}$, obtaining $n$ copies of a random purification $\ket{\brho} \in \C^{D}$, for $D = d \cdot r$. Then, it performs Hayashi's measurement on $\ket{\brho}^{\otimes n}$, which produces a pure state $\ket{\bu}$; forms $\widehat{\bsigma} = \frac{n+D}{n} \cdot \ketbra{\bu} - \frac{1}{n} \cdot I_D$; and outputs $\widehat{\brho} \coloneq \tr_{\reg{2}}(\widehat{\bsigma})$.

    Note that we have
    \begin{equation*}
        \bX = \tr\big(O \cdot \widehat{\brho}\big) = \tr \big( O \cdot \tr_{\reg{2}}(\widehat{\bsigma})\big) = \tr\big( O \otimes I_r \cdot \widehat{\bsigma}\big).
    \end{equation*}
    First assume $O$ is PSD, so that $O \otimes I_r$ is as well.
    We may view $\bX$ as arising in two stages: first sample the random purification $\ket{\brho}$, and then, conditioned on $\ket{\brho}$, output
        $\bX_{\ket{\brho}} \coloneq \tr\big((O\otimes I_r) \cdot \widehat{\bsigma}\big),$
    where $\widehat{\bsigma}$ is the estimator produced by $\gps$ from $n$ copies of the pure state $\ket{\brho}$.
    Since $\norm{O \otimes I_r}_\infty = \norm{O}_\infty \leq 1$, \Cref{prop:pure_states_observable_subgamma} tells us that $\bX_{\ket{\brho}}$ is $(v_{\ket{\brho}},c)$-sub-gamma, with parameters
    \begin{equation*}
        v_{\ket{\brho}} = \frac{2 \cdot \tr\big( ( O \otimes I_r)^2 \big) + 4n \cdot \tr\big( (O \otimes I_r)^2 \cdot \ketbra{\brho}\big)}{n^2}
    \end{equation*}
    and $c = 1/n$. However, since
    \begin{equation*}
        \tr\big( ( O \otimes I_r)^2 \big) = \tr \big( O^2 \otimes I_r \big) =  r \cdot \tr\big( O^2 \big),
    \end{equation*}
    and
    \begin{equation*}
        \tr\big( (O \otimes I_r)^2 \cdot \ketbra{\brho}\big) = \tr\big( O^2 \otimes I_r \cdot \ketbra{\brho}\big) = \tr \big( O^2 \cdot \tr_{\reg{2}}(\ketbra{\brho})\big) = \tr \big( O^2 \cdot \rho \big),
    \end{equation*}
    we have
    \begin{equation*}
        v_{\ket{\brho}} = \frac{2r \cdot \tr(O^2) + 4n\cdot \tr(O^2 \cdot \rho)}{n^2} = v,
    \end{equation*}
    independent of $\ket{\brho}$. Moreover,
    \begin{equation*}
    \E \big[ \bX_{\ket{\brho}}\big] = \tr \big( O \cdot \E \big[ \widehat{\brho} \, | \, \ket{\brho}\big] \big) = \tr \big( O \cdot \tr_{\reg{2}} \big( \E \big[ \widehat{\bsigma}\, | \, \ket{\brho}\big] \big) \big) = \tr \big( O \cdot \tr_{\reg{2}} \big( \ketbra{\brho} \big) \big) = \tr \big( O \cdot \rho \big).
\end{equation*}
So, the $\bX_{\ket{\brho}}$ have common  mean, and are all $(v,c)$-sub-gamma. As a result, \Cref{lem:mixture_of_subgammas} tells us that $\bX$ itself is $(v,c)$-sub-gamma.

For $O$ not necessarily PSD, \Cref{prop:pure_states_observable_subgamma} tells us instead that $\bX_{\ket{\brho}}$ is still $(2v_{\ket{\brho}},2c)$-sub-gamma, and the result follows by an identical argument.
\end{proof}

\subsection{Relative error bounds for rank-one observables}
\label{sec:refined-uniform-via-observable-conc}

Now, we apply the sub-gamma observable concentration to show our main tool.

\begin{theorem}[A relative error bound, \cref{thm:main-tool} restated]
    \label{thm:uniform-bound-pointwise}
    Let $\rho \in \C^{d \times d}$ be a mixed state, and let $\widehat{\brho}$ be the output of $\mixed(\gps)$, given $n$ copies of $\rho$. With probability at least $0.99$, for all pure states $\ket{w}$,
    \begin{equation*}
        \ABs{\bra{w} ( \widehat{\brho} - \rho )  \ket{w}} \leq C \cdot \sqrt{ \frac{d}{n} \cdot \Big( \bra{w} \rho \ket{w} + \frac{d}{n} \Big) }.
    \end{equation*}
    Here, $C$ is a universal constant, which we can take to be at most $432$.
\end{theorem}

\begin{remark} \label{rem:C_lower_bound}
    We will occasionally want a lower bound on $C$. In such situations, we assume $C \geq 1$. 
\end{remark}

Before diving into the proof, we describe a useful construction. Let $R \in \C^{d \times d}$ be a matrix. Then the \emph{Sylvester superoperator} $\calL_R: \C^{d \times d} \to \C^{d \times d}$ is the superoperator defined via
\begin{equation*}
        \calL_R(M) \coloneq \frac{RM + M R}{2}.
\end{equation*}
In the case we will care about, $R$ will be diagonal in the standard basis, PSD, and invertible. In this setting, $\calL_R$ has an inverse superoperator. If we write $R_i$ for the $i$-th eigenvalue of $R$, then
\begin{equation*}
        \calL_R^{-1}(M) = \sum_{i,j=1}^d \frac{2}{R_i + R_j} \cdot M_{ij} \cdot \ketbra{i}{j}.
\end{equation*}
It is straightforward to check that $\calL_R \circ \calL_R^{-1}(M) = \calL_R^{-1} \circ \calL_R(M) = M$. For example,
\begin{align*}
    \calL_R \circ \calL_R^{-1} (M)&  = \sum_{i,j=1}^d \frac{2}{R_{i}+R_j} \cdot M_{ij} \cdot \Bigg(\frac{R \ketbra{i}{j}  + \ketbra{i}{j} R}{2} \Bigg) \\
    & = \sum_{i,j=1}^d \frac{2}{R_{i}+R_j} \cdot M_{ij} \cdot \Bigg(\frac{R_i + R_j}{2} \cdot \ketbra{i}{j}\Bigg) = \sum_{i,j=1}^d M_{ij} \cdot \ketbra{i}{j} = M.
\end{align*}

\begin{proof}[Proof of \Cref{thm:uniform-bound-pointwise}]
We assume $\rho$ is diagonal in the computational basis, and write $\rho = \sum_{i=1}^d \alpha_i \cdot \ketbra{i}$. This is without loss of generality: we can otherwise view, for example, $\ket{i}$ as shorthand for the $i$-th eigenvector of $\rho$. We then make a particular choice of $R$:
\begin{equation*}
        R \coloneq \sqrt{\frac{n}{d}} \cdot \sum_{i=1}^{d} \sqrt{ \alpha_i + \frac{d}{n}} \cdot \ketbra{i}.
\end{equation*}
Now, let $\ket{w} \in \C^{d}$ be a pure state, and write $\ket{w} = \sum_{i=1}^d w_i \cdot \ket{i}$. Note the following:
\begin{align*}
        \bra{w}(\widehat{\brho} - \rho) \ket{w} & = \bra{w} \cdot \calL_R \circ \calL_R^{-1} ( \widehat{\brho}- \rho)  \cdot \ket{w} \\
        & = \frac{1}{2} \cdot \Big( \bra{w} \cdot R \cdot \calL_R^{-1}( \widehat{\brho}- \rho) \cdot \ket{w}  +\bra{w} \cdot  \calL_R^{-1}( \widehat{\brho}- \rho) \cdot R \cdot \ket{w}\Big).
    \end{align*}
    Therefore,
    \begin{equation}\label{eq:first_bound_w_dependent}
        \ABs{\bra{w} \cdot (\widehat{\brho} - \rho) \cdot \ket{w}} \leq \norm{\calL_R^{-1}(\widehat{\brho} - \rho)}_\infty \cdot \norm{ \ket{w}} \cdot \norm{ R \ket{w}}.
    \end{equation}
    We have $\norm{ \ket{w}} = 1$, and
    \begin{equation*}
        \norm{ R \ket{w}}^2 = \bra{w} R^2 \ket{w}  = \frac{n}{d}  \cdot \sum_{i=1}^d \Big( \alpha_i + \frac{d}{n} \Big) \cdot \Abs{w_i}^2 =  \frac{n}{d}  \cdot \Big( \bra{w} \rho \ket{w} + \frac{d}{n} \Big).
    \end{equation*}
    So, substituting back into \Cref{eq:first_bound_w_dependent},
    \begin{equation}\label{eq:bound_without_infinity_norm_result}
    \ABs{\bra{w} \cdot (\widehat{\brho} - \rho) \cdot \ket{w}} \leq \norm{\calL_R^{-1}(\widehat{\brho} - \rho)}_\infty \cdot \sqrt{ \frac{n}{d} \cdot \Big( \bra{w} \rho \ket{w} + \frac{d}{n} \Big)}.
    \end{equation}
    This is already close to what we want to show. Indeed, it now suffices to show that with high probability,
    \begin{equation*}
        \norm{\calL_R^{-1}(\widehat{\brho} - \rho)}_\infty \leq C \cdot \frac{d}{n}.
    \end{equation*}

    To show this infinity-norm bound, it will be useful to observe that $\bra{w} \cdot \calL_R^{-1} (M) \cdot \ket{w}$ can be reinterpreted as an expectation value. Specifically,
    \begin{equation*}
        \bra{w} \cdot \calL_R^{-1}(M) \cdot \ket{w} = \sum_{ij} \frac{2}{R_i + R_j} \cdot M_{ij} \cdot w_i^\dagger w_j = \tr \Big( \Big(\sum_{ij} \frac{2}{R_i+R_j} \cdot w_j \ketbra{j}{i} w_i^\dagger\Big) \cdot M \Big).
    \end{equation*}
    So, if we define an observable
    \begin{equation*}
        O_w \coloneq \sum_{ij} \frac{2}{R_i+R_j} \cdot w_j \ketbra{j}{i} w_i^\dagger,
    \end{equation*}
    we have
    \begin{equation}\label{eq:inner_product_to_trace}
        \bra{w} \cdot \calL_R^{-1}(M) \cdot \ket{w} = \tr( O_w \cdot M).
    \end{equation}
    For us, this means
    \begin{equation*}
        \bra{w} \cdot \calL_R^{-1}(\widehat{\brho} - \rho) \cdot \ket{w} = \tr( O_w \cdot (\widehat{\brho} - \rho) ) = \bX - \E[\bX],
    \end{equation*}
    for $\bX = \tr(O_w \cdot \widehat{\brho})$. Moreover, we have $\norm{O_w}_\infty \leq 1$. To show this, note
    $\norm{O_w}_\infty \leq \sqrt{\tr( O_w^2 )}$ and
    \begin{equation*}
        \tr\big( O_w^2 \big) = \tr(O_w^\dagger O_w) = \sum_{ij} \Abs{O_{ij}}^2 = \sum_{ij} \frac{4}{(R_i+R_j)^2} \cdot \Abs{w_i}^2 \cdot \Abs{w_j}^2 \leq \sum_{ij} \Abs{w_i}^2 \cdot \Abs{w_j}^2 = 1.
    \end{equation*}
    The inequality uses $R_i \geq 1$ for each $i$, so that $(R_i + R_j)^2 \geq 4$. Thus, $\norm{O_w}_\infty \leq 1$.

    Using the general observable case\footnote{In fact, we could use the PSD case here, but the non-PSD case is readily available without further justification, and loses only constant factors. For completeness, one way to see that $O_w$ is PSD is to note that $O_w$ is the Hadamard product of the matrix $\ketbra{w}$ and the matrix $A$ with $(i,j)$ entry $A_{ij} = 2/(R_i+R_j)$. The matrix $A$ is PSD, since it is the Gram matrix of the functions $\{f_i\}_{i \in [d]}$, with $f_i \in L^2(\R)$ given by $f_i(x) = e^{-R_i\Abs{x}}$:
    \begin{equation*}
        \langle f_i, f_j \rangle = \int_{-\infty}^{\infty} f_i(x)^* f_j(x) dx = \int_{-\infty}^{\infty} e^{-(R_i+R_j) \Abs{x}} dx = \frac{2}{R_i+R_j}.
    \end{equation*}
    Since $O_w$ is a Hadamard product of two PSD matrices, the Schur product theorem tells us it itself is PSD.} of \Cref{prop:mixed_states_observables_subgamma}, we have that $\bX$ is $(v,c)$-sub-gamma, for
    \begin{equation*}
        v = \frac{4d \cdot \tr\big( O_w^2 \big) + 8n \cdot \tr \big( O_w^2 \cdot \rho \big)}{n^2}, \qquad c = \frac{2}{n}.
    \end{equation*}
    We know already that $\tr(O_w^2) \leq 1$.
    To simplify our expression for $v$, we also need
    \begin{align*}
        \tr\big( O_w^2 \cdot \rho\big) & = \tr\Big(\sum_{ijk} \frac{4}{(R_i+R_j) \cdot (R_j + R_k)} \cdot w_i \ketbra{i}{j} w_j^\dagger \cdot w_j \ketbra{j}{k} w_k^\dagger \cdot \alpha_k \ketbra{k} \Big) \\
        & = \sum_{ij} \frac{4 \alpha_i}{(R_i+R_j)^2} \cdot \Abs{w_i}^2 \cdot \Abs{w_j}^2.
    \end{align*}
    Since $(R_i +R_j)^2 \geq R_i^2  = 1 + \frac{n}{d} \cdot  \alpha_i \geq \frac{n}{d} \cdot \alpha_i$, we have
    \begin{equation*}
        \tr \big( O_w^2 \cdot \rho \big) \leq \sum_{ij} \frac{4d}{n} \cdot \Abs{w_i}^2 \cdot \Abs{w_j}^2 = \frac{4d}{n}.
    \end{equation*}
    Thus,
    \begin{equation*}
        v \leq \frac{4d \cdot 1 + 8n \cdot (4d/n)}{n^2} =  \frac{36d}{n^2} \eqcolon v'.
    \end{equation*}
    In particular, $\bX$ is $(v',c)$-sub-gamma.

    By the standard tail bound for sub-gamma random variables, \Cref{lem:subgamma_probability_bound}, we have for any $\delta > 0$,
    \begin{equation*}
        \Pr \Big[ \ABs{\bX - \E[\bX]} \geq \sqrt{2v' \delta} + c \delta \Big] \leq 2e^{-\delta}
    \end{equation*}
    If we take $\delta = C_1 \cdot d$, for some constant $C_1 > 0$, we have
    \begin{equation*}
        \Pr \Big[ \ABs{\bX-\E[\bX]} \geq \big(6\sqrt{2C_1} + 2C_1\big) \cdot \frac{d}{n} \Big] \leq 2e^{-C_1  d}.
    \end{equation*}
    In particular, it will be convenient to take $C_1 = 72$, so that $C_1 = 6\sqrt{2C_1}$, and we can simplify the probability bound to
    \begin{equation*}
        \Pr \Big[ \ABs{\bX-\E[\bX]} \geq 3C_1 \cdot \frac{d}{n} \Big] \leq 2e^{-C_1 d}.
    \end{equation*}
    Since $\bX -\E[\bX]= \tr( O_w \cdot (\widehat{\brho} - \rho))$, by \Cref{eq:inner_product_to_trace}, we have
    \begin{equation} \label{eq:prob_for_fixed_w}
        \Pr \Big[\ABs{\bra{w} \cdot \calL_R^{-1}(\widehat{\brho}-\rho) \cdot \ket{w}} \geq 3C_1 \cdot \frac{d}{n} \Big] \leq 2e^{-C_1 d}.
    \end{equation}
    This bound holds for any fixed pure state $\ket{w}$.

    We proceed by taking a net over pure states and using it to promote \Cref{eq:prob_for_fixed_w} to a uniform bound that holds with high probability. In particular, let $\calN$ be a net of mesh $\theta$ over the pure states. Identifying a pure state with a point on the real Euclidean sphere $\mathbb{S}^{2d-1}$, we have, for $\theta \in (0,1/2)$,
    \begin{equation*}
        \norm{\calL_R^{-1}(\widehat{\brho}-\rho)}_\infty \leq \frac{1}{1-2\theta} \cdot \max_{\ket{w} \in \calN} \ABs{\bra{w} \cdot \calL_R^{-1}(\widehat{\brho}-\rho) \cdot \ket{w}},
    \end{equation*}
    by \cite[Lemma 4.4.2]{Ver18}. Furthermore, \cite[Corollary 4.2.11]{Ver18} tells us that there exists such an $\calN$ with size at most $(1 + \frac{2}{\theta})^{2d}$. Taking $\theta = 1/4$ gives us a net of size $\Abs{\calN} \leq 9^{2d} \leq e^{5d}$, and sets $1/(1-2\theta) = 2$. Then, applying \Cref{eq:prob_for_fixed_w} for each $\ket{w} \in \calN$, a union bound gives us
    \begin{equation*}
        \Pr \Big[ \norm{\calL_R^{-1}(\widehat{\brho}-\rho)}_\infty \geq  6C_1 \cdot \frac{d}{n} \Big] \leq \sum_{\ket{w} \in \calN} \Pr \Big[\ABs{\bra{w}  \calL_R^{-1}(\widehat{\brho}-\rho) \ket{w}} \geq 3C_1 \cdot \frac{d}{n} \Big] \leq 2\Abs{\calN} \cdot e^{-C_1 d} \leq 2e^{(5-C_1) d} \leq 0.01.
    \end{equation*}
    Thus $\norm{\calL_R^{-1}(\widehat{\brho}-\rho)}_\infty \leq  6C_1 \cdot \frac{d}{n}$ with high probability. Substituting back into \Cref{eq:bound_without_infinity_norm_result} and taking $C \coloneq 6 C_1$ completes the proof.
    \end{proof}

    Before generalizing this bound to higher-rank observables, we state and prove two corollaries that are used in \Cref{sec:additional_tomography_results} and which illustrate consequences of the relative error bound. 

    First, we prove a bound that depends on $\bra{w} \widehat{\brho} \ket{w}$ instead of $\bra{w} \rho \ket{w}$. This is useful, for example, whenever we want to apply a relative-error-type bound to a vector $\ket{w}$ which lies in a subspace that depends on $\widehat{\brho}$, e.g.\ the span of certain eigenvectors of $\widehat{\brho}$.

\begin{corollary}[A relative error bound scaling with $\bra{w} \widehat{\brho}\ket{w}$]
    \label{cor:uniform-bound-pointwise_rho-hat-version}
    Let $\rho \in \C^{d \times d}$ be a mixed state, and let $\widehat{\brho}$ be the output of $\mixed(\gps)$, given $n$ copies of $\rho$. With probability at least $0.99$,
    \begin{equation*}
        \ABs{\bra{w} ( \widehat{\brho} - \rho )  \ket{w}} \leq C' \cdot \sqrt{ \frac{d}{n} \cdot \Big( \bra{w} \widehat{\brho} \ket{w} + \frac{d}{n} \Big) },
    \end{equation*}
    for all pure states $\ket{w}$, where $C'$ is a universal constant we may take to be at most $C \cdot \sqrt{C^2+2}$.
\end{corollary}

\begin{proof}
    We condition on the relative error bound holding. In this case, we have
    \begin{align*}
        \bra{w} \rho \ket{w} & \leq \bra{w} \widehat{\brho} \ket{w} + \ABs{\bra{w} (\widehat{\brho} - \rho) \ket{w}} \\
        & \leq \bra{w} \widehat{\brho} \ket{w} + C \cdot \sqrt{ \frac{d}{n} \cdot \Big( \bra{w} \rho \ket{w} + \frac{d}{n} \Big)} \\
        & \leq \bra{w} \widehat{\brho} \ket{w} + \bigg( \frac{C^2}{2} \cdot \frac{d}{n} + \frac{1}{2} \cdot \Big(\bra{w} \rho \ket{w} + \frac{d}{n} \Big) \bigg).
    \end{align*}
    The last step is AM-GM. Rearranging this inequality yields
    \begin{equation*}
        \bra{w} \rho \ket{w} \leq 2 \cdot \bra{w} \widehat{\brho} \ket{w} +  (C^2+1) \cdot \frac{d}{n}.
    \end{equation*}
    Thus,
    \begin{align*}
        \ABs{\bra{w} ( \widehat{\brho} - \rho )  \ket{w}} & \leq C \cdot \sqrt{ \frac{d}{n} \cdot \Big( \bra{w} {\rho} \ket{w} + \frac{d}{n} \Big) } \\
        & \leq C \cdot \sqrt{ \frac{d}{n} \cdot \Big( 2 \cdot \bra{w} \widehat{\brho} \ket{w} + (C^2+2) \cdot \frac{d}{n} \Big) } \\
        & \leq C' \cdot \sqrt{ \frac{d}{n} \cdot \Big( \bra{w} \widehat{\brho} \ket{w} +\frac{d}{n} \Big)}.
    \end{align*}
    This completes the proof.
\end{proof}

    Next, we prove a statement about differences between the eigenvalues of $\rho$ and $\widehat{\brho}$. We need only a one-sided inequality, but similar ideas yield two-sided versions.

   \begin{corollary}[A one-sided relative error bound for eigenvalues]
    \label{cor:eigenvalue_differences}
    Let $\rho \in \C^{d \times d}$ be a mixed state, and let $\widehat{\brho}$ be the output of $\mixed(\gps)$, given $n$ copies of $\rho$. With probability at least $0.99$, for all $i \in [d]$,
    \begin{equation*}
        \lambda_i(\widehat{\brho}) \leq \lambda_i(\rho)  +  C \cdot \sqrt{ \frac{d}{n} \cdot \Big( \lambda_i(\rho) + \frac{d}{n} \Big)}.
    \end{equation*}
\end{corollary}

\begin{proof}
    We condition on the relative error bound holding. Let $\{\ket{v_i}\}_i$ be the eigenvectors of $\rho$, ordered by decreasing eigenvalues. Fix an $i$, and suppose $\ket{v} \in V_{\geq i} \coloneq \mathrm{span}( \ket{v_i}, \dots, \ket{v_d} )$. By the relative error bound, we have
    \begin{equation*}
        \bra{v} \widehat{\brho} \ket{v}  \leq \bra{v} \rho \ket{v} + C \sqrt{ \frac{d}{n} \cdot \Big( \bra{v} \rho \ket{v} + \frac{d}{n} \Big) } \leq \lambda_i(\rho) + C \sqrt{ \frac{d}{n} \cdot \Big( \lambda_i(\rho) + \frac{d}{n} \Big) }.
    \end{equation*}
    This holds for all such $\ket{v}$, so by the Courant--Fischer min-max principle,
    \begin{equation*}
        \lambda_i(\widehat{\brho}) = \min_{V} \max_{\ket{v} \in V} \big[ \bra{v} \widehat{\brho} \ket{v} \big] \leq \max_{\ket{v} \in V_{\geq i}} \big[ \bra{v} \widehat{\brho} \ket{v}\big] \leq \lambda_i(\rho) + C \sqrt{ \frac{d}{n} \cdot \Big( \lambda_i(\rho) + \frac{d}{n} \Big) },
    \end{equation*}
    where the minimization is over subspaces $V$ with $\dim(V) = d-i+1$. This establishes the desired inequality.
\end{proof}

    \subsection{Relative error bounds for higher-rank observables} \label{sec:relative_error_bounds_higher_rank}

    In the previous subsection, we stated and proved \Cref{thm:main-tool}, our main tool for bucketing. Here we extend the bound to higher-rank observables. This higher-rank version is a straightforward consequence of the rank-one version.

\begin{theorem}[A generalization of the relative error bound to higher-rank observables] \label{thm:uniform-statement-for-observables}
    Let $\rho \in \C^{d \times d}$ be a mixed state, and let $\widehat{\brho}$ be the output of $\mixed(\gps)$, given $n$ copies of $\rho$. Then
    \begin{equation*}
        \Abs[\Big]{\tr( O \cdot (\widehat{\brho} - \rho) )} \leq 2C \cdot \sqrt{ \frac{d}{n} \cdot \mathrm{rank}(O) \cdot \Big( \tr(O^2 \cdot \rho) + \frac{d}{n} \cdot \tr(O^2) \Big) },
    \end{equation*}
    with probability at least $0.99$, for all Hermitian operators $O$ simultaneously.
\end{theorem}

Restricting to $O = \ketbra{w}$ recovers \Cref{thm:uniform-bound-pointwise} up to a constant factor. In this sense, the two theorems are equivalent.

\begin{proof}[Proof of \Cref{thm:uniform-statement-for-observables}]
    Condition on the event that the relative error bound holds. That is, for all $\ket{w}$, we have
    \begin{equation*}
        \ABs{\bra{w} ( \widehat{\brho} - \rho )  \ket{w}} \leq C \cdot \sqrt{ \frac{d}{n} \cdot \Big( \bra{w} \rho \ket{w} + \frac{d}{n} \Big) }.
    \end{equation*}
    Now let $O$ be any observable, with spectral decomposition $O = \sum_{j} o_j \cdot \ketbra{w_j}$, where $j \in [\mathrm{rank}(O)]$. We have
    \begin{align}
         \Abs[\Big]{\tr( O \cdot (\widehat{\brho} - \rho) )} & = \Abs[\bigg]{\sum_{j} o_j \cdot \bra{w_j} (\widehat{\brho} - \rho) \ket{w_j}} \nonumber \\
         & \leq \sum_{j} \Abs{o_j} \cdot \ABs{\bra{w_j}(\widehat{\brho}-\rho) \ket{w_j}} \nonumber \\
        & \leq C \cdot \sum_{j} \Abs{o_j} \cdot \sqrt{ \frac{d}{n} \cdot \Big( \bra{w_j} \rho \ket{w_j} + \frac{d}{n} \Big) } \nonumber \\
        & \leq C \cdot \sum_{j} \Abs{o_j} \cdot \Bigg( \sqrt{ \frac{d}{n} \cdot  \bra{w_j} \rho \ket{w_j}} + \frac{d}{n} \Bigg). \nonumber
    \end{align}
    In the last step we have used $\sqrt{a+b} \leq \sqrt{a} + \sqrt{b}$. We now use Cauchy--Schwarz:
    \begin{align*}
        \Abs[\Big]{\tr( O \cdot (\widehat{\brho} - \rho) )} & \leq C \cdot  \Bigg( \sqrt{\frac{d}{n} \cdot \mathrm{rank}(O) \cdot \sum_j \Abs{o_j}^2 \cdot \bra{w_j}\rho \ket{w_j} } + \frac{d}{n} \cdot \sqrt{\mathrm{rank}(O) \cdot \sum_j \Abs{o_j}^2 }\Bigg) \\
        & = C \cdot \Bigg( \sqrt{\frac{d}{n} \cdot \mathrm{rank}(O) \cdot \tr(O^2 \cdot \rho) } + \frac{d}{n} \cdot \sqrt{\mathrm{rank}(O) \cdot \tr(O^2) }\Bigg).
    \end{align*}
    Finally, we can use $\sqrt{a} + \sqrt{b} \leq 2 \sqrt{a+b}$ to compress this into a slightly weaker bound
    \begin{equation*}
        \Abs[\Big]{\tr( O \cdot (\widehat{\brho} - \rho) )} \leq 2C \cdot \sqrt{ \frac{d}{n} \cdot \mathrm{rank}(O) \cdot \Big( \tr(O^2 \cdot \rho) + \frac{d}{n} \cdot \tr(O^2) \Big) }.\qedhere
    \end{equation*}
\end{proof}

We also give a $1$-norm bound for projectors, which follows as an easy corollary from the theorem above.

\begin{corollary}
    [A $1$-norm relative error bound for projectors] \label{cor:1-norm-bound-projectors}
    Let $\rho \in \C^{d \times d}$ be a mixed state, and let $\widehat{\brho}$ be the output of $\mixed(\gps)$, given $n$ copies of $\rho$. Then
    \begin{equation*}
        \big \| \Pi \cdot (\widehat{\brho} - \rho) \cdot \Pi \big \|_1 \leq 4C \cdot \sqrt{ \frac{d}{n} \cdot \mathrm{rank}(\Pi) \cdot \Big( \tr(\Pi \cdot \rho) + \frac{d}{n} \cdot \mathrm{rank}(\Pi) \Big) },
    \end{equation*}
    with probability at least $0.99$, for all projectors $\Pi$ simultaneously.
\end{corollary}

\begin{proof}
    We condition on the relative error bound for higher-rank observables holding, and note
    \begin{equation} \label{eq:1-norm-bound}
        \big \| \Pi \cdot (\widehat{\brho} - \rho) \cdot \Pi \big \|_1 = \Abs[\Big]{\tr\big( \Pi_+ \cdot (\widehat{\brho} - \rho) \big)} + \Abs[\Big]{\tr\big( \Pi_- \cdot (\widehat{\brho} - \rho) \big)},
    \end{equation}
    where $\Pi_+$ and $\Pi_-$ are the projectors onto the positive and negative eigenspaces of $\Pi \cdot (\widehat{\brho} - \rho) \cdot \Pi$. Applying \Cref{thm:uniform-statement-for-observables} with $O \leftarrow \Pi_{\pm}$, we get
    \begin{align*}
        \Abs[\Big]{\tr( O \cdot (\widehat{\brho} - \rho) )} & \leq 2C \cdot \sqrt{ \frac{d}{n} \cdot \mathrm{rank}(\Pi_{\pm}) \cdot \Big( \tr(\Pi_{\pm} \cdot \rho) + \frac{d}{n} \cdot \mathrm{rank}(\Pi_{\pm}) \Big) } \\
        & \leq 2C \cdot \sqrt{ \frac{d}{n} \cdot \mathrm{rank}(\Pi) \cdot \Big( \tr(\Pi \cdot \rho) + \frac{d}{n} \cdot \mathrm{rank}(\Pi) \Big) }.
    \end{align*}
    Plugging these bounds into \Cref{eq:1-norm-bound} completes the proof.
\end{proof}

\section{Spectrum learning beyond the Keyl--Werner algorithm} \label{sec:spectrum_learning}
\newcommand{\Dtv}{\mathrm{d}_{\mathrm{TV}}}

In this section, we prove \Cref{thm:main}
by giving an algorithm that improves upon the sample complexity of the Keyl--Werner algorithm in the regime where the precision $\epsilon$ is not too small.
The general framework of this algorithm comes from prior work of \cite{PTTW25}, which performs spectrum estimation in the unentangled measurement setting.
This prior work, in turn, adapts a framework from the classical literature on sorted distribution estimation~\cite{HJW18}.

We sketch the strategy here; see~\cite{PTTW25} for a more detailed discussion. First, a bucketing step splits the input state $\rho$ into two buckets: roughly speaking, one contains the large eigenvalues of $\rho$, and the other contains the small eigenvalues. We do this by generating an estimate $\widehat{\brho}$ and letting $\bPi$ be the projector onto the eigenvectors corresponding to eigenvalues larger than a threshold $B \in (0,1)$, to be determined later. With high constant probability, when using $n = O( d/B\epsilon^2)$ samples, $\bPi$ is aligned sufficiently well with the large eigenspace of $\rho$ so that $\overline{\bPi} \rho \overline{\bPi}$ contains only small eigenvalues: eigenvalues at most $1.1 B$, say. This number of samples also estimates all large eigenvalues well.

It then suffices to learn the small eigenvalues. These are estimated by measuring a further $n$ fresh copies of $\rho$ with $\{\bPi, \overline{\bPi}\}$ and applying local moment matching to the copies of $\overline{\bPi} \rho \overline{\bPi}$. We can estimate the first $K$ moments of the small eigenvalues when $B = \Theta(\epsilon^2 K^2/d)$ and use them to estimate the small eigenvalues. Setting $K = O( \log d / \log \log d)$ yields a good estimate. Unraveling our parameter settings gives a final sample complexity of
\begin{equation*}
    n = O \bigg( \frac{d}{B\epsilon^2} \bigg) = O \bigg( \frac{d^2}{K^2 \epsilon^4} \bigg) = O \bigg( \frac{d^2 \cdot (\log \log d)^2 }{\epsilon^4 \cdot (\log d)^2}  \bigg).
\end{equation*}
See \Cref{fig:spectrum-estimation-algorithm} for the algorithm, and \Cref{thm:main-thm} for a formal theorem statement.

{
\floatstyle{boxed}
\restylefloat{figure}
\begin{figure}[H]

Let $B \in (0,1)$ be the bucketing threshold, and $K \in \N$ the number of moments we estimate. The values of these parameters will be determined later. Given $2n$ copies of $\rho \in \C^{d \times d}$:
\begin{enumerate}
    \item Use the first $n$ copies to run $\mixed(\gps)$-based \textbf{bucketing} (\Cref{def:bucketing}) with threshold $B$. The outcome is a projector $\bPi$ of rank~$\br$ and the estimator~$\widehat{\balpha}_{\mathrm{Large}} = (\widehat{\balpha}_1, \dots, \widehat{\balpha}_{\br})$.
    \item Use the remaining $n$ samples to perform \textbf{moment estimation} (\Cref{def:moment-estimator-sub-norm}) and obtain estimates $\widehat{\bp}_1, \dots, \widehat{\bp}_K$ for the moments $\tr(\bsigma), \dots, \tr(\bsigma^K)$ of the state $\bsigma = \overline{\bPi} \rho \overline{\bPi}$.
    \item Apply \textbf{local moment matching} (\Cref{thm:smallestbucketLMM}) to $\widehat{\bp}_1, \dots, \widehat{\bp}_K$ to produce~$\widehat{\balpha}_{\mathrm{Small}} = (\widehat{\balpha}_{\br+1}, \dots, \widehat{\balpha}_{d})$.
\end{enumerate}
The final spectrum estimate is $\widehat{\balpha} = \widehat{\balpha}_{\mathrm{Large}} \| \widehat{\balpha}_{\mathrm{Small}} = (\widehat{\balpha}_1, \dots, \widehat{\balpha}_d)$.
\caption{Our spectrum estimation algorithm.}
\label{fig:spectrum-estimation-algorithm}
\end{figure}
}


\begin{theorem}[\Cref{thm:main} restated]
    \label{thm:main-thm}
    Let $\widehat{\balpha}$ be the output of the spectrum estimation algorithm from~\Cref{fig:spectrum-estimation-algorithm} when run on $2n$ copies of the mixed state $\rho \in \C^{d \times d}$. Then, with high probability, $\dtv{\alpha}{\widehat{\balpha}} \leq \epsilon$ so long as
    \begin{equation*}
        n = O\Big(\frac{d^2 \cdot (\log \log d)^2}{\epsilon^4 \cdot (\log d)^2}\Big).
    \end{equation*}
\end{theorem}

Our analysis is broadly the same as in \cite{PTTW25}, and the only difference between our algorithm and theirs is that our bucketing and moment estimation are performed with different quantum algorithms.
Our central technical contribution is that we perform bucketing using our new tomography guarantee (\Cref{thm:main-tool}).
The analysis of this can be found in \Cref{subsec:bucket}.
For moment estimation, we use existing results~\cite{AISW20}, analyzed in \Cref{sec:moment-estimation}.
After quoting a local moment matching guarantee from \cite{PTTW25}, the proof is completed in \Cref{subsec:final}.

Note that, when $\eps \geq 1$, the guarantee of \Cref{thm:main-thm} is trivial.
Therefore, throughout this section, we assume that $\epsilon \in (0,1)$.


\subsection{Bucketing} \label{subsec:bucket}

In this section, we give an algorithm that performs the bucketing step in \Cref{fig:spectrum-estimation-algorithm}.
The relative error bound of~\Cref{thm:main-tool} for the entangled tomography algorithm $\mixed(\gps)$ will play a central role in~\Cref{lem:bucketing} when analyzing the performance of the bucketing algorithm.

\begin{definition}[$\mixed(\gps)$-based bucketing algorithm]
    \label{def:bucketing}
    Given a threshold $B \in [0,1]$ and $n$ copies of a mixed state $\rho \in \C^{d \times d}$, the \emph{$\mixed(\gps)$-based bucketing algorithm} executes the following steps:
    \begin{enumerate}
        \item Perform mixed state tomography using $\mixed(\gps)$ on $\rho^{\otimes n}$ and obtain an estimate of the state $\widehat{\brho}$.
        \begin{enumerate}
            \item[(a)] Let $\bPi$ be the projector onto the eigenvectors of $\widehat{\brho}$ corresponding to eigenvalues that are at least $B$, and let $\br = \mathrm{rank}(\bPi)$.
            \item[(b)] Let $\widehat{\balpha}_{\mathrm{Large}} = \mathrm{spec}(\widehat{\brho})_{\leq \br}$ be the estimate for the large eigenvalues of $\rho$.
        \end{enumerate}
        \item Return $\widehat{\balpha}_{\mathrm{Large}}$ and $\{\bPi, \overline{\bPi}\}$.
    \end{enumerate}
\end{definition}

For bucketing to be successful, we need the following three guarantees. For an intuitive explanation of why we need these guarantees, see~\Cref{sec:lmm-quantum}, or~\cite[Section 3]{PTTW25} for a more thorough treatment.

\begin{lemma}[Bucketing algorithm guarantees]
    \label{lem:bucketing}
    Given a threshold $0 < B < 1$, suppose we perform the $\mixed(\gps)$-based
    bucketing algorithm on $n$ copies of $\rho$, receiving outputs $\widehat{\balpha}_{\mathrm{Large}}$
    and~$\{\bPi, \overline{\bPi} \}$, where $\br = \rank(\bPi)$. There exists some constant $C_2 > 0$ such that if
    \begin{equation*}
        n \geq C_2 \cdot \frac{d}{B\varepsilon^2},
    \end{equation*}
    then with probability at least $0.99$ the following hold simultaneously: 
    \begin{enumerate}
        \item \label{it:bucketing-misclass} \emph{(Low misclassification error).}
        \begin{equation*}
            \left\|\overline{\bPi}\rho \overline{\bPi} \right\|_\infty \le 1.1 B.
        \end{equation*}
        \item \label{it:bucketing-learn-large-bucket} \emph{(Low error in learning the large eigenvalues).}
        \begin{equation*}
            \br \leq \frac{1.1}{B}
            \quad \text{ and } \quad
            \Dtv\big({\widehat{\balpha}_{\mathrm{Large}}},{ \spec(\rho)_{\leq \br}}\big)
            \le \varepsilon.
        \end{equation*}
        \item \label{it:bucketing-alignment}\emph{(Low alignment error).}
        \begin{equation*}
            \Dtv\big({\spec(\bPi \rho \bPi + \overline{\bPi} \rho \overline{\bPi})},{\spec(\rho)}\big) \leq \varepsilon.
        \end{equation*}
    \end{enumerate}
\end{lemma}

\noindent All three of these guarantees follow from the relative error bound, i.e.\ the tomography guarantee in \Cref{thm:main-tool}.

\begin{proof}
    Since $\widehat{\brho}$ is the output of $\mixed(\gps)$,~\Cref{thm:main-tool} implies that with probability at least $0.99$, there exists a constant $C > 1$
    such that for all $\ket{w}$,
    \begin{align}
        \ABs{\bra{w} ( \rho - \widehat{\brho} )  \ket{w} }
        &\leq C \cdot \sqrt{ \frac{d}{n} \cdot \Big( \bra{w} \rho \ket{w} + \frac{d}{n} \Big)} \nonumber \\
        &\leq C \cdot \frac{d}{n} + C\cdot \sqrt{\frac{d}{n} \cdot \bra{w} \rho \ket{w}}. \label{eq:uniform-bound-in-bucket-proof}
    \end{align}
    For the rest of the proof, we condition on the above inequality being satisfied.
    A standard application of the weighted AM-GM inequality implies that for every parameter $t > 0$,
    \begin{equation*}
         C\cdot \sqrt{\frac{d}{n} \cdot \bra{w} \rho \ket{w}}
         =  C\cdot \sqrt{\frac{d \cdot t}{n} \cdot \frac{\bra{w} \rho \ket{w}}{t}}
         \leq \frac{Ct}{2} \cdot \frac{d}{n} + \frac{C}{2t} \cdot \bra{w} \rho \ket{w}.
    \end{equation*}
    Taking $t = 50C$, we get that
    \begin{equation}
         \label{eq:uniform-bound-in-bucket-proof-2}
        \ABs{ \bra{w} ( \rho - \widehat{\brho} )  \ket{w} }
        \leq (25C^2 +C) \cdot \frac{d}{n} + 0.01 \cdot \bra{w} \rho \ket{w} \leq 26C^2 \cdot \frac{d}{n} + 0.01 \cdot \bra{w} \rho \ket{w}.
    \end{equation}
    Here we are using $C \geq 1$ by \Cref{rem:C_lower_bound} to simplify the bound.

    \emph{\Cref{it:bucketing-misclass}}:
    By construction of the bucketing procedure, the restriction of $\widehat{\brho}$ to
    the $\overline{\bPi}$ block satisfies
    \begin{equation*}
        \bra{w} \overline\bPi\,\widehat{\brho}\overline{\bPi}\ket{w} \leq B,
    \end{equation*}
    for all states $\ket{w}$.
    Let us now apply~\Cref{eq:uniform-bound-in-bucket-proof-2} for any $\ket{w}$ that is in the span of $\overline{\bPi}$:
    \begin{equation*}
        \bra{w} \rho \ket{w}
        \leq \bra{w} \widehat{\brho} \ket{w} + \ABs{\bra{w} ( \rho - \widehat{\brho} ) \ket{w}}
        \leq \bra{w} \widehat{\brho} \ket{w} + 26C^2 \cdot \frac{d}{n} + 0.01 \cdot \bra{w} \rho \ket{w},
    \end{equation*}
    which implies from our choice of $n$ that
    \begin{equation*}
        \bra{w} \rho \ket{w} \leq 1.02 \cdot \Big(\bra{w} \widehat{\brho} \ket{w} + 26C^2 \cdot \frac{d}{n}\Big)
        \leq 1.02B + \frac{27C^2}{C_2}\cdot B\epsilon^2 \leq 1.1B.
    \end{equation*}
    Since the above holds for all $\ket{w}$ in the span of $\overline{\bPi}$, we conclude that $\left\|\overline{\bPi} \rho \overline{\bPi} \right\|_\infty \le 1.1 B$, for sufficiently large $C_2$. Since $\eps \leq 1$, we can take $C_2 \geq 2700 C^2$, for example.

    \emph{\Cref{it:bucketing-learn-large-bucket}}: From our choice of $\br$, it holds that the top $\br$ eigenvectors $\{\ket{\bv_i}\}_{i\in [\br]}$ of $\widehat{\brho}$ satisfy $\bra{\bv_i} \widehat{\brho} \ket{\bv_i} \geq B$. Applying~\Cref{eq:uniform-bound-in-bucket-proof-2}
    with $\ket{\bv_i}$ for each $i \in [\br]$ implies that
    \begin{equation*}
        \bra{\bv_i} \rho \ket{\bv_i}
        \geq \bra{\bv_i} \widehat{\brho} \ket{\bv_i} - 26C^2 \cdot\frac{d}{n} - 0.01 \cdot \bra{\bv_i} \rho \ket{\bv_i}.
    \end{equation*}
    We collect terms and conclude that
    \begin{equation*}
        \bra{\bv_i} \rho \ket{\bv_i}
        \geq \frac{\bra{\bv_i} \widehat{\brho} \ket{\bv_i}}{1.01} - \frac{26C^2}{1.01} \cdot \frac{B\epsilon^2}{C_2}
        \geq B \cdot \Big(\frac{1 - 26C^2\epsilon^2/C_2}{1.01}\Big) \geq \frac{1}{1.1} B.
    \end{equation*}
    The last inequality follows by setting the constant $C_2$ in our number of samples to be large enough with respect to $C$. For example, it again suffices to take $C_2 \geq 2600C^2$.

    Therefore, there exist $\br$ orthogonal unit-norm vectors for which $\bra{\bv_i} \rho \ket{\bv_i} \geq \frac{1}{1.1}B$. Since $\rho$ is positive semidefinite and has trace $1$, this must mean that
    \begin{equation*}
        \br\cdot \frac{B}{1.1}\le \tr(\rho)=1 \implies \br \leq \frac{1.1}{B}.
    \end{equation*}
    Next, we use \Cref{it:spec_lemma_1} of \Cref{lem:dtv_lemma_k} below, applied with $k \leftarrow \br$, to get
    \begin{align*}
    \Dtv\big({\widehat{\balpha}_{\mathrm{Large}}},{\spec(\rho)_{\le \br}}\big)
         & = \Dtv\big({\spec(\widehat{\brho})_{\le \br}},{\spec(\rho)_{\le \br}}\big) \\
         & \leq 2C \cdot \Big(\sqrt{2 \br d/n} + 2 \br d/n\Big) \\
        & \leq 4.4 C \cdot \Big(\sqrt{ 2 d/nB} + 2d/nB \Big) \\
        & \leq O(\epsilon),
    \end{align*}
    where we have used $\br \leq 1.1/B$ and $n = \Theta(d/B\epsilon^2)$.

    \emph{\Cref{it:bucketing-alignment}}:
    Let $\alpha_1 \geq \dots \geq \alpha_d$ be the eigenvalues of $\rho$, let $\bbeta_1 \geq \dots \geq \bbeta_{\br}$ be the largest $\br$ eigenvalues of $\bPi \rho \bPi$ (i.e.\ the possibly nonzero ones, as $\bPi$ has rank $\br$ only), and similarly let $\bbeta_{\br+1} \geq \dots \geq \bbeta_{d}$ be the largest $(d-\br)$ eigenvalues of $\overline{\bPi} \rho \overline{\bPi}$.
    Note that it is not necessarily the case that $\bbeta_{\br} \geq \bbeta_{\br+1}$.
    It holds that
    \begin{align}
        \Dtv \big( \spec(\bPi \rho \bPi + \overline{\bPi} \rho \overline{\bPi}),  {\spec(\rho)} \big)
        ={}& \Dtv \big( \mathrm{sort}(\bbeta_1, \dots, \bbeta_d), (\alpha_1, \dots, \alpha_d) \big) \nonumber \\
        \leq{}& \dtv{(\bbeta_1, \dots, \bbeta_{\br})}{(\alpha_1, \dots, \alpha_{\br})}
        + \Dtv\big({(\bbeta_{\br+1}, \dots, \bbeta_{d})},{(\alpha_{\br+1}, \dots, \alpha_{d})}\big) \nonumber \\
        ={}& \frac{1}{2} \sum_{i=1}^{\br} \Abs{\alpha_i - \bbeta_i}
        + \frac{1}{2} \sum_{i=\br+1}^{d} \Abs{\alpha_i - \bbeta_i}.
        \label{eq:bucketing-temp-eq-2}
    \end{align}
    The inequality holds because the total variation distance between two spectra is minimized when both lists of eigenvalues are sorted.
    Recall that $\mathrm{spec}(\cdot)$ returns the sorted list of the eigenvalues, hence any other ordering of the eigenvalues corresponds to a distance at least as large.

    Cauchy's interlacing theorem tells us that for any projector $\Pi$ of rank $r$, we have 
    \begin{equation*}
        \lambda_i(\rho) \geq \lambda_i(\Pi \rho \Pi) \geq \lambda_{i + d-r}(\rho).
    \end{equation*}
    We apply this twice. First, we will set $\Pi \leftarrow \bPi$, and take the upper bound: for all $i \in \{1, \dots, \br\}$, \begin{equation*}\alpha_i = \lambda_i(\rho) \geq \lambda_i(\bPi \rho \bPi) = \bbeta_i. \end{equation*}
    Second, we will set $\Pi \leftarrow \overline{\bPi}$, and take the lower bound: for all $i \in \{\br+1, \dots, d\}$,
    \begin{equation*}
        \bbeta_i = \lambda_{i - \br}(\overline{\bPi} \rho \overline{\bPi}) \geq \lambda_{(i-\br) + d - (d-\br)}(\rho) = \lambda_{i}(\rho) = \alpha_i.
    \end{equation*}
    Therefore,
    \begin{equation*}
        \Abs{\alpha_i - \bbeta_i} = \begin{cases*}
            \alpha_i - \bbeta_i, & $i \in \{1, \dots, \br\}$, \\
            \bbeta_i - \alpha_i, & $i \in \{\br+1, \dots, d \}$.
        \end{cases*}
    \end{equation*}
    Also, since
    \begin{equation*}
        \sum_{i=1}^{\br} \alpha_i + \sum_{i=\br+1}^d \alpha_i = \tr( \rho) =  \tr\big( (\bPi + \overline{\bPi}) \rho \big) = \tr\big(\bPi \rho \bPi + \overline{\bPi} \rho \overline{\bPi}\big) = \sum_{i=1}^{\br} \bbeta_i + \sum_{i=\br+1}^d \bbeta_i,
    \end{equation*}
    we have
    \begin{equation*}
         \sum_{i=1}^{\br} (\alpha_i - \bbeta_i) = -\sum_{i=\br+1}^{d} (\alpha_i - \bbeta_i).
    \end{equation*}
    Combining these two facts gives
    \begin{align}
        \frac{1}{2} \sum_{i=1}^{\br} \Abs{\alpha_i - \bbeta_i}
        + \frac{1}{2} \sum_{i=\br+1}^{d} \Abs{\alpha_i - \bbeta_i}
        &= \frac{1}{2} \sum_{i=1}^{\br} (\alpha_i - \bbeta_i)
        - \frac{1}{2} \sum_{i=\br+1}^{d} (\alpha_i - \bbeta_i) \nonumber \\
        &= \sum_{i=1}^{\br} (\alpha_i - \bbeta_i) \nonumber \\
        &= 2 \cdot \dtv{\mathrm{spec}(\bPi \rho \bPi)_{\leq \br}}{\mathrm{spec}(\rho)_{\leq \br}}. \label{eq:bucketing-temp-eq-1}
    \end{align}

    Putting~\Cref{eq:bucketing-temp-eq-1,eq:bucketing-temp-eq-2} together implies that
    \begin{equation*}
        \Dtv\big({\spec(\bPi \rho \bPi + \overline{\bPi} \rho \overline{\bPi})},{\spec(\rho)}\big)
        \leq 2 \cdot \Dtv\big({\spec(\bPi \rho \bPi)_{\leq \br}},{\spec(\rho)_{\leq \br}}\big).
    \end{equation*}
    However, by \Cref{it:spec_lemma_3} of \Cref{lem:dtv_lemma_k} below, applied with $k \leftarrow \br$, we have
    \begin{align*}
        \Dtv \big( \spec(\rho)_{\leq \br}, {\spec( \bPi \rho \bPi )_{\leq \br}} \big) & \leq 4C \cdot \Big(\sqrt{2 \br d/n} + 2 \br d/n\Big)  \\
        & \leq 4.4 C \cdot \Big(\sqrt{ 2 d/nB} + 2d/nB \Big) \\
        & \leq O(\epsilon),
    \end{align*}
    where we have used $\br \leq 1.1/B$ and $n = \Theta(d/B\epsilon^2)$. This concludes the proof. \qedhere

\end{proof}


We now circle back to state and prove some deferred results. We have isolated these here instead of including them inside the proof of \Cref{lem:bucketing} since they will also be used in \Cref{sec:additional_tomography_results}. 

\begin{lemma}[A few useful total variation distance estimates]\label{lem:dtv_lemma_k}
    Let $\rho \in \C^{d \times d}$ be a mixed state, and let $\widehat{\brho}$ be the output of $\mixed(\gps)$, given $n$ copies of $\rho$. Fix $k \in [d]$, and write $\bPi$ for the projector onto the top $k$ eigenvectors of $\widehat{\brho}$. When the relative error bound holds, then the following hold simultaneously:
    \begin{enumerate}
        \item \label{it:spec_lemma_1} $\Dtv \big( {\spec(\rho)_{\leq k}}, \spec(\widehat{\brho})_{\leq k} \big) \leq 2C \cdot \Big(\sqrt{2kd/n} + 2kd/n\Big),$
        \item \label{it:spec_lemma_2} $\Dtv \big( \spec(\widehat{\brho})_{\leq k}, {\spec( \bPi \rho \bPi )_{\leq k}} \big) \leq 2C \cdot \Big(\sqrt{2kd/n} + 2kd/n\Big),$
        \item \label{it:spec_lemma_3} $\Dtv \big( \spec(\rho)_{\leq k}, {\spec( \bPi \rho \bPi )_{\leq k}} \big) \leq 4C \cdot \Big(\sqrt{2kd/n} + 2kd/n\Big).$
    \end{enumerate}
\end{lemma}

\begin{proof}
\emph{\Cref{it:spec_lemma_1}}: Define $\bPi^+$ as the projector onto the span of the top $k$ eigenvectors of \emph{both} $\rho$ \emph{and} $\widehat{\brho}$. Then $\rank(\bPi^+) \leq 2k$, and
\begin{align*}
    \Dtv \big( {\spec(\rho)_{\leq k}}, \spec(\widehat{\brho})_{\leq k} \big) & = \Dtv \big( {\spec(\bPi^+\rho\bPi^+)_{\leq k}}, \spec(\bPi^+\widehat{\brho}\bPi^+)_{\leq k} \big) \\
    & \leq \Dtv \big( \spec(\bPi^+\rho\bPi^+), \spec(\bPi^+\widehat{\brho}\bPi^+) \big) \\
    & \leq \dtr \big( \bPi^+\rho\bPi^+, \bPi^+\widehat{\brho}\bPi^+ \big).
\end{align*}
The second inequality is, for example, \cite[Corollary 7.4.9.3]{HJ13}. Now we use \Cref{cor:1-norm-bound-projectors} with $\Pi \leftarrow \bPi^+$ to get
\begin{align*}
    \dtr \big( \bPi^+\rho\bPi^+, \bPi^+\widehat{\brho}\bPi^+ \big) & = \frac{1}{2} \norm{ \bPi^+ ( \rho - \widehat{\brho}) \bPi^+}_1 \\
    & \leq 2C \cdot \sqrt{ \frac{d}{n} \cdot \rank(\bPi^+) \cdot \Big(  \tr( \bPi^+ \cdot \rho) + \frac{d}{n} \cdot \mathrm{rank}(\bPi^+) \Big) } \\
    & \leq 2C \cdot \sqrt{ \frac{2kd}{n}  \cdot \Big(  1 + \frac{2kd}{n} \Big) }  \\
    & \leq 2C \cdot \Big(\sqrt{2kd/n} + 2kd/n\Big).
\end{align*}
This completes the proof of the first item.

\noindent \emph{\Cref{it:spec_lemma_2}}: Similar to the first item, we begin by observing
\begin{align*}
    \Dtv \big( \spec(\widehat{\brho})_{\leq k}, {\spec( \bPi \rho \bPi )_{\leq k}} \big) & = \Dtv \big( \spec( \bPi \widehat{\brho} \bPi )_{\leq k}, {\spec( \bPi \rho \bPi )_{\leq k}} \big) \\
    & = \Dtv \big( \spec( \bPi \widehat{\brho} \bPi ), {\spec( \bPi \rho \bPi )} \big) \\
    & \leq \dtr \big( \bPi \widehat{\brho} \bPi , {\bPi \rho \bPi } \big).
\end{align*}
Moreover, $\dtr \big( \bPi \widehat{\brho} \bPi , {\bPi \rho \bPi } \big) \leq \dtr \big( \bPi^+ \widehat{\brho} \bPi^+ , {\bPi^+ \rho \bPi^+ } \big)$, 
and reuse the first item's upper bound.

\noindent \emph{\Cref{it:spec_lemma_3}}: This third item follows from the first and second items, via the triangle inequality:
\begin{equation*}
    \Dtv \big( \spec(\rho)_{\leq k}, {\spec( \bPi \rho \bPi )_{\leq k}} \big) \leq \Dtv \big( {\spec(\rho)_{\leq k}}, \spec(\widehat{\brho})_{\leq k} \big) + \Dtv \big( \spec(\widehat{\brho})_{\leq k}, {\spec( \bPi \rho \bPi )_{\leq k}} \big).
\end{equation*}
This completes the proof of all three items.
\end{proof}

\subsection{Moment estimation}
\label{sec:moment-estimation}
The second step of our algorithm performs moment estimation. We estimate the moments of the spectrum of $\bsigma = \overline{\bPi} \rho \overline{\bPi}$, which roughly correspond to moments of the small eigenvalues of $\rho$. Since this state is not generally normalized, we extend the moment estimation algorithm of~\Cref{def:moment-est-alg} to sub-normalized states in~\Cref{def:moment-estimator-sub-norm} and analyze its bias and variance in~\Cref{lem:variance-of-subnorm-moment}.

\begin{definition}[Moment estimator for sub-normalized state]
    \label{def:moment-estimator-sub-norm}
    Given $n$ copies of a mixed state $\rho \in \C^{d \times d}$ and the description of a projector $\Pi$, let $\sigma = \Pi \rho \Pi$ be a sub-normalized state. The \emph{$k$-th moment estimator for $\sigma$} is denoted by $\widehat{\bp}_k(\sigma)$ and defined as follows:
    \begin{enumerate}
        \item \label{it:first} Measure all $n$ copies of $\rho$ using $\{\Pi, I-\Pi\}$. Let $\bm$ be the number of samples for which we obtain outcome $\Pi$, and discard the remaining samples. The state has collapsed to $\tau^{\otimes \bm}$, where $\tau = \sigma/\tr(\sigma)$.
        \item Perform the projective measurement $\{\Pi_{\lambda}\}_{\lambda}$ on $\tau^{\otimes \bm}$. Let $\blambda$ be the resulting Young diagram.
        \item Return the estimate
        \begin{equation*}
            \widehat{\bp}_k(\sigma) = \begin{cases} p^{\#}_{(k)}(\blambda)/n^{\downarrow k}, & \bm \geq k,\\
            0, & \bm < k.
            \end{cases}
        \end{equation*}
    \end{enumerate}
\end{definition}

Below, we prove that, as in the normalized-state case, the estimator for sub-normalized states is also unbiased and derive a similar upper bound on its variance.
\begin{lemma}[Bounds on the moment estimator]
    \label{lem:variance-of-subnorm-moment}
    Let $\rho \in \C^{d \times d}$ be a mixed state, $\Pi$ a projector, and $\sigma = \Pi \rho \Pi$.
    Given $n$ copies of $\rho$, the $k$-th moment estimator $\widehat{\bp}_k(\sigma)$ for $\sigma$ from~\Cref{def:moment-estimator-sub-norm} is an unbiased estimator for $\tr(\sigma^k)$ whenever $k \leq n$. Moreover, when $k \leq \sqrt{n}$, it has variance
    \begin{equation*}
        \Var\big[\widehat{\bp}_k(\sigma)\big] \leq \frac{100\cdot 2^{6k}k^{8k}}{n^k} + \frac{100\cdot 2^{6k}k^{8k}}{n} \cdot \tr(\sigma^{2k-1}).
    \end{equation*}
\end{lemma}
\begin{proof}
    We first prove unbiasedness. The number of samples $\bm$ for which we obtain outcome $\Pi$ in Step 1 of~\Cref{def:moment-estimator-sub-norm} satisfies $\bm \sim \mathrm{Binom}(n, \tr(\sigma))$. Recall that we define $\tau = \sigma/\tr(\sigma)$ to be the normalized version of $\sigma$; hence, $\tr(\tau^k) = \tr(\sigma^k)/\tr(\sigma)^k$.

    To start, we claim that
    \begin{equation} \label{eq:exp_of_p_k_hat}
        \E\big[\widehat{\bp}_k(\sigma) \mid \bm \big] = \frac{\bm^{\downarrow k} \cdot \tr(\tau^k)}{n^{\downarrow k}}.
    \end{equation}
    If $\bm < k$, then both sides are zero; otherwise,
    \begin{equation*}
         \E\big[\widehat{\bp}_k(\sigma) \mid \bm \big] = \frac{\E \big[ p^{\#}_{(k)}(\blambda) \big]}{n^{\downarrow k}} = \frac{\bm^{\downarrow k} \cdot \tr(\tau^k)}{n^{\downarrow k}},
    \end{equation*}
    using \Cref{lem:moment-estimation-norm}. Thus, by the law of total expectation,
    \begin{equation*}
        \E\big[\widehat{\bp}_k(\sigma)\big] = \E_{\bm} \Big[ \E\big[\widehat{\bp}_k(\sigma) \mid \bm \big] \Big] = \frac{\tr(\tau^k)}{n^{\downarrow k}} \cdot \E_{\bm} \big[ \bm^{\downarrow k} \big] = \frac{\tr(\tau^k)}{n^{\downarrow k}} \cdot n^{\downarrow k} \tr( \sigma )^k = \tr(\sigma^k).
    \end{equation*}
    In the second-to-last step, we have used \Cref{lem:binom-facts-downarrow}. This completes the proof of unbiasedness.

    To bound the variance, we use the law of total variance to write
    \begin{equation}
        \Var\big[\widehat{\bp}_k(\sigma)\big] = \E_{\bm} \Big[\Var\big[\widehat{\bp}_k(\sigma) \mid \bm\big]\Big]
        + \Var_{\bm}\Big[\E\big[\widehat{\bp}_k(\sigma) \mid \bm\big]\Big]. \label{eq:two_term_variance}
    \end{equation}
    For the first term, we note that \Cref{lem:moment-estimation-norm} bounds the variance when $\bm \geq k$. When $\bm < k$, the variance is zero, so the same upper bound holds. Thus, we bound the first contribution to the variance as
    \begin{align}
        \E_{\bm} \Big[\Var\big[\widehat{\bp}_k(\sigma) \mid \bm\big]\Big] & \leq \frac{1}{(n^{\downarrow k})^2} \cdot \E_{\bm} \Big[ k^{6k} \cdot \bm^k + k^{6k} \cdot \bm^{2k-1} \cdot \tr(\tau^{2k-1}) \Big]  \nonumber \\
        & = \frac{k^{6k}}{(n^{\downarrow k})^2} \cdot \E_{\bm} \Big[ \bm^{k} + \bm^{2k-1} \cdot \tr(\tau^{2k-1}) \Big]. \label{eq:expectation_of_variance}
    \end{align}
    For the first term here, we use the bound
    \begin{align*}
        \frac{k^{6k}}{(n^{\downarrow k})^2} \cdot \E_{\bm} \big[  \bm^k \big] & \leq \frac{2^k \cdot k^{7k}}{(n^{\downarrow k})^2} \cdot \bigg( \frac{\E \big[  \bm^k \big]}{2^k \cdot k^k} \bigg) \\
        & \leq \frac{2^k \cdot k^{7k}}{(n^{\downarrow k})^2} \cdot \bigg( \frac{2^k \cdot \big( k^k + n^k \cdot \tr(\sigma)^k\big)}{2^k \cdot k^k} \bigg) \tag{\Cref{lem:binom-facts-power}} \\
        & \leq \frac{2^k \cdot k^{7k}}{(n^{\downarrow k})^2} \cdot \Big( 1 + n^k \cdot \tr(\sigma)^k \Big) \\
        & \leq \frac{2^k \cdot k^{7k}}{(n^{\downarrow k})^2} \cdot \big( 1 + n^k \big).
    \end{align*}
    For the second term in \Cref{eq:expectation_of_variance}, we have
    \begin{align*}
        \frac{k^{6k}}{(n^{\downarrow k})^2} \cdot \tr(\tau^{2k-1}) \cdot \E_{\bm} \big[ \bm^{2k-1} \big] & \leq \frac{2^{2k} \cdot k^{6k} \cdot (2k)^{2k}}{(n^{\downarrow k})^2} \cdot \tr(\tau^{2k-1}) \cdot \bigg( \frac{\E \big[  \bm^{2k-1} \big]}{2^{2k-1} \cdot (2k-1)^{2k-1}} \bigg) \\
        & = \frac{2^{4k} \cdot k^{8k}}{(n^{\downarrow k})^2} \cdot \tr(\tau^{2k-1}) \cdot \bigg( \frac{\E \big[  \bm^{2k-1} \big]}{2^{2k-1} \cdot (2k-1)^{2k-1}} \bigg) \\
        & \leq \frac{2^{4k} \cdot k^{8k}}{(n^{\downarrow k})^2} \cdot \tr(\tau^{2k-1}) \cdot \bigg( \frac{2^{2k-1} \cdot \big( (2k-1)^{2k-1} + n^{2k-1} \cdot \tr(\sigma)^{2k-1}  \big)}{2^{2k-1} \cdot (2k-1)^{2k-1}}\bigg) \tag{\Cref{lem:binom-facts-power}} \\
        & \leq \frac{2^{4k} \cdot k^{8k}}{(n^{\downarrow k})^2} \cdot \tr(\tau^{2k-1}) \cdot \Big(1 + n^{2k-1} \cdot \tr(\sigma)^{2k-1} \Big) \\
        & \leq \frac{2^{4k} \cdot k^{8k}}{(n^{\downarrow k})^2} \cdot \Big( 1 + n^{2k-1} \cdot \tr(\sigma^{2k-1} ) \Big).
    \end{align*}
    In the last line, we have used $\tr(\tau^{2k-1}) \leq 1$, and $\tau = \sigma/\tr(\sigma)$.
    We now bound the second contribution to \Cref{eq:two_term_variance} as
    \begin{align*}
    \Var_{\bm}\Big[\E\big[\widehat{\bp}_k(\sigma) \mid \bm\big]\Big] & = \bigg( \frac{\tr(\tau^k)}{n^{\downarrow k}}\bigg)^2 \cdot \Var_{\bm} \big[ \bm^{\downarrow k} \big] \tag{\Cref{eq:exp_of_p_k_hat}}\\
    & \leq \bigg( \frac{\tr(\tau^k)}{n^{\downarrow k}}\bigg)^2 \cdot 2^k \cdot k! \cdot \Big(n^k \cdot \tr (\sigma)^k + n^{2k-1}\cdot \tr(\sigma)^{2k-1} \Big) \tag{\Cref{lem:binom-facts-var-downarrow}} \\
    & \leq \frac{2^k \cdot k!}{( n^{\downarrow k} )^2} \cdot \Big( n^k + n^{2k-1} \cdot \tr( \sigma^{2k-1} )\Big).
    \end{align*}
    In the last step, we have used $\tr(\tau^k) \leq 1$ and $\tr(\sigma) \leq 1$ to simplify the first term, and $\tr(\tau^k)^2 \leq \tr(\tau^{2k-1})$ (which follows from Cauchy-Schwarz\footnote{Let $\{\lambda_i\}$ be the eigenvalues of $\tau$. Then for $k \geq 1$,\begin{equation*}
        \tr(\tau^k)^2 = \Big( \sum_{i} \lambda^k_i \Big)^2 = \Big( \sum_i \lambda^{k - 1/2}_i  \cdot \lambda_i^{1/2} \Big)^2 \leq \Big( \sum_i \lambda^{2k-1}_i \Big) \cdot \Big( \sum_i \lambda_i \Big) = \tr(\tau^{2k-1} ) \cdot 1.
    \end{equation*}}) and $\tau =\sigma/\tr(\sigma)$ to simplify the second. Altogether, we have
    \begin{align*}
         \Var\big[\widehat{\bp}_k(\sigma)\big] & \leq \frac{2^k \cdot k^{7k}}{(n^{\downarrow k})^2} \cdot \big( 1 + n^k \big) + \frac{2^{4k} \cdot k^{8k}}{(n^{\downarrow k})^2} \cdot \Big( 1 + n^{2k-1} \cdot \tr(\sigma^{2k-1} ) \Big) + \frac{2^k \cdot k!}{( n^{\downarrow k} )^2} \cdot \Big( n^k + n^{2k-1} \cdot \tr( \sigma^{2k-1} )\Big) \\
         & \leq \frac{2^{4k} \cdot k^{8k}}{(n^{\downarrow k})^2} \cdot \Big( \big( 1+ n^k \big) + \big( 1 + n^{2k-1} \cdot \tr(\sigma^{2k-1}) \big) + \big( n^k + n^{2k-1} \cdot \tr(\sigma^{2k-1}) \big) \Big) \\
         & \leq \frac{2^{5k} \cdot k^{8k}}{(n^{\downarrow k})^2} \cdot \Big( 1+ n^k + n^{2k-1} \cdot \tr(\sigma^{2k-1}) \Big) \\
         & \leq \frac{2^{6k} \cdot k^{8k}}{(n^{\downarrow k})^2} \cdot \Big(n^k + n^{2k-1} \cdot \tr(\sigma^{2k-1}) \Big).
    \end{align*}

    Finally, with our choices of $k$ and $n$, it holds that
    \begin{equation*}
        \frac{n^k}{n^{\downarrow k}}
        \leq \Big(\frac{n}{n-k+1}\Big)^k
        \leq \Big(1 + \frac{k}{n-k+1}\Big)^k
        \leq \exp(\frac{k^2}{n-k+1})
        \leq \exp(\frac{n}{n/2}) \leq 10.
    \end{equation*}
    Thus $n^k \leq 10n^{\downarrow k}$, from which we obtain the following convenient variance bound:
    \begin{equation*}
        \Var\big[\widehat{\bp}_k(\sigma)\big]
        \leq \frac{100\cdot 2^{6k}k^{8k}}{n^k} + \frac{100\cdot 2^{6k}k^{8k}}{n} \cdot \tr(\sigma^{2k-1}). \qedhere
    \end{equation*}
\end{proof}
A standard application of Chebyshev's inequality implies that we can obtain estimates for $\tr(\sigma^k)$ for all $k \in [K]$ that simultaneously have small additive error with high probability.
\begin{corollary}
    \label{cor:additive-error-moments-small-bucket}
    Let $0 < B < 1$ be a threshold, $\rho \in \C^{d \times d}$ a mixed state, and $\Pi$ a projector, such that the state $\sigma = \Pi \rho \Pi$ has eigenvalues at most $B$.
    Using $n = O(d/B\epsilon^2)$ copies of $\rho$, and a positive integer $K$ such that $K \leq \sqrt{n}$, there exist estimators $\{\widehat{\bp}_k\}_{k \in[K]}$ such that with high probability
    \begin{equation*}
        \ABs{\widehat{\bp}_k - \tr\big(\sigma^k\big)} \leq 2^{3K}K^{5K} \cdot \Big(\frac{B^{k/2} \epsilon^{k}}{d^{k/2}} + B^{k}\epsilon\Big)
    \end{equation*}
    for all $k \in [K]$ simultaneously.
\end{corollary}

\begin{proof}
    Our estimators will be the moment estimators $\widehat{\bp}_k = \widehat{\bp}_k(\sigma)$ from~\Cref{def:moment-estimator-sub-norm}. As explained in~\Cref{rem:estimate-moments-simultaneously}, multiple moments can be estimated simultaneously using the same set of samples. Since each $\widehat{\bp}_k$ is an unbiased estimator for $\tr(\sigma^k)$, Chebyshev's inequality implies that
    \begin{equation*}
        \Pr\Big[\ABs{\widehat{\bp}_k - \tr\big(\sigma^k\big)} > 10\cdot \sqrt{K} \cdot \sqrt{\Var\big[\widehat{\bp}_k\big]}\Big]
        \leq \frac{1}{100K}.
    \end{equation*}
    The variance bound from~\Cref{lem:variance-of-subnorm-moment} above gives
    \begin{align*}
        \Var[\widehat{\bp}_k]
        &\leq \frac{100\cdot 2^{6k}k^{8k}}{n^k} + \frac{100\cdot 2^{6k}k^{8k}}{n} \cdot \tr(\sigma^{2k-1}) \\
        &\leq \frac{1}{100}\cdot\bigg(\frac{2^{6k}k^{8k} \cdot B^k\epsilon^{2k}}{d^k} + \frac{2^{6k}k^{8k} \cdot B\epsilon^2}{d} \cdot d B^{2k-1}\bigg) \\
        &= \frac{1}{100} \cdot \bigg(\frac{2^{6k}k^{8k} \cdot B^k\epsilon^{2k}}{d^k} + 2^{6k}k^{8k} \cdot B^{2k}\epsilon^2\bigg).
    \end{align*}
    Thus, applying a union bound over all $k \in [K]$ implies that with probability at least $0.99$, it holds that for all $k \in [K]$:
    \begin{align*}
        \Abs{\widehat{\bp}_k - \tr(\sigma^k)}
        &\leq 10\cdot \sqrt{K} \cdot \sqrt{\Var[\widehat{\bp}_k]} \\
        &\leq \sqrt{K} \cdot \bigg(\frac{2^{3k}k^{4k} \cdot B^{k/2}\epsilon^{k}}{d^{k/2}} + 2^{3k}k^{4k} \cdot B^{k}\epsilon\bigg) \\
        &\leq 2^{3K}K^{5K} \cdot \bigg(\frac{B^{k/2} \epsilon^{k}}{d^{k/2}} + B^{k} \epsilon\bigg).
    \end{align*}
    This completes the proof.
\end{proof}

\subsection{Spectrum estimation for full-rank states} \label{subsec:final}

The final step of the algorithm is local moment matching, which provides an algorithm that translates the moment estimates of the small bucket $\bsigma = \overline{\bPi} \rho \overline{\bPi}$ into estimates of the small eigenvalues $\widehat{\balpha}_{\mathrm{Small}}$. We will use the following guarantee from the local moment matching technique as a black box.

\begin{theorem}[Local moment matching guarantees,~{\cite[Theorem~7.1]{PTTW25}}]
    \label{thm:smallestbucketLMM}
    For some threshold $0 < B \leq 1$, let $\alpha = (\alpha_1, \dots, \alpha_d)$ be values in $[0, B]$ sorted in nonincreasing order and $\sum_{i=1}^d \alpha_i \leq 1$.
    Fix some $K \in \N$ and suppose that we have an estimate $\widehat{p}_k$ for $p_k(\alpha) = \sum_{i=1}^d \alpha_i^k$ with error $V_k$, for each $k \in [K]$. That is,
    \begin{equation*}
        \ABS{\widehat{p}_k - p_k(\alpha) } \leq V_k,
        \quad \forall k\in [K].
    \end{equation*}
    Then there exists a randomized algorithm which produces a sorted estimate $\widehat{\balpha} = (\widehat{\balpha}_1, \dots, \widehat{\balpha}_d)$ of $\alpha$ with expected error at most
    \begin{equation}\label{eq:lmm-error}
        \E_{\widehat{\balpha}}[\dtv{\widehat{\balpha}}{\alpha}] \leq O\Big( \frac1K\sqrt{Bd} + 2^{9K/2}B\sum_{k=1}^K B^{-k}V_k\Big).
    \end{equation}
\end{theorem}
We are now ready to prove the main theorem of this paper, which we restate again for the reader's convenience. 

\begin{theorem}[\Cref{thm:main} restated]
    Let $\widehat{\balpha}$ be the output of the spectrum estimation algorithm from~\Cref{fig:spectrum-estimation-algorithm} when run on $2n$ copies of the mixed state $\rho \in \C^{d \times d}$. Then, with high constant probability, $\dtv{\alpha}{\widehat{\balpha}} \leq \epsilon$ so long as
    \begin{equation*}
        n = O\bigg(\frac{d^2 \cdot (\log \log d)^2}{\epsilon^4 \cdot (\log d)^2}\bigg).
    \end{equation*}
\end{theorem}

\begin{proof}
By taking $n = O(d/B\epsilon^2)$ copies, we have the following bucketing algorithm guarantees, by \Cref{lem:bucketing}:
\begin{enumerate}
    \item Low misclassification error: $\norm{ \overline{\bPi} \rho \overline{\bPi} }_\infty \leq 1.1 B$;
    \item Low error in learning the large eigenvalues: there are $\br \leq 1.1/B$ such eigenvalues, and $\Dtv( \widehat{\balpha}_{\mathrm{Large}}, \spec(\rho)_{\leq \br} ) \leq \epsilon$; and
    \item Low alignment error: $\Dtv( \spec( \bPi \rho \bPi + \overline{\bPi} \rho \overline{\bPi}), \spec(\rho) ) \leq \epsilon$.
\end{enumerate}
All three hold simultaneously with high constant probability when the relative error bound holds. Now, from the triangle inequality, we have
\begin{equation*}
    \Dtv \big( \alpha, \widehat{\balpha} \big) \leq \Dtv \big( \alpha, \spec( \bPi \rho \bPi + \overline{\bPi} \rho \overline{\bPi} )\big) + \Dtv \big(\spec( \bPi \rho \bPi + \overline{\bPi} \rho \overline{\bPi}), \widehat{\balpha} \big).
\end{equation*}
The first term is at most $\epsilon$, by low alignment error (as $\alpha = \spec(\rho)$). For the second term, we can upper bound the distance by rearranging the spectrum of $\widehat{\balpha}$, and then a further application of the triangle inequality:
\begin{align*}
    \Dtv\big( \spec(\bPi \rho \bPi + \overline{\bPi} \rho \overline{\bPi}),  \widehat{\balpha} \big) & \leq \Dtv \big( \spec( \bPi \rho \bPi), \widehat{\balpha}_{\mathrm{Large}} \big) + \Dtv \big( \spec( \overline{\bPi} \rho \overline{\bPi}), \widehat{\balpha}_{\mathrm{Small}} \big) \\
    & \leq \Dtv\big( \spec( \bPi \rho \bPi), \spec(\rho)_{\leq \br} \big) + \Dtv \big( \spec(\rho)_{\leq \br}, \widehat{\balpha}_{\mathrm{Large}}\big) + \Dtv \big( \spec( \overline{\bPi} \rho \overline{\bPi}), \widehat{\balpha}_{\mathrm{Small}} \big).
\end{align*}
For the first term: combining \Cref{it:spec_lemma_3} of \Cref{lem:dtv_lemma_k} and our bound on $\br$, which both hold when the relative error bound holds, with our choice of $n$, we have
\begin{equation*}
    \Dtv\big( \spec( \bPi \rho \bPi), \spec(\rho)_{\leq \br} \big) \leq O \Big(\sqrt{2\br d/n} + 2 \br d/n\Big) \leq O \Big( \sqrt{d/Bn} + d/Bn \Big) \leq O(\eps).
\end{equation*}
The second term is at most $\epsilon$ by low error in learning the large eigenvalues. The third term remains, but below we show it is at most $O(\epsilon)$ for $B = \Theta(\epsilon^2 K^2/d)$ and $K = O(\log d/\log \log d)$. Combining all above steps, we have $\Dtv(\alpha, \widehat{\balpha}) \leq O(\epsilon)$, and the result follows by rescaling $\epsilon$, pending an analysis of sample complexity.

We now focus on $\Dtv ( \spec( \overline{\bPi} \rho \overline{\bPi}), \widehat{\balpha}_{\mathrm{Small}})$. By our moment estimation algorithm, since $n = O(d/B\epsilon^2)$, we have 
\begin{equation*}
    \ABs{\widehat{\bp}_k - p_k(\overline{\bPi} \rho \overline{\bPi})} \leq 2^{3K}K^{5K} \cdot \bigg(\frac{B^{k/2} \epsilon^{k}}{d^{k/2}} + B^{k}\epsilon\bigg)
\end{equation*}
for all $k \in [K]$ simultaneously, with high constant probability. Here, we have used the low misclassification error, so that $\overline{\bPi} \rho \overline{\bPi}$ has eigenvalues of size at most $O(B)$. By the local moment matching guarantee of \Cref{thm:smallestbucketLMM},
\begin{equation}\label{eq:lmm-error-J}
\Dtv(\spec(\overline{\bPi}\rho \overline{\bPi}), \widehat{\balpha}_{\mathrm{Small}}) \leq O\bigg( \frac{\sqrt{Bd}}{K} + 2^{9K/2}B\sum_{k=1}^K B^{-k}V_k\bigg),
    \end{equation}
    with high constant probability, by Markov's inequality. Here, $V_k \coloneq 2^{3K}K^{5K} \cdot (B^{k/2} \epsilon^{k}/d^{k/2} + B^{k}\epsilon)$. This is the last term we need to bound. Our choice of $B = \Theta(\epsilon^2K^2/d)$ makes the first term at most $O(\epsilon)$. Substituting the definition of $V_k$ into the remaining term in \Cref{eq:lmm-error-J} gives
    \begin{equation*}
        2^{9K/2} B \sum_{k=1}^K B^{-k} V_k \leq 2^{8K}K^{5K} \cdot B \cdot \bigg(\sum_{k=1}^K B^{-k} \cdot \frac{B^{k/2} \cdot \epsilon^k}{d^{k/2}}\bigg) + 2^{8K}K^{5K} \cdot B \cdot \bigg(\sum_{k=1}^K B^{-k}  \cdot B^k \epsilon\bigg).
    \end{equation*}
    The first term can be bounded as:
    \begin{align*}
        2^{8K}K^{5K} \cdot B \cdot \bigg(\sum_{k=1}^K B^{-k} \cdot \frac{B^{k/2} \cdot \epsilon^k}{d^{k/2}}\bigg) & = 2^{8K}K^{5K} \cdot B \cdot \bigg(\sum_{k=1}^K \frac{\epsilon^k}{B^{k/2} \cdot d^{k/2}}\bigg) \\
        & = O \bigg( \frac{2^{8K} K^{5K+2}}{d} \cdot \epsilon^2 \cdot \sum_{k=1}^{K} \frac{1}{K^{k}} \bigg) \\
        & = O \bigg( \frac{2^{8K} K^{5K+2}}{d} \cdot \epsilon^2 \bigg).
    \end{align*}
    The second term is
    \begin{equation*}
        2^{8K}K^{5K} \cdot B \cdot \bigg(\sum_{k=1}^K B^{-k}  \cdot B^k \epsilon\bigg) = 2^{8K}K^{5K+1} \cdot B\epsilon = O \bigg( \frac{2^{8K} K^{5K+3}}{d} \cdot \epsilon^3 \bigg)
    \end{equation*}
    We can set $K = O( \log d/ \log \log d )$, for a sufficiently small constant, while preserving $2^{8K} K^{5K+3} \leq d$, and thus these terms are $O(\epsilon^2)$ and $O(\epsilon^3)$ respectively. Therefore, the right-hand side of \Cref{eq:lmm-error-J} becomes $O(\epsilon)$.

    Finally, our sample complexity is
    \begin{equation*}
        2n = O \bigg( \frac{d}{B\epsilon^2} \bigg) = O \bigg( \frac{d^2}{K^2 \epsilon^4} \bigg) = O \bigg( \frac{d^2 \cdot (\log \log d)^2}{\epsilon^4 \cdot (\log d)^2} \bigg).
    \end{equation*}
    This completes the proof. \qedhere

\end{proof}

\subsection{Spectrum estimation for rank-$r$ states}

When the input state is promised to be of rank at most $r$, the Keyl--Werner algorithm can estimate the input's spectrum using $\Theta(r^2)$ copies. In this case, the number of copies consumed is fewer than for state tomography, which requires $\Theta(rd)$ copies.

In this section, we extend our full-rank spectrum estimation algorithm to the rank-$r$ case, and show that $O(r^2 \cdot (\log \log (r)/\log(r))^2 )$ copies suffice for this task. The main idea is that for a bipartite pure state $\ket{\psi} \in \C^d \otimes \C^r$, the reduced density matrices $\tr_{\reg{1}}(\ketbra{\psi})$ and $\tr_{\reg{2}}(\ketbra{\psi})$ have the same nonzero spectrum. In particular, for a mixed state $\rho \in \C^{d \times d}$, with purification $\ket{\rho} \in \C^{d} \otimes \C^{r }$, we may learn the spectrum of $\tr_{\reg{1}}(\ketbra{\rho})$, a possibly full-rank state in an $r$-dimensional Hilbert space, rather than the spectrum of $\rho = \tr_{\reg{2}}(\ketbra{\rho})$, the original rank-$r$ state in a $d$-dimensional Hilbert space.\footnote{We are grateful to Xinyu (Norah) Tan for sharing with us this nice idea, which itself is a special case of a more general transformation that subsequently appeared in \cite{LT26}.} Here is the full algorithm. 

{
\floatstyle{boxed}
\restylefloat{figure}
\begin{figure}[H]
Given $n$ copies of a rank-$r$ mixed state $\rho \in \C^{d \times d}$:
\begin{enumerate}
    \item Apply $\purifychan^{(d, r)}$ to prepare $n$ copies of a random purification $\ket{\brho} \in \C^d \otimes \C^r$.
    \item Trace out the original registers to obtain $n$ copies of $\bsigma \coloneq \tr_{\reg{1}}(\ketbra{\brho}) \in \C^{r \times r}$.
    \item Apply the full-rank spectrum learning algorithm in \Cref{fig:spectrum-estimation-algorithm} to learn an estimate $\widehat{\balpha} = (\widehat{\balpha}_1, \dots, \widehat{\balpha}_r)$ of the spectrum of $\bsigma$.
    \item Output $\widehat{\bbeta} \coloneq \widehat{\balpha}\| 0^{d-r} = (\widehat{\balpha}_1, \dots, \widehat{\balpha}_r, 0, \dots, 0)$.
\end{enumerate}
\caption{Our spectrum estimation algorithm, tailored to rank-$r$ inputs.}
\label{fig:spectrum-estimation-algorithm-rank-r}
\end{figure}
}

\begin{proposition}[Spectrum estimation for rank-$r$ states]
    \label{thm:spectrum_learning_rank-r}
    Let $\widehat{\bbeta}$ be the output of the spectrum estimation algorithm from~\Cref{fig:spectrum-estimation-algorithm-rank-r} when run on $n$ copies of the mixed state $\rho \in \C^{d \times d}$. Then, with high constant probability, $\dtv{\widehat{\bbeta}}{\spec(\rho)} \leq \epsilon$ so long as
    \begin{equation*}
        n = O\Big(\frac{r^2 \cdot (\log \log r)^2}{\epsilon^4 \cdot (\log r)^2}\Big).
    \end{equation*}
\end{proposition}

\begin{proof}
    For any purification $\ket{\brho}$, the states $\rho = \tr_{\reg{2}}(\ketbra{\brho})$ and $\bsigma = \tr_{\reg{1}}(\ketbra{\brho})$ have the same nonzero spectrum, which we denote $\alpha = (\alpha_1, \dots, \alpha_r)$. Moreover, $\spec(\rho) = \spec(\bsigma) \| 0^{d-r} = (\alpha_1, \dots, \alpha_r, 0, \dots, 0)$. Since $\rho$ is rank-$r$, we have
    \begin{equation*}
        \dtv{\widehat{\bbeta}}{\spec(\rho)} = \frac{1}{2}\sum_{i=1}^{d} \Abs{\widehat{\bbeta}_i - \alpha_i} = \frac{1}{2}\sum_{i=1}^{r} \Abs{\widehat{\balpha}_i - \alpha_i} = \dtv{\widehat{\balpha}}{\spec(\bsigma)}.
    \end{equation*}
    With $n$ samples, we have $\dtv{\widehat{\balpha}}{\spec(\bsigma)} \leq \epsilon$ with high probability, by \Cref{thm:main}.
\end{proof}

\section{Further applications of the relative error bound}\label{sec:additional_tomography_results}
In this section, we show how our relative error bound can be used to obtain additional results in quantum state tomography.

Our first two results concern principal component analysis (PCA). In PCA, the goal is to output the best rank-$k$ approximation in some distance measure, given $k \in [d]$ as input. In \Cref{subsection_PCA_td}, we consider the case where our distance measure is trace distance. Trace distance PCA was previously considered in \cite{OW16}, in which the authors showed that $n = O(kd/\epsilon^2)$ copies suffice using the representation-theoretic Keyl's algorithm. Here, we reprove this result with a new analysis, leveraging our new bounds from \Cref{sec:moments}. In \Cref{subsection_PCA_fidelity}, we consider the case where our distance measure is the more challenging Bures distance. We show that $n = O(kd/\epsilon^2)$ copies also suffice for Bures distance PCA, which is a new result.

Our final result concerns full state tomography with respect to Bures $\chi^2$-divergence, $\Dchi$. Bures $\chi^2$-divergence is a more challenging distance measure to learn in than the more common Bures and trace distances, as $\Dtr(\rho,\sigma)^2 \leq \DBur(\rho,\sigma)^2 \leq \Dchi(\rho\|\sigma)$. The best previous bound for learning in this distance measure is due to Flammia and O'Donnell~\cite{FO24}, who showed that $n = \widetilde{O}(\sqrt{rd^3}/\eps)$ copies suffice to learn to error $\epsilon$. In \Cref{subsection_Bures_chi-squared}, we remove the tilde and show that $O(\sqrt{rd^3}/\eps)$ copies suffice.

All of these results follow from applications of the machinery developed in \Cref{sec:moments}.

\subsection{Principal component analysis in trace distance} \label{subsection_PCA_td}

In this section, we recover the fact that trace distance PCA can be solved using $O(kd/\eps^2)$ copies, using the techniques developed in this paper. This result was previously obtained in \cite{OW16}, using the representation-theoretic Keyl's algorithm.

We begin by formally defining the relevant distance measure.

\begin{definition}[Trace distance PCA error]
    Let $\rho \in \C^{d \times d}$ be a mixed state, with eigenvalues $\alpha_1, \dots, \alpha_d$.
    A rank-$k$ matrix $\widehat{\rho}$ is said to be a rank-$k$ approximation to $\rho$ with \emph{trace distance PCA error $\eps$} if
    \begin{equation*}
        \Dtr(\rho, \widehat{\rho}) \leq \Dtr( \rho, \widehat{\rho}_{\mathrm{opt}}) + \epsilon = \frac{1}{2} ( \alpha_{k+1} + \dots + \alpha_d) + \epsilon.
    \end{equation*}
    Here $\widehat{\rho}_{\mathrm{opt}}$ is the optimal rank-$k$ approximation to $\rho$, i.e.\ the restriction of $\rho$ to the support of its largest $k$ eigenvectors.
\end{definition}

\begin{proposition}\label{prop:trace_distance_PCA}
    Let $\rho \in \C^{d \times d}$ be a mixed state with eigenvalues $\alpha_1, \dots, \alpha_d$, and let $\widehat{\brho}$ be the output of $\mixed(\gps)$, given $n$ copies of $\rho$. Then, with probability at least $0.99$,
    \begin{equation*}
        \Dtr(\rho,\widehat{\brho}_{\leq k}) \leq \frac{1}{2} (\alpha_{k+1} + \dots + \alpha_d) + O\big(\sqrt{kd/n}\big) + O\big(kd/n\big).
    \end{equation*}
Thus, with $n = O(kd/\epsilon^2)$ copies of a mixed state $\rho \in \C^{d \times d}$, $ \widehat{\brho}_{\leq k}$ is a rank-$k$ approximation to $\rho$ with trace distance PCA error $\eps$, with high constant probability.
\end{proposition}

\begin{proof}[Proof of \Cref{prop:trace_distance_PCA}]
    We condition on the relative error bound holding. Let $\bPi$ be the projector onto the largest $k$ eigenvectors of $\widehat{\brho}$, settling ties arbitrarily (but consistently with $\widehat{\brho}_{\leq k}$, so that $\widehat{\brho}_{\leq k} = \bPi \widehat{\brho} \bPi$). We start by writing
    \begin{align}
        \Dtr(\rho, \bPi \widehat{\brho} \bPi) & \leq \Dtr(\rho, \bPi \rho \bPi) + \Dtr(\bPi \rho \bPi, \bPi \widehat{\brho} \bPi) \nonumber \\
        & = \frac{1}{2} \big( \norm{ \rho - \bPi \rho \bPi}_1 + \norm{\bPi ( \rho - \widehat{\brho}) \bPi}_1 \big) \nonumber \\
        & \leq \frac{1}{2} \big( \norm{ \overline{\bPi} \rho \overline{\bPi}}_1 + \norm{\bPi \rho \overline{\bPi} + \overline{\bPi} \rho \bPi}_1 +  \norm{\bPi ( \rho - \widehat{\brho}) \bPi}_1 \big). \label{eq:sum_of_three_trace_norms}
    \end{align}
    To get the second inequality, we have used
    \begin{equation*}
        \rho - \bPi \rho \bPi = (\bPi + \overline{\bPi}) \rho (\bPi + \overline{\bPi}) - \bPi \rho \bPi =  \overline{\bPi} \rho \overline{\bPi} + \big(\bPi \rho \overline{\bPi} + \overline{\bPi} \rho \bPi\big)
    \end{equation*}
    and the triangle inequality. We proceed by bounding the three terms of \Cref{eq:sum_of_three_trace_norms} individually.

    We start with the first term. Since $\tr( \bPi \rho \bPi) + \tr( \overline{\bPi} \rho \overline{\bPi}) = 1 = \tr(\rho_{\leq k}) + \tr(\rho_{>k})$, we have
    \begin{equation*}
        \norm{ \overline{\bPi} \rho \overline{\bPi} }_1 = \tr(\overline{\bPi} \rho \overline{\bPi}) = \tr(\rho_{>k}) + \big( \tr(\rho_{\leq k}) - \tr( \bPi \rho \bPi) \big).
    \end{equation*}
    However, $\tr(\rho_{>k}) = \alpha_{k+1} + \dots + \alpha_d$, so it suffices to bound the difference in parentheses. To do so, start by noting that $\lambda_i(\rho) \geq \lambda_i(\bPi \rho \bPi)$ by Cauchy's interlacing theorem, and hence
    \begin{equation*}
        \tr(\rho_{\leq k}) - \tr( \bPi \rho \bPi) = \sum_{i=1}^k \big(\lambda_i (\rho) - \lambda_i(\bPi \rho \bPi)\big) = \sum_{i=1}^{k} \ABs{\lambda_i (\rho) - \lambda_i(\bPi \rho \bPi)} = 2 \cdot \Dtv \big( \spec(\rho)_{\leq k}, \spec(\bPi \rho \bPi)_{\leq k} \big).
    \end{equation*}
    However, by \Cref{it:spec_lemma_3} of \Cref{lem:dtv_lemma_k},
    $\Dtv \big( \spec(\rho)_{\leq k}, \spec(\bPi \rho \bPi)_{\leq k} \big) \leq O\big(\sqrt{k d/n} \big) + O \big( kd/n \big).$
    Putting everything together, we get
    \begin{equation}\label{eq:first_term}
        \norm{ \overline{\bPi} \rho \overline{\bPi} }_1  \leq (\alpha_{k+1} + \dots + \alpha_{d} ) + O\big(\sqrt{k d/n} \big) + O \big( kd/n \big).
    \end{equation}

    For the second term in \Cref{eq:sum_of_three_trace_norms}, the matrix $\bPi \rho \overline{\bPi} + \overline{\bPi} \rho \bPi$ is Hermitian and has rank at most $2k$. Thus, there are eigenvectors $\ket{\bu_1}, \dots, \ket{\bu_{2k}}$ such that
    \begin{equation*}
        \norm{\bPi \rho \overline{\bPi} + \overline{\bPi} \rho \bPi }_1 = \sum_{i=1}^{2k} \ABs{\bra{\bu_i} (\bPi \rho \overline{\bPi} + \overline{\bPi} \rho \bPi) \ket{\bu_i}} = \sum_{i=1}^{2k} \ABs{\tr( \bO_i \cdot \rho )} = \sum_{i=1}^{2k} \ABs{\tr( \bO_i \cdot (\rho - \widehat{\brho}) )}.
    \end{equation*}
    Here, $\bO_i \coloneq \bPi \ketbra{\bu_i} \overline{\bPi} + \overline{\bPi} \ketbra{\bu_i} \bPi$ is an observable of rank at most $2$. The last step follows from $\tr( \bO_i \cdot \widehat{\brho}) = 0$, because $\bPi$ and $\overline{\bPi}$ commute through $\widehat{\brho}$, and $\bPi \cdot \overline{\bPi} = \overline{\bPi} \cdot \bPi = 0$. Applying \Cref{thm:uniform-statement-for-observables} with $O \leftarrow \bO_i$ gives
    \begin{align}
        \norm{\bPi \rho \overline{\bPi} + \overline{\bPi} \rho \bPi }_1 & \leq C \cdot \sum_{i=1}^{2k} \sqrt{ \frac{d}{n} \cdot 2 \cdot \Big( \tr( \bO_i^2 \cdot \rho) + \frac{d}{n} \cdot \tr( \bO_i^2)\Big)} \nonumber \\
        & \leq C \cdot \sum_{i=1}^{2k} \Bigg( \sqrt{\frac{2d}{n} \cdot \tr(\bO_i^2 \cdot \rho)} + \frac{d}{n} \sqrt{2 \tr(\bO_i^2)}\Bigg) \nonumber \\
        & \leq C \cdot \Bigg( \sqrt{ \frac{4kd}{n} \cdot \sum_{i=1}^{2k} \tr(\bO_i^2 \cdot \rho)} + \frac{d}{n} \sum_{i=1}^{2k} \sqrt{2\tr(\bO_i^2)} \Bigg). \label{eq:second_term}
    \end{align}
    To simplify this, we use
    \begin{align*}
        \bO_i^2 & = \bPi \ketbra{\bu_i} \overline{\bPi} \ketbra{\bu_i} \bPi + \overline{\bPi} \ketbra{\bu_i} \bPi \ketbra{\bu_i} \overline{\bPi} \\
        & = \bra{\bu_i} \overline{\bPi} \ket{\bu_i} \cdot \bPi \ketbra{\bu_i} \bPi + \bra{\bu_i} {\bPi} \ket{\bu_i} \cdot \overline{\bPi} \ketbra{\bu_i} \overline{\bPi}.
    \end{align*}
    Thus,
    \begin{equation*}
        \tr(\bO_i^2) = 2 \cdot \bra{\bu_i} \bPi \ket{\bu_i} \cdot \bra{\bu_i} \overline{\bPi} \ket{\bu_i} \leq 2.
    \end{equation*}
    Moreover,
    \begin{align*}
        \sum_{i=1}^{2k} \tr(\bO_i^2 \cdot \rho) & = \sum_{i=1}^{2k} \Big(\bra{\bu_i} \overline{\bPi} \ket{\bu_i} \cdot \bra{\bu_i} \bPi \rho \bPi \ket{\bu_i} + \bra{\bu_i} {\bPi} \ket{\bu_i} \cdot \bra{\bu_i} \overline{\bPi} \rho \overline{\bPi} \ket{\bu_i}\Big) \\
        & \leq \sum_{i=1}^{2k} \Big( \bra{\bu_i} \bPi \rho \bPi \ket{\bu_i} +  \bra{\bu_i} \overline{\bPi} \rho \overline{\bPi} \ket{\bu_i}\Big) \\
        & \leq \tr( \bPi \rho \bPi) +  \tr(\overline{\bPi} \rho \overline{\bPi}) = \tr( (\bPi + \overline{\bPi}) \rho ) = 1.
    \end{align*}
    Using these bounds in \Cref{eq:second_term}, we obtain
    \begin{equation} \label{eq:second_term_final_bound}
       \norm{\bPi \rho \overline{\bPi} + \overline{\bPi} \rho \bPi }_1 \leq C \cdot \Big( \sqrt{ 4kd/n} + 4kd/n \Big) \leq O(\sqrt{kd/n}) + O(kd/n).
    \end{equation}
    Finally, for the third term in \Cref{eq:sum_of_three_trace_norms}, we apply \Cref{cor:1-norm-bound-projectors} with $\Pi \leftarrow \bPi$:
    \begin{align}
        \frac{1}{2} \norm{\bPi ( \rho - \widehat{\brho}) \bPi}_1
        & \leq 2C \cdot \sqrt{ \frac{d}{n} \cdot \mathrm{rank}(\bPi) \cdot \Big( \tr(\bPi \cdot \rho) + \frac{d}{n} \cdot \mathrm{rank}(\bPi) \Big) } \nonumber \\
        & \leq 2C \cdot \sqrt{ \frac{d}{n} \cdot k \cdot \Big( 1 + \frac{d}{n} \cdot k \Big) } \nonumber \\
        & \leq O(\sqrt{kd/n}) + O(kd/n). \label{eq:third_term}
    \end{align}
     Inserting \Cref{eq:first_term,eq:second_term_final_bound,eq:third_term} into \Cref{eq:sum_of_three_trace_norms}, we get
     \begin{equation*}
         \Dtr(\rho, \widehat{\brho}_{\leq k}) \leq \frac{1}{2} ( \alpha_{k+1} + \dots + \alpha_d) + O(\sqrt{kd/n}) + O(kd/n).
     \end{equation*}
     This completes the proof.
\end{proof}

\subsection{Principal component analysis in Bures distance} \label{subsection_PCA_fidelity}

Bures distance is the metric on quantum states given by
\begin{equation*}
    \DBur(\rho, \sigma)^2 \coloneq  2(1 - \mathrm{F}(\rho, \sigma)),
\end{equation*}
where $\mathrm{F}(\rho, \sigma) = \norm{\sqrt{\rho} \sqrt{\sigma} }_1$ is the fidelity.
Full state tomography is more challenging in Bures distance than trace distance, since $\Dtr(\rho,\sigma) \leq \DBur(\rho,\sigma)$, so that algorithms learning in Bures distance automatically learn in trace distance too. It turns out that the same is not necessarily true for PCA.\footnote{In this footnote we give a concrete example of a qubit state $\rho$ and a rank-$1$ estimate $\widehat{\rho}$ such that $\widehat{\rho}$ has Bures distance PCA error (see \Cref{def:bures_PCA}) that is smaller than its trace distance PCA error, for $k=1$. Take $\rho = (1-x) \cdot \ketbra{0} + x \cdot \ketbra{1}$. For $x$ small, the optimal rank-$1$ approximation is $\widehat{\rho}_{\mathrm{opt}} = (1-x) \cdot \ketbra{0}$. Now consider the output $\widehat{\rho} = \ketbra{0}$. We have $\Dtr(\rho, \widehat{\rho}) = x$, so that the trace distance PCA error is $x/2$. However, we have 
\begin{equation*}
    \DBur(\rho, \widehat{\rho}) = \sqrt{2 - 2 \sqrt{1-x}} = \sqrt{x} + x^{3/2}/8 + O(x^{5/2}).
\end{equation*}
so that the Bures distance PCA error is $x^{3/2}/8$. For $x$ small, $x^{3/2}/8 \ll x/2$. 
}
Regardless, algorithms for learning mixed states in Bures distance have always been more difficult to obtain, and there are no prior Bures distance PCA algorithms in the literature provably achieving the expected $O(kd/\epsilon^2)$ sample complexity. In this section, we give the first such algorithm. First, however, we must introduce the generalization of Bures distance to unnormalized states. Then, we will define Bures distance PCA error, and state and prove our result. 

\begin{definition}[Bures distance for unnormalized states] \label{def:Bures_distance}
    Given two PSD matrices $\rho, \sigma \in \C^{d\times d}$, we define the \emph{Bures distance} between $\rho$ and $\sigma$ to be
    \begin{equation*}
        \DBur(\rho, \sigma)^2 \coloneq \tr(\rho) + \tr(\sigma) - 2\Fid(\rho, \sigma).
    \end{equation*}
\end{definition}
This definition generalizes the Bures distance beyond the case where $\tr(\rho) = \tr(\sigma) = 1$. It is also natural in the following sense: the Bures distance (for normalized states) can be equivalently formulated via the equation
\begin{equation*}
    \DBur(\rho, \sigma) = \min_{U \in U(d)} \norm{ \sqrt{\rho} - \sqrt{\sigma} U }_2,
\end{equation*}
and this formula extends as-is for unnormalized states, recovering \Cref{def:Bures_distance}. For a more thorough treatment of this metric, we refer the reader to~\cite{BJL19}.

We next define Bures distance PCA error.

\begin{definition}[Bures distance PCA error] \label{def:bures_PCA}
    Let $\rho \in \C^{d \times d}$ be a mixed state, with eigenvalues $\alpha_1, \dots, \alpha_d$.
    A PSD rank-$k$ matrix $\widehat{\rho}$ is said to be a rank-$k$ approximation to $\rho$ with \emph{Bures distance PCA error $\eps$} if
    \begin{equation*}
        \DBur(\rho, \widehat{\rho}) \leq \DBur( \rho, \widehat{\rho}_{\mathrm{opt}}) + \epsilon = \sqrt{\alpha_{k+1} + \dots + \alpha_d} + \epsilon.
    \end{equation*}
    Here $\widehat{\rho}_{\mathrm{opt}}$ is the optimal rank-$k$ approximation to $\rho$, i.e.\ the restriction of $\rho$ to the support of its largest $k$ eigenvectors.\footnote{Note that this differs from the notion of fidelity PCA introduced in \cite{PTTW25}. There, the authors defined fidelity PCA error as
    \begin{equation*}
        \sum_{i=1}^{k} \alpha_i + \tr( \widehat{\rho} ) - 2\Fid(\rho, \widehat{\rho}) = \DBur(\rho, \widehat{\rho})^2 - \sum_{i = k+1}^{d} \alpha_i \leq \epsilon.
    \end{equation*}
    That is, the authors considered a version of PCA based on the ``Bures distance-squared'', which is proportional to the infidelity, rather than our ``unsquared'' version. The authors conjectured that $\widetilde{O}(kd/\epsilon)$ copies sufficed to learn in fidelity PCA error. However, this is false even for classical two-dimensional states. Consider the special case where one is promised one of the two states $\rho_{\pm} \coloneq \mathrm{diag}(1/2 \pm \eps, 1/2 \mp \eps)$, and take $k=1$. If one can perform fidelity PCA with constant success probability, then one can also distinguish the two states with high constant probability. For any rank-$1$ (not necessarily normalized, but WLOG real-valued) output $\ket{u} = u_1 \ket{1} + u_2 \ket{2}$, the fidelity PCA error for, say, $\rho_+$ is
    \begin{equation*}
        (1/2 + \eps) + ( u_1^2 + u_2^2 ) - 2 u_1 \sqrt{1/2 + \eps} = (1/2 + \epsilon - u_1 )^2 + u_2^2.
    \end{equation*}
    If this is at most $\eps$ with high probability, then $u_1 \geq 1/2$, and $u_2 \leq \epsilon$. Similar reasoning holds for when given $\rho_-$: we have $u_1 \leq \epsilon$, and $u_2 \geq 1/2$. For $\eps$ small, we can then distinguish the two cases with the same probability. However, since $F(\rho_+, \rho_-) = 1 - 4\eps^2$, we require $n = \Omega(1/\eps^2)$ copies to distinguish.
    }
\end{definition}

To solve Bures distance PCA, we must output a PSD operator, so $\widehat{\brho}_{\leq k}$ does not suffice: $\widehat{\brho}$, the output of $\mixed(\gps)$, may have negative eigenvalues. However, it is easy to convert $\widehat{\brho}$ to a PSD operator. For full-rank inputs, the output of $\mixed(\gps)$ (see \Cref{fig:mix_gps}) has the form
\begin{equation*}
    \widehat{\brho} = \frac{d^2 + n}{n} \cdot \tr_{\reg{2}}(\ketbra{\bu}) - \frac{1}{n} \cdot \tr_{\reg{2}}(I_{d} \otimes I_d) = \frac{d^2 + n}{n} \cdot \tr_{\reg{2}}(\ketbra{\bu}) - \frac{d}{n} \cdot I_{d},
\end{equation*}
for some pure state $\ket{\bu}$. Since $\ketbra{\bu}$ is PSD, so too is $\tr_{\reg{2}}(\ketbra{\bu})$. The problematic term is $-(d/n) \cdot I_d$. So, we can force our output to be PSD by adding a multiple of the identity. For technical reasons that will become clear in our proof,
we will also want our estimator's eigenvalues to be at least $d/n$. For this reason, we will slightly modify our original estimator and instead output
\begin{equation} \label{def:tilde_rho}
    \widetilde{\brho} \coloneq \widehat{\brho} + \frac{2d}{n} \cdot I_d.
\end{equation}
The minimum eigenvalue of $\widetilde{\brho}$ is always at least $(2d-d)/n = d/n$.



We will show the following result.

\begin{proposition}\label{prop:Bures_distance_PCA}
    Let $\rho \in \C^{d \times d}$ be a mixed state,  with eigenvalues $\alpha_1, \dots, \alpha_d$, and let $\widehat{\brho}$ be the output of $\mixed(\gps)$, given $n$ copies of $\rho$. Form $\widetilde{\brho}$ as in \Cref{def:tilde_rho}. Then with probability at least $0.99$,
    \begin{equation*}
        \DBur(\rho,\widetilde{\brho}_{\leq k}) \leq \sqrt{\alpha_{k+1} + \dots + \alpha_d} + O\big(\sqrt{kd/n}\big).
    \end{equation*}
Thus, with $n = O(kd/\epsilon^2)$ copies of a mixed state $\rho \in \C^{d \times d}$, $\widetilde{\brho}_{\leq k}$ is a rank-$k$ approximation to $\rho$ with Bures PCA error $\eps$, with high probability.
\end{proposition}




\begin{proof}[Proof of \Cref{prop:Bures_distance_PCA}]
Throughout, we condition on the relative error bound. Let $\Pi$ be the projector onto the top $k$ eigenvectors of $\rho$, and $\bPi$ be the projector onto the top $k$ eigenvectors of $\widetilde{\brho}$. Since $\DBur$ is a metric on PSD matrices, by the triangle inequality,
\begin{equation*}
    \DBur(\rho, \widetilde{\brho}_{\leq k}) \leq \DBur(\rho, \bPi \rho \bPi ) + \DBur(\bPi \rho \bPi, \widetilde{\brho}_{\leq k} ).
\end{equation*}
In \Cref{lem:Bures_PCA_aux_lemma} below, we show that $\DBur(\bPi \rho \bPi, \widetilde{\brho}_{\leq k} ) \leq O( \sqrt{kd/n})$. In the remainder of the proof, we show that $\DBur(\rho, \bPi \rho \bPi ) \leq \sqrt{\alpha_{k+1} + \dots + \alpha_d} + O(\sqrt{kd/n})$. The proposition then follows from these two facts.

First, note that
\begin{equation*}
    \Fid(\rho, \bPi \rho \bPi) = \tr\Big( \sqrt{\sqrt{\rho} \cdot \bPi \rho \bPi \cdot \sqrt{\rho}} \Big) = \tr\Big( \sqrt{ \big(\sqrt{\rho} \bPi \sqrt{\rho} \big) \cdot \big(\sqrt{\rho} \bPi \sqrt{\rho}\big)} \Big) = \tr \big( \sqrt{\rho} \bPi \sqrt{\rho}\big) = \tr( \rho \bPi).
\end{equation*}
Thus,
\begin{equation*}
    \DBur(\rho, \bPi \rho \bPi)^2 = \tr(\rho) + \tr(\bPi \rho \bPi) - 2 \Fid(\rho, \bPi \rho \bPi)  = 1 - \tr(\rho \bPi) = \tr(\rho \overline{\bPi}).
\end{equation*}
We now rewrite this as
\begin{equation} \label{eq:Bur_rho_Pi_rho_Pi}
    \DBur(\rho, \bPi \rho \bPi)^2 = \tr(\rho \overline{\Pi}) + \tr\big(\rho \cdot (\overline{\bPi} - \overline{\Pi} ) \big),
\end{equation}
to isolate the term $\tr(\rho \overline{\Pi}) = \alpha_{k+1} + \dots + \alpha_d$. The second term is nonnegative, since $\tr(\rho \cdot (\overline{\bPi} - \overline{\Pi})) = \tr(\rho \cdot (\Pi - {\bPi}))$ and Ky Fan's maximum principle tells us $\tr(\rho \bPi ) \leq \alpha_{1} + \dots + \alpha_k = \tr(\rho \Pi)$.

We now bound $\tr(\rho \cdot (\overline{\bPi} - \overline{\Pi}))$. Observe that $\tr(\widehat{\brho} \cdot (\bPi - \Pi)) \geq 0$ as well, so that
\begin{equation*}
    \tr\big( \rho \cdot (\overline{\bPi} - \overline{\Pi} ) \big) = \tr\big(\rho \cdot (\Pi - \bPi ) \big) \leq \tr\big(\rho \cdot (\Pi - \bPi ) \big) + \tr(\widehat{\brho} \cdot (\bPi - \Pi)) = \tr\big( (\rho - \widehat{\brho}) \cdot (\Pi - \bPi) \big).
\end{equation*}
We now apply \Cref{thm:uniform-statement-for-observables} with $O \leftarrow \bO \coloneq \Pi - \bPi$, obtaining:
\begin{equation*}
    \tr\big(\rho \cdot (\Pi - \bPi ) \big)  \leq 2C \cdot \sqrt{ \frac{d}{n} \cdot \mathrm{rank}(\bO) \cdot \Big( \tr(\bO^2 \cdot \rho) + \frac{d}{n} \cdot \tr( \bO^2) \Big) }.
\end{equation*}
To simplify this, note that $\rank(\bO) \leq 2k$. We also have $\tr(\bO^2) \leq 2k$, since each eigenvalue of $\bO$ is in the range $[-1,1]$. Moreover,
\begin{equation*}
    \bO^2 = (\Pi - \Pi \cdot \bPi) + (\bPi - \bPi \cdot \Pi) = \Pi \cdot ( I- \bPi) + \bPi \cdot ( I- \Pi) = \Pi \cdot \overline{\bPi} + \bPi \cdot \overline{\Pi}.
\end{equation*}
Thus,
\begin{align*}
    \tr(\bO^2 \cdot \rho) & =  \tr( \overline{\bPi} \cdot \rho \cdot \Pi ) + \tr( \overline{\Pi} \cdot \rho \cdot \bPi) \\
    & = \tr\big(  \overline{\bPi} \cdot \sqrt{\rho} \cdot \Pi \cdot \sqrt{\rho} \cdot \overline{\bPi} \big) + \tr\big( \overline{\Pi} \cdot \sqrt{\rho} \cdot \bPi \cdot \sqrt{\rho} \cdot \overline{\Pi}\big) \\
    & \leq \tr\big(  \overline{\bPi} \cdot \sqrt{\rho} \cdot I \cdot \sqrt{\rho} \cdot \overline{\bPi} \big) + \tr\big( \overline{\Pi} \cdot \sqrt{\rho} \cdot I \cdot \sqrt{\rho} \cdot \overline{\Pi}\big) \\
    & = \tr( \rho \overline{\bPi} ) + \tr( \rho \overline{\Pi} ).
\end{align*}
In the second step, we are using the fact that $\rho$ commutes with $\Pi$ and $\overline{\Pi}$; in the third step, we use that $\tr( A^\dagger P A ) \leq \tr( A^\dagger Q A)$, if $P, Q$ are PSD matrices such that $P \preceq Q$ in the PSD order, and $A$ is any matrix. Thus
\begin{align*}
    \tr\big(\rho \cdot (\overline{\bPi} - \overline{\Pi} ) \big)  & \leq 2C \cdot \sqrt{ \frac{2kd}{n} \cdot \Big( \tr( \rho \overline{\bPi} ) + \tr( \rho \overline{\Pi} ) + \frac{2kd}{n} \Big) } \\
    & = 2C \cdot \sqrt{ \frac{2kd}{n} \cdot \Big( \tr\big(\rho \cdot (\overline{\bPi} - \overline{\Pi} ) \big) + 2 \tr(\rho \overline{\Pi}) + \frac{2kd}{n} \Big) } \\
    & \leq 2C \cdot \sqrt{ \frac{2kd}{n} \cdot \tr\big(\rho \cdot (\overline{\bPi} - \overline{\Pi} ) \big)} + 2C \cdot \sqrt{\frac{4kd}{n} \cdot \tr(\rho \overline{\Pi})} + 4C \cdot \frac{kd}{n} \\
    & \leq \Big( 4C^2 \cdot \frac{kd}{n} + \frac{1}{2}  \tr\big(\rho \cdot (\overline{\bPi} - \overline{\Pi} ) \big)\Big) + 2C \cdot \sqrt{\frac{4kd}{n} \cdot \tr(\rho \overline{\Pi})} + 4C \cdot \frac{kd}{n}.
\end{align*}
Rearranging then implies
\begin{equation*}
    \tr\big(\rho \cdot (\overline{\bPi} - \overline{\Pi} ) \big) \leq 8C \cdot \sqrt{\frac{kd}{n}} \cdot \sqrt{\tr(\rho \overline{\Pi} )} + (8C^2 + 8C) \cdot \frac{kd}{n}.
\end{equation*}
Substituting this bound back into \Cref{eq:Bur_rho_Pi_rho_Pi} gives
\begin{align*}
    \DBur(\rho, \bPi \rho \bPi)^2 & \leq \tr(\rho \overline{\Pi}) + 8C \cdot \sqrt{\frac{kd}{n}} \cdot \sqrt{\tr(\rho \overline{\Pi} )} + (8C^2 + 8C) \cdot \frac{kd}{n} \\
    & \leq \tr(\rho \overline{\Pi}) + 8C \cdot \sqrt{\frac{kd}{n}} \cdot \sqrt{\tr(\rho \overline{\Pi} )} + 16C^2 \cdot \frac{kd}{n} \\
    & = \Bigg( \sqrt{ \tr(\rho \overline{\Pi})} + 4C \cdot \sqrt{ \frac{kd}{n} } \Bigg)^2.
\end{align*}
In the second line, we have used $C \geq 1$ (see \Cref{rem:C_lower_bound}). Taking square roots yields
\begin{equation*}
    \DBur(\rho, \bPi \rho \bPi) \leq \sqrt{ \tr(\rho \overline{\Pi})} + 4C \cdot \sqrt{ \frac{kd}{n} } = \sqrt{\alpha_{k+1} + \dots + \alpha_d} + 4C \cdot \sqrt{ \frac{kd}{n}},
\end{equation*}
which is what we wanted to show.
\end{proof}

To complete our proof, we need to state and prove a deferred lemma.

\begin{lemma} \label{lem:Bures_PCA_aux_lemma}
    Conditioned on the relative error bound holding,
        $\DBur(\bPi \rho \bPi, \widetilde{\brho}_{\leq k} ) \leq O( \sqrt{kd/n}).$

\end{lemma}

Before proving this lemma, we introduce one useful quantity; when working with Bures distance, it is often convenient to bound the Bures distance with the Bures $\chi^2$-divergence.

\begin{definition}[Bures $\chi^2$-divergence]
    Let $\rho, \sigma \in \C^{d \times d}$ be PSD operators. Suppose $\sigma$ is diagonal in the standard basis, $\sigma = \mathrm{diag}(\sigma_1, \dots, \sigma_d)$, and $\mathrm{supp}(\rho) \subseteq \mathrm{supp}(\sigma)$. Then we define the \emph{Bures $\chi^2$-divergence} of $\rho$ and $\sigma$ to be
    \begin{equation*}
        \Dchi(\rho\|\sigma) \coloneq \sum_{i,j=1}^{d} \frac{2}{\sigma_i + \sigma_j} \cdot \Abs{\rho_{ij}-\sigma_{ij}}^2.
    \end{equation*}
    When $\sigma$ is not diagonal in the standard basis, we extend the above formula via unitary invariance. When $\mathrm{supp}(\rho)$ is not contained in $\mathrm{supp}(\sigma)$, we take $\Dchi(\rho\|\sigma) = \infty$.
\end{definition}

We have the following bound which relates the Bures distance to the Bures $\chi^2$-divergence.

\begin{lemma}[Bures distance and Bures $\chi^2$-divergence] \label{lem:Bures_dist_to_Bures_chi_squared}
    Let $\rho, \sigma \in \C^{d \times d}$ be two PSD operators. Then
    \begin{equation*}
        \DBur(\rho, \sigma)^2 \leq \Dchi(\rho\|\sigma).
    \end{equation*}
\end{lemma}

\begin{proof}
   The inequality holds in the special case where $\tr(\rho) = \tr(\sigma) = 1$ \cite[Proposition 2.31]{FO24}. For non-normalized states, take $x > \max(\tr(\rho), \tr(\sigma))$. We will define states $\rho'$ and $\sigma'$ on a larger Hilbert space $\C^d \oplus \C$, where we denote the new basis element as $\ket{0}$. Let $r = \tr(\rho)$ and $s = \tr(\sigma)$, and set
   \begin{equation*}
       \rho' \coloneq \frac{\rho}{x} + \bigg( 1 - \frac{r}{x} \bigg) \cdot \ketbra{0}, \qquad  \sigma' \coloneq \frac{\sigma}{x} + \bigg( 1 - \frac{s}{x} \bigg) \cdot \ketbra{0}.
   \end{equation*}
   Note that $\tr(\rho') = \tr(\sigma') = 1$. The states $\rho'$ and $\sigma'$ are simultaneously block-diagonal, and since Bures distance squared and chi-squared divergence can be computed block-by-block, we have
   \begin{align*}
       \DBur(\rho', \sigma')^2 & = \frac{1}{x} \cdot \DBur(\rho,\sigma)^2 + \bigg( \Big(1 - \frac{r}{x}\Big) + \Big(1 - \frac{s}{x} \Big) - 2 \Big(1 - \frac{r}{x}\Big)^{1/2} \cdot \Big(1 - \frac{s}{x}\Big)^{1/2} \bigg) \\
       & = \frac{1}{x} \cdot \DBur(\rho,\sigma)^2 + \bigg( \Big(1 - \frac{r}{x}\Big)^{1/2} - \Big(1 - \frac{s}{x}\Big)^{1/2} \bigg)^2,
   \end{align*}
   and
   \begin{equation*}
       \Dchi(\rho' || \sigma') = \frac{1}{x}\cdot \Dchi(\rho || \sigma) + \bigg( 1 - \frac{s}{x} \bigg)^{-1} \cdot \bigg( \frac{r - s}{x} \bigg)^2.
   \end{equation*}
   Applying the trace-one case to $\rho'$ and $\sigma'$, gives us
   \begin{equation*}
       \frac{1}{x} \cdot \DBur(\rho,\sigma)^2 + \bigg( \Big(1 - \frac{r}{x}\Big)^{1/2} - \Big(1 - \frac{s}{x}\Big)^{1/2} \bigg)^2 \leq  \frac{1}{x}\cdot \Dchi(\rho || \sigma) + \bigg( 1 - \frac{s}{x} \bigg)^{-1} \cdot \bigg( \frac{r - s}{x} \bigg)^2,
   \end{equation*}
   which we can rearrange to get
   \begin{equation}
       x \cdot \bigg( \Big(1 - \frac{r}{x}\Big)^{1/2} - \Big(1 - \frac{s}{x}\Big)^{1/2} \bigg)^2 - x \cdot \bigg( 1 - \frac{s}{x} \bigg)^{-1} \cdot \bigg( \frac{r - s}{x} \bigg)^2 \leq  \Dchi(\rho || \sigma) - \DBur(\rho,\sigma)^2. \label{eq:pre-limit}
   \end{equation}
   This equation holds for any $x > \max(r,s)$. What happens as we take $x \to \infty$? On the left, we have two $x$-dependent terms. For the first term, we have
   \begin{equation*}
       \Big(1 - \frac{r}{x}\Big)^{1/2} - \Big(1 - \frac{s}{x}\Big)^{1/2} = \frac{\big( 1 - \frac{r}{x} \big) - \big( 1 - \frac{s}{x} \big)}{\big(1 - \frac{r}{x}\big)^{1/2} + \big(1 - \frac{s}{x}\big)^{1/2}} = \frac{s-r}{x} \cdot \frac{1}{\big(1 - \frac{r}{x}\big)^{1/2} + \big(1 - \frac{s}{x}\big)^{1/2}},
   \end{equation*}
   so that
   \begin{equation*}
       \lim_{x \to \infty} \bigg[x \cdot \bigg( \Big(1 - \frac{r}{x}\Big)^{1/2} - \Big(1 - \frac{s}{x}\Big)^{1/2} \bigg)^2 \bigg] = \lim_{x \to \infty} \bigg[ \frac{(s-r)^2}{x} \bigg] \cdot \lim_{x \to \infty} \bigg[ \Big(1 - \frac{r}{x}\Big)^{1/2} + \Big(1 - \frac{s}{x}\Big)^{1/2}\bigg]^{-2} = 0 \cdot \frac{1}{4} = 0.
   \end{equation*}
   The second term is easier:
   \begin{equation*}
       \lim_{x \to \infty} \bigg[ x \cdot \bigg( 1 - \frac{s}{x} \bigg)^{-1} \cdot \bigg( \frac{r - s}{x} \bigg)^2\bigg]  = \lim_{x \to \infty} \bigg[ 1 - \frac{s}{x}  \bigg]^{-1} \cdot \lim_{x \to \infty} \bigg[ \frac{(r-s)^2}{x}\bigg] = 1 \cdot 0 = 0.
   \end{equation*}
   Thus, taking the $x \to \infty$ limit of \Cref{eq:pre-limit} gives
   \begin{equation*}
       0 \leq \Dchi(\rho || \sigma) - \DBur(\rho,\sigma)^2.
   \end{equation*}
   This completes the proof.
\end{proof}

We are now ready to prove the deferred lemma.

\begin{proof}[Proof of \Cref{lem:Bures_PCA_aux_lemma}]
    Recall that our estimator is $\widetilde{\brho}_{\leq k}$, with
    $\widetilde{\brho} = \widehat{\brho} + (2d/n) \cdot I_d$,
    and where $\widehat{\brho}$ is the output of $\mixed(\gps)$ given $n$ copies of $\rho$. Let $\widehat{\brho}$ have eigenvectors $\{ \ket{\bv_i} \}_{i \in [d]}$, ordered by decreasing eigenvalue, so that $\bPi = \sum_{i=1}^k \ketbra{\bv_i}$. These are also the top $k$ eigenvectors of $\widetilde{\brho}$. Since $\mathrm{supp}(\bPi \rho \bPi) \subseteq \mathrm{supp}(\bPi \widetilde{\brho} \bPi )$, \Cref{lem:Bures_dist_to_Bures_chi_squared} tells us that
    \begin{align}
        \DBur\big(\bPi \rho \bPi, \widetilde{\brho}_{\leq k}\big)^2 & = \DBur\big(\bPi \rho \bPi, \bPi \widetilde{\brho} \bPi\big)^2 \leq \Dchi\big( \bPi \rho \bPi \| \bPi \widetilde{\brho} \bPi \big) \nonumber \\
        & = \sum_{i,j=1}^k \frac{2}{\widetilde{\balpha}_i + \widetilde{\balpha}_j} \cdot \ABs{\bra{\bv_i} \bPi ( \rho - \widetilde{\brho} )\bPi\ket{\bv_j}}^2 = \sum_{i,j=1}^k \frac{2}{\widetilde{\balpha}_i + \widetilde{\balpha}_j} \cdot \ABs{\bra{\bv_i} ( \rho - \widetilde{\brho} ) \ket{\bv_j}}^2. \label{eq:Bures_aux_1}
    \end{align}
    Here, $\widetilde{\balpha}_i$ is the $i$-th largest eigenvalue of $\widetilde{\brho}$. We will similarly write $\widehat{\balpha}_i$ for the $i$-th largest eigenvalue of $\widehat{\brho}$, but of course $\widetilde{\balpha}_i = \widehat{\balpha}_i + 2d/n$ by construction. It will be more convenient to work with $\widehat{\brho}$, since we can leverage the relative error bound. To this end, we note 
    \begin{align*}
        \ABs{\bra{\bv_i} ( \rho - \widetilde{\brho} ) \ket{\bv_j}}^2 & = \ABs{\bra{\bv_i} (\rho - \widehat{\brho} ) \ket{\bv_j} + \bra{\bv_i} (\widehat{\brho} - \widetilde{\brho}) \ket{\bv_j}}^2 \\
        & \leq 2 \cdot \Big( \ABs{\bra{\bv_i} (\rho - \widehat{\brho} ) \ket{\bv_j}}^2 + \ABs{\bra{\bv_i} (\widehat{\brho} - \widetilde{\brho}) \ket{\bv_j}}^2 \Big) \\
        & = 2 \cdot  \ABs{\bra{\bv_i} (\rho - \widehat{\brho} ) \ket{\bv_j}}^2 + \frac{8d^2}{n^2} \cdot \delta_{ij}.
    \end{align*}
    In the second step above, we have used AM-GM in the form $(a+b)^2 \leq 2a^2+2b^2$. So, we can split the sum in \Cref{eq:Bures_aux_1} into two parts:
    \begin{equation}
        \sum_{i,j=1}^k \frac{2}{\widetilde{\balpha}_i + \widetilde{\balpha}_j} \cdot \ABs{\bra{\bv_i} ( \rho - \widetilde{\brho} )\ket{\bv_j}}^2
        \leq \bigg( \sum_{i,j = 1}^k \frac{4}{\widetilde{\balpha}_i + \widetilde{\balpha}_j} \cdot  \ABs{\bra{\bv_i} (\rho - \widehat{\brho} ) \ket{\bv_j}}^2\bigg) + \bigg( \frac{8d^2}{n^2}\cdot  \sum_{i,j=1}^{k} \frac{2\delta_{ij}}{\widetilde{\balpha}_i + \widetilde{\balpha}_j}\bigg). \label{eq:Bures_aux_2}
    \end{equation}
    The second part is easy to bound: by construction, the minimum eigenvalue of $\widetilde{\brho}$ is at least $d/n$, so that
    \begin{equation}
        \frac{8d^2}{n^2} \cdot \sum_{i,j=1}^{k} \frac{2\delta_{ij}}{\widetilde{\balpha}_i + \widetilde{\balpha}_j} = \frac{8d^2}{n^2} \cdot \sum_{i=1}^k \frac{1}{\widetilde{\balpha}_i} \leq \frac{8kd}{n}. \label{eq:second_part}
    \end{equation}
    For the first part, we start by reorganizing the sum according to the minimum index $m = \min(i,j)$ and then bounding $\widetilde{\balpha}_i + \widetilde{\balpha}_j \geq \widetilde{\balpha}_m$, since the spectrum is sorted in decreasing order:
    \begin{align}
        \sum_{i,j = 1}^k \frac{4}{\widetilde{\balpha}_i + \widetilde{\balpha}_j} \cdot  \ABs{\bra{\bv_i} (\rho - \widehat{\brho} ) \ket{\bv_j}}^2 & = \sum_{m = 1}^k \bigg( \sum_{\substack{i,j = 1 \nonumber \\ \min(i,j) = m}}^k \frac{4}{\widetilde{\balpha}_i + \widetilde{\balpha}_j} \cdot  \ABs{\bra{\bv_i} (\rho - \widehat{\brho} ) \ket{\bv_j}}^2 \bigg) \nonumber \\
        & \leq  \sum_{m = 1}^k \bigg(  \frac{4}{\widetilde{\balpha}_m} \cdot \sum_{\substack{i,j = 1 \\ \min(i,j) = m}}^k  \ABs{\bra{\bv_i} (\rho - \widehat{\brho} ) \ket{\bv_j}}^2 \bigg). \label{eq:chi-squared-sum-rewrite}
    \end{align}
    Using \Cref{lem:Bures_bound_technical_lemma} below, we can further bound this sum by relaxing the inner sum's upper limit to $d$:
    \begin{align*}
        \sum_{m=1}^k \bigg(\frac{4}{\widetilde{\balpha}_m} \cdot \sum_{\substack{i,j=1 \\ \min(i,j)=m}}^d \ABs{\bra{\bv_i} (\rho - \widehat{\brho} ) \ket{\bv_j}}^2 \bigg) & \leq 8(C')^2 \cdot \frac{d}{n} \cdot \sum_{m=1}^{k} \Big(\frac{\widehat{\balpha}_m}{\widetilde{\balpha}_m} +  \frac{d/n}{\widetilde{\balpha}_m} \Big)  \\
        & \leq 8(C')^2 \cdot \frac{d}{n} \cdot \sum_{m=1}^{k} ( 1+1 ) \\
        & = 16(C')^2 \cdot \frac{kd}{n}.
    \end{align*}
    Here, we have used $\widehat{\balpha}_m \leq \widetilde{\balpha}_m$, and $d/n \leq \widetilde{\balpha}_m$. Substituting this and \Cref{eq:second_part} back into \Cref{eq:Bures_aux_2} and then \Cref{eq:Bures_aux_1} yields the desired bound:
    \begin{equation*}
        \DBur(\bPi \rho \bPi, \widetilde{\brho}_{\leq k} )^2 \leq \big(16(C')^2 + 8\big) \cdot \frac{kd}{n} \leq O \big( kd/n \big). \qedhere
    \end{equation*}
\end{proof}

We have one last technical lemma to state and prove. We have separated this result out from the proof of \Cref{lem:Bures_PCA_aux_lemma}, since it will also be useful in the next section.

\begin{lemma} \label{lem:Bures_bound_technical_lemma}
    Assume the relative error bound holds. Then
    \begin{equation*}
        \sum_{\substack{i,j = 1 \\ \min(i,j) = m}}^d  \ABs{\bra{\bv_i} (\rho - \widehat{\brho} ) \ket{\bv_j}}^2 \leq 2(C')^2 \cdot \frac{d}{n} \cdot \Big( \widehat{\balpha}_m + \frac{d}{n} \Big).
    \end{equation*}
\end{lemma}

\begin{proof}
     First,
    \begin{align*}
        \sum_{\substack{i,j = 1 \\ \min(i,j) = m}}^d  \ABs{\bra{\bv_i} (\rho - \widehat{\brho} ) \ket{\bv_j}}^2 & \leq \sum_{i=m}^d \ABs{\bra{\bv_i} (\rho - \widehat{\brho} ) \ket{\bv_m}}^2 + \sum_{j=m}^d \ABs{\bra{\bv_m} (\rho - \widehat{\brho} ) \ket{\bv_j}}^2 \nonumber \\
        & = 2 \sum_{i=m}^d \ABs{\bra{\bv_i} (\rho - \widehat{\brho} ) \ket{\bv_m}}^2.
    \end{align*}
    We claim that this sum can be rewritten as $\norm{ \bM_m }_2^2$, with
    \begin{equation*}
        \bM_m \coloneq \bPi_{\geq m} (\widehat{\brho} - \rho) \bPi_{\geq m} \cdot \ketbra{\bv_m} ,
    \end{equation*}
    and where $\bPi_{\geq m} \coloneq \sum_{i = m}^d \ketbra{\bv_i}$. This can be seen from the computation:
    \begin{align*}
        \norm{ \bPi_{\geq m} (\widehat{\brho} - \rho) \bPi_{\geq m} \cdot \ketbra{\bv_m} }_2^2 & = \sum_{i,j=1}^d \ABs{\bra{\bv_i} \bPi_{\geq m} (\widehat{\brho} - \rho) \bPi_{\geq m}\ket{\bv_m} \cdot \braket{\bv_m}{\bv_j}}^2 \\
        & = \sum_{i,j=1}^d \ABs{\bra{\bv_i}  (\widehat{\brho} - \rho) \ket{\bv_m}}^2 \cdot  1\{m \leq i \leq d\} \cdot \delta_{jm} \\
        & = \sum_{i=m}^d \ABs{\bra{\bv_i} (\rho - \widehat{\brho} ) \ket{\bv_m}}^2.
    \end{align*}
    We would like to bound this norm. For any matrix $M$, with singular values $\{m_i\}$, we have
    \begin{equation*}
        \norm{M}_2^2 = \sum_{i} m_i^2 \leq \rank(M) \cdot \max_i (m_i^2) = \rank(M) \cdot \norm{M}^2_\infty.
    \end{equation*}
    We want to apply this with $M \leftarrow \bM_m$. To this end, we use
    \begin{equation*}
        \mathrm{rank}\big( \bM_m \big) \leq \mathrm{rank}\big(  \ketbra{\bv_m} \big) =1,
    \end{equation*}
    and
    \begin{equation*}
        \norm{ \bM_m }_\infty  \leq \norm{ \bPi_{\geq m} (\widehat{\brho} - \rho) \bPi_{\geq m}}_\infty.
    \end{equation*}
    To bound this operator norm, we use \Cref{cor:uniform-bound-pointwise_rho-hat-version} with $\ket{w} \in \mathrm{supp}(\bPi_{\geq m})$, which gives
    \begin{equation*}
        \ABs{\bra{w} \bPi_{\geq m} (\widehat{\brho} - \rho) \bPi_{\geq m} \ket{w}}  \leq C' \cdot \sqrt{ \frac{d}{n} \cdot \Big( \bra{w} \widehat{\brho} \ket{w} + \frac{d}{n} \Big) }  \leq C' \cdot \sqrt{ \frac{d}{n} \cdot \Big( \widehat{\balpha}_m +\frac{d}{n} \Big) }.
    \end{equation*}
    Since this holds for all such $\ket{w}$, we have the same upper bound on $\norm{\bM_m}_\infty$. We now unwind our steps:
    \begin{equation*}
        \sum_{\substack{i,j = 1 \\ \min(i,j) = m}}^d  \ABs{\bra{\bv_i} (\rho - \widehat{\brho} ) \ket{\bv_j}}^2  \leq 2 \sum_{i=m}^d \ABs{\bra{\bv_i} (\rho - \widehat{\brho} ) \ket{\bv_m}}^2  = 2 \norm{\bM_m}_2^2 \leq 2 \norm{\bM_m}_\infty^2 \leq 2(C')^2 \cdot \frac{d}{n} \cdot \Big( \widehat{\balpha}_m + \frac{d}{n} \Big).
    \end{equation*}
    This completes the proof.
\end{proof}

\subsection{Full state tomography in Bures $\chi^2$-divergence} \label{subsection_Bures_chi-squared}

In this section, we give the first algorithm for learning rank-$r$ states in Bures $\chi^2$-divergence using $O(\sqrt{rd^3}/\eps)$ copies. We already gave a definition of this distance measure in the previous section, but we repeat it here for convenience.

\begin{definition}[Bures $\chi^2$-divergence]
    Let $\rho, \sigma \in \C^{d \times d}$ be PSD operators. Suppose $\sigma$ is diagonal in the standard basis, $\sigma = \mathrm{diag}(\sigma_1, \dots, \sigma_d)$, and $\mathrm{supp}(\rho) \subseteq \mathrm{supp}(\sigma)$. Then we define the \emph{Bures $\chi^2$-divergence} of $\rho$ and $\sigma$ to be
    \begin{equation*}
        \Dchi(\rho\|\sigma) \coloneq \sum_{i,j=1}^{d} \frac{2}{\sigma_i + \sigma_j} \cdot \Abs{\rho_{ij}-\sigma_{ij}}^2.
    \end{equation*}
    When $\sigma$ is not diagonal in the standard basis, we extend the above formula via unitary invariance. When $\mathrm{supp}(\rho)$ is not contained in $\mathrm{supp}(\sigma)$, we take $\Dchi(\rho\|\sigma) = \infty$.
\end{definition}

To produce an estimator with finite Bures $\chi^2$-divergence, we output a full-rank matrix by adding a full-rank PSD operator to $\widehat{\brho}$, the output of $\mixed(\gps)$. We emphasize that even for rank-$r$ states, our tomography algorithm performs full-rank purification. Given $\widehat{\brho}$, we will choose the following output:
\begin{equation}
    \widetilde{\brho}_{\chi^2} \coloneq \bPi \widehat{\brho} \bPi  + \eta_1 \cdot \bPi + \eta_2 \cdot \overline{\bPi}, \label{eq:chi-squared-estimator}
\end{equation}
where $\bPi$ is the projector onto the support of $\widehat{\brho}_{\leq r}$, and the $\eta_1, \eta_2$ are nonnegative noise parameters we will specify later. 
We will show the following result. 

\begin{proposition}
    Let $\rho\in \C^{d \times d}$ be a rank-$r$ mixed state, and let $\widehat{\brho}$ be the output of $\mixed(\gps)$, given $n$ copies of $\rho$. Form $\widetilde{\brho}_{\chi^2} $ as in \Cref{eq:chi-squared-estimator}, with $\eta_1 = 2d/n$ and $\eta_2 = \Theta(\sqrt{rd}/n)$. Then, conditioned on the relative error bound holding, we have
    \begin{equation*}
        \Dchi(\rho \| \widetilde{\brho}_{\chi^2}) \leq O(rd/n) + O( \sqrt{r d^3}/n).
    \end{equation*}
    Thus, with high constant probability, for $n = O(\sqrt{rd^3}/\eps)$, we have $\Dchi(\rho \| \widetilde{\brho}_{\chi^2}) \leq \epsilon$.
\end{proposition}

\begin{proof}

    The matrices $\widehat{\brho}$ and $\widetilde{\brho}_{\chi^2}$ share eigenvectors, which we label $\{\ket{\bv_i}\}_{i \in [d]}$ in decreasing order of eigenvalue. Let the corresponding eigenvalues be $\{ \widehat{\balpha}_i\}$ and $\{\widetilde{\balpha}_i\}$, respectively. Then
    \begin{equation*}
        \widetilde{\balpha}_i = \begin{cases}
            \widehat{\balpha}_i + \eta_1 & i \in \{1, \dots, r\} \\
            \eta_2 & i \in \{r+1, \dots, d\}.
        \end{cases}
    \end{equation*}
    Also, since $\widehat{\balpha}_i \geq -d/n$, and since $\eta_1 = 2d/n$, we have $\widetilde{\balpha}_i \geq d/n$ for $i \in [r]$. As one consequence, $\widetilde{\brho}_{\chi^2}$ is a positive operator, and $\Dchi(\rho \| \widetilde{\brho}_{\chi^2} )$ is well-defined.

    It is worth noting here that $\bPi$ is not required to be the projector onto the top $r$ eigenvalues of $\widetilde{\brho}_{\chi^2}$. Indeed, depending on the exact constant factors, $\eta_2$ can be greater than $\eta_1$, which may cause some of the eigenvalues in $\overline{\bPi}$ to be greater than those in $\bPi$.

    We start by writing
    \begin{align}
        \Dchi(\rho \| \widetilde{\brho}_{\chi^2} )  & = \sum_{i,j=1}^d \frac{2}{\widetilde{\balpha}_i + \widetilde{\balpha}_j} \cdot \ABs{\bra{\bv_i} (\rho - \widetilde{\brho}_{\chi^2}) \ket{\bv_j}}^2 \nonumber \\
        & = \sum_{i,j=1}^d \frac{2}{\widetilde{\balpha}_i + \widetilde{\balpha}_j} \cdot \ABs{\bra{\bv_i} (\rho - \bPi \widehat{\brho} \bPi - \eta_1 \cdot \bPi - \eta_2 \cdot \overline{\bPi}) \ket{\bv_j}}^2 \nonumber \\
        & \leq \sum_{i,j=1}^d \frac{6}{\widetilde{\balpha}_i + \widetilde{\balpha}_j} \cdot  \Big( \ABs{\bra{\bv_i} (\rho - \bPi \widehat{\brho} \bPi)\ket{\bv_j}}^2  + \eta_1^2 \cdot \ABs{\bra{\bv_i} \bPi\ket{\bv_j}}^2 + \eta_2^2 \cdot \ABs{\bra{\bv_i} \overline{\bPi} \ket{\bv_j}}^2 \Big). \label{eq:chi-squared_aux_1}
    \end{align}
    In the last step, we have used the inequality $\Abs{a+b+c}^2 \leq ( \Abs{a} + \Abs{b} + \Abs{c} )^2 \leq 3(\Abs{a}^2 + \Abs{b}^2 + \Abs{c}^2 )$. It is easy to bound two parts of this sum. First, since $\eta_1 = 2d/n$, 
    \begin{align}
        \sum_{i,j=1}^d \frac{6}{\widetilde{\balpha}_i + \widetilde{\balpha}_j} \cdot  \eta_1^2 \cdot \ABs{\bra{\bv_i} \bPi\ket{\bv_j}}^2  & = \sum_{i,j=1}^d \frac{6}{\widetilde{\balpha}_i + \widetilde{\balpha}_j} \cdot  \eta_1^2 \cdot \delta_{ij} \cdot 1 \{ i \leq r \} \nonumber \\
        & = 6\eta_1^2 \cdot \sum_{i=1}^r \frac{1}{2\cdot \widetilde{\balpha}_i} \leq 3 \eta_1^2 \cdot \sum_{i=1}^{r} \frac{1}{d/n} = \frac{12rd}{n} \leq O( rd/n). \label{eq:chi_bound_1}
    \end{align}
    Similarly, we have
    \begin{align}
        \sum_{i,j=1}^d \frac{6}{\widetilde{\balpha}_i + \widetilde{\balpha}_j} \cdot  \eta_2^2 \cdot \ABs{\bra{\bv_i} \overline{\bPi}\ket{\bv_j}}^2  & = \sum_{i,j=1}^d \frac{6}{\widetilde{\balpha}_i + \widetilde{\balpha}_j} \cdot  \eta_2^2 \cdot \delta_{ij} \cdot 1 \{ i > r \} \nonumber\\
        & = 6\eta_2^2 \cdot \sum_{i=r+1}^d \frac{1}{2\cdot \widetilde{\balpha}_i} = 3 \eta_2^2 \cdot \sum_{i=r+1}^{d
} \frac{1}{\eta_2} \leq 3 \eta_2 d \leq O(\sqrt{r d^3}/n).  \label{eq:chi_bound_2}
    \end{align}
    Here we have used that for $i \in \{r+1, \dots, d\}$, we have $\widetilde{\balpha}_i = \eta_2$. The remaining piece of \Cref{eq:chi-squared_aux_1} can be further broken up as follows:
    \begin{equation*}
         \sum_{i,j=1}^d \frac{6}{\widetilde{\balpha}_i + \widetilde{\balpha}_j} \cdot \ABs{\bra{\bv_i} (\rho - \bPi \widehat{\brho} \bPi)\ket{\bv_j}}^2 = \sum_{\substack{i,j=1 \\ \min(i,j) \leq r}}^d \frac{6}{\widetilde{\balpha}_i + \widetilde{\balpha}_j} \cdot \ABs{\bra{\bv_i} (\rho - \widehat{\brho}) \ket{\bv_j}}^2 + \sum_{i,j=r+1}^d \frac{6}{\widetilde{\balpha}_i + \widetilde{\balpha}_j} \cdot \ABs{\bra{\bv_i} \rho \ket{\bv_j}}^2.
    \end{equation*}
    We bound each of these pieces separately. First,
    \begin{align}
        \sum_{\substack{i,j=1 \\ \min(i,j) \leq r}}^d \frac{6}{\widetilde{\balpha}_i + \widetilde{\balpha}_j} \cdot \ABs{ \bra{\bv_i} (\rho - \widehat{\brho}) \ket{\bv_j} }^2 & = \sum_{m=1}^r \bigg( \sum_{\substack{i,j = 1 \\ \min(i,j) = m}}^d \frac{6}{\widetilde{\balpha}_i + \widetilde{\balpha}_j} \cdot \ABs{ \bra{\bv_i} (\rho - \widehat{\brho}) \ket{\bv_j} }^2 \bigg) \nonumber\\
        & \leq \sum_{m=1}^r \frac{6}{\widetilde{\balpha}_m} \cdot \bigg( \sum_{\substack{i,j = 1 \\ \min(i,j) = m}}^d \ABs{ \bra{\bv_i} (\rho - \widehat{\brho}) \ket{\bv_j} }^2 \bigg) \nonumber\\
        & \leq \sum_{m=1}^{r} \frac{6}{\widetilde{\balpha}_m} \cdot 2(C')^2 \cdot \frac{d}{n} \cdot \Big( \widehat{\balpha}_m + \frac{d}{n} \Big) \nonumber\\
        & = 12 (C')^2 \cdot \frac{d}{n} \cdot \sum_{m=1}^{r} \Big( \frac{\widehat{\balpha}_m}{\widetilde{\balpha}_m} + \frac{d/n}{\widetilde{\balpha}_m} \Big) \nonumber \\
        & \leq 24 (C')^2 \cdot \frac{rd}{n} \leq O(rd/n). \label{eq:chi_bound_3}
    \end{align}
    In the third step, we have used \Cref{lem:Bures_bound_technical_lemma}, and in the second-last step, we have used $\widetilde{\balpha}_m \geq \max(\widehat{\balpha}_m, d/n)$, since $m \in [r]$. Second,
    \begin{equation*}
        \sum_{i,j=r+1}^d \frac{6}{\widetilde{\balpha}_i + \widetilde{\balpha}_j} \cdot \ABs{ \bra{\bv_i} \rho \ket{\bv_j} }^2  = \frac{3}{\eta_2} \cdot \sum_{i,j=r+1}^d \ABs{ \bra{\bv_i} \rho \ket{\bv_j} }^2  = \frac{3}{\eta_2} \cdot \norm{ \overline{\bPi} \rho \overline{\bPi} }_2^2.
    \end{equation*}
    We now bound this norm. Because $\rho$ is rank $r$, we have
    \begin{equation*}
        \norm{ \overline{\bPi} \rho \overline{\bPi}}_2^2 \leq r \cdot \norm{ \overline{\bPi} \rho \overline{\bPi} }_\infty^2.
    \end{equation*}
    Moreover, for any $\ket{w} \in \mathrm{supp}(\overline{\bPi})$,
    \begin{align*}
        \bra{w} \rho \ket{w} & \leq \bra{w} \widehat{\brho} \ket{w} + C' \cdot \sqrt{ \frac{d}{n} \cdot \Big( \bra{w} \widehat{\brho} \ket{w} + \frac{d}{n} \Big)} \\
        & \leq \widehat{\balpha}_{r+1} + C' \cdot \sqrt{ \frac{d}{n} \cdot \Big( \widehat{\balpha}_{r+1} + \frac{d}{n} \Big)},
    \end{align*}
    using \Cref{cor:uniform-bound-pointwise_rho-hat-version}. However, by \Cref{cor:eigenvalue_differences}, we have
    \begin{equation*}
        \widehat{\balpha}_{r+1} \leq \alpha_{r+1} + C \cdot \sqrt{\frac{d}{n} \cdot \Big( \alpha_{r+1} + \frac{d}{n} \Big)} = C \cdot \frac{d}{n},
    \end{equation*}
    since $\alpha_{r+1} = 0$ as $\rho$ is rank $r$. Thus,
    \begin{equation*}
        \bra{w} \rho \ket{w} \leq C \cdot \frac{d}{n} + C' \cdot \sqrt{C+1} \cdot \frac{d}{n} \leq 3(C')^2 \cdot \frac{d}{n}.
    \end{equation*}
    for any $\ket{w} \in \mathrm{supp}(\overline{\bPi})$, using $1 \leq C < C'$. This yields $\norm{ \overline{\bPi} \rho \overline{\bPi} }_\infty \leq 3(C')^2 \cdot d/n$. Retracing our steps, we get
    \begin{equation} \label{eq:chi_bound_4}
        \sum_{i,j=r+1}^d \frac{6}{\widetilde{\balpha}_i + \widetilde{\balpha}_j} \cdot \ABs{ \bra{\bv_i} \rho \ket{\bv_j} }^2 \leq \frac{3r}{\eta_2} \cdot \norm{ \overline{\bPi} \rho \overline{\bPi} }_\infty^2  \leq \frac{27(C')^4 rd^2}{\eta_2 n^2} \leq  O(\sqrt{rd^3}/n).
    \end{equation}
    Here we have used $\eta_2 = \Theta(\sqrt{rd}/n)$.
    Plugging \Cref{eq:chi_bound_1,eq:chi_bound_2,eq:chi_bound_3,eq:chi_bound_4} all into \Cref{eq:chi-squared_aux_1} gives us
    \begin{equation*}
        \Dchi(\rho \| \widetilde{\brho}_{\chi^2} )   = \sum_{i,j=1}^d \frac{2}{\widetilde{\balpha}_i + \widetilde{\balpha}_j} \cdot \ABs{\bra{\bv_i} (\rho - \widetilde{\brho}_{\chi^2}) \ket{\bv_j} }^2 \leq O(rd/n) + O(\sqrt{rd^3}/n).
    \end{equation*}
    This completes the proof.
\end{proof}

\section*{Acknowledgments}
We would like to thank Ishaq Aden-Ali, Omar Alrabiah, and Xinyu (Norah) Tan for helpful conversations.
A.P.\ is supported by DARPA under Agreement No.\ HR00112020023.
J.S.\ and J.W.\ are supported by the NSF CAREER award CCF-233971.
E.T.\ is supported by the Miller Institute for Basic Research in Science, University of California, Berkeley.

We used AI tools throughout the course of this project to discuss ideas and prove low-level results.
The central ideas of this work are human-generated.

\bibliographystyle{alpha}
\bibliography{wright,other_references}

\end{document}